\newif\ifdraft\drafttrue
\documentclass{article} 
\usepackage[text={6.5in,9in},centering]{geometry}
\usepackage{amsmath,amssymb,latexsym,stmaryrd,bbm,eucal}
\usepackage{xcolor,enumerate,multirow,bussproofs}
\usepackage[colorlinks=true,citecolor=red]{hyperref}
\usepackage{tikz,tikz-3dplot}
\usetikzlibrary{chains,arrows,arrows.meta,automata,calc,shadings}
\usepackage{subcaption}
\usepackage{apxproof,cite}
\usepackage{authblk}

\numberwithin{equation}{section}
\renewcommand{\arraystretch}{1.15}
\newtheoremrep{theorem}{Theorem}[section]
\newtheoremrep{lemma}[theorem]{Lemma}
\newtheoremrep{corollary}[theorem]{Corollary}

\usepackage{tcolorbox,color}
\definecolor{cornellred}{RGB}{196,18,48}
\definecolor{dartmouthgreen}{RGB}{0,112,60}
\definecolor{mcgillred}{RGB}{237,27,47}
\definecolor{uclpurple}{RGB}{96,40,153}
\definecolor{ucllightblue}{RGB}{102,204,255}

\usepackage{wrapfig}

\newcommand\eid\id

\newcommand\Rto\Rrightarrow
\newcommand\src{\mathop{\mathsf{src}}\nolimits}
\newcommand\tgt{\mathop{\mathsf{tgt}}\nolimits}
\newcommand\TernaryInfC\TrinaryInfC %ugh
\renewcommand\dim{\mathop{\mathsf{dim}}\nolimits}
\newcommand\id{\mathop{\mathsf{id}}\nolimits}
\newcommand\crc[1]{\mathrel{\circ_{#1}}}
\newcommand\s[2]{^{(#1)}(#2)}
\newcommand\prefix{\mathop{\mathsf{P}}}
\newcommand\gto[1]{\stackrel{\mathrm{#1}}\longrightarrow}
\newcommand\Pic{\mathop{\mathsf{Pic}}\nolimits}
\newcommand\hgt[2]{h_{#1}(#2)}
\newcommand\refl{\mathop{\mathsf{refl}}\nolimits}
\newcommand\ident[1]{1_{[#1]}}
\newcommand\crct[2]{\circ_{#2}^{#1}}
\newcommand\op{^{\mathrm{op}}}
\newcommand\cmp{{\kern2pt;\kern2pt}}
\newcommand\med{\mathrel{\lozenge}}
\newcommand\Otimes{\textstyle\bigotimes}
\newcommand\Id[3]{\mathsf{Id}_{#1}(#2,#3)}
\newcommand{\Top}{\mathsf{Top}}
\newcommand\curry{\mathsf{curry}}
\newcommand\uncurry{\mathsf{uncurry}}
\newcommand\bdry[1]{\sim_{#1}}
\newcommand\rest\restriction
\newcommand\Pinh{P^{\mathrm{inh}}}
\newcommand\EEst{\EE^{st}}

\newcommand\extagqed{\tag*\qedhere}
\newcommand\exqed{}

\newcommand\seq[3]{#1_{#2},\ldots,#1_{#3}}
\newcommand\lam[2]{\lambda#1\kern1pt.\kern1pt#2}
\newcommand\rev[1]{#1\relax{}^-}
\newcommand\naturals{{\mathbb N}}
\newcommand\reals{{\mathbb R}}
\newcommand\rationals{{\mathbb Q}}
\newcommand\set[2]{\{#1 \mid #2\}}
\newcommand\longline{\begin{flushleft}\mbox{}\hrulefill\mbox{}\end{flushleft}}
\newcommand\EE{{\cal E}}
\newcommand\FF{{\cal F}}
\newcommand\subs{\mathrel{\subseteq}}
\newcommand\floor[1]{\lfloor #1 \rfloor}
\newcommand\ceil[1]{\lceil #1 \rceil}
\newcommand\eps{\varepsilon}
\newcommand\Set{\ensuremath{\mathsf{Set}}}
\newcommand\Hom{\mathop{\mathsf{Hom}}}
\newcommand\len[1]{|#1|}
\newcommand\Imp{\Rightarrow}

\newcommand{\theory}{\mathsf{SEq}}

\title{Strict Equality in Higher Identity Types:\\Axioms, Completeness, Decidability}
\author[1]{Keri D'Angelo}
\author[1,2]{Liron Cohen}
\author[1]{Luke Serafin}
\author[1]{Dexter Kozen}
\affil[1]{Cornell University}
\affil[2]{Ben-Gurion University}
\begin{document}

\maketitle

\begin{abstract}

Reasoning about higher identity types uses two equalities, strict equality and equivalence of cells witnessed by a path. Strict equality is a congruence for every operation on cells, but associativity and the identity laws fail for it. 
Equivalence of cells satisfies those two laws, but it is not a congruence. 
Representations of higher identity types often conflate  equational properties that hold strictly with those that hold only up to a path. We separate the two and axiomatize the former, introducing $\theory$, an equational theory generated by eleven axioms for strict equality between cells.
We prove soundness and completeness for homotopy models, and completeness already over finite CW-complexes. Nine of these axioms form a strongly normalizing rewrite system whose normal forms are unique up to picture equivalence, and the resulting word problem is decidable in exponential time. We also give presheaf models over a cubical base category in which composition arises from a pushout in the base category itself, without requiring Kan filling structure. The same pushout structure provides explicit higher-dimensional witnesses for associativity and the identity laws.

%\keywords{higher identity types, equational theory, completeness, decidability, cubical sets, rewriting, presheaf models}
\end{abstract}
 
\section{Introduction}

Recent years have seen the development of many complementary approaches to understanding the structure of higher types, especially higher identity types, in Martin-L\"of type theory and constructive mathematics. This includes topological approaches such as HoTT and univalent foundations \cite{HoTTBook}, categorical approaches such as strong and weak $\omega$-categories and their cubical counterparts
\cite{Futia04,AlAgl02}, type-theoretic approaches such as cubical type theory \cite{Bezem14,Cohen16,Licata18,Bentzen_2021}, and algebraic approaches such as computads \cite{Makkai05}. These approaches are all related on some basic level but may differ in their axiomatizations, choice of primitive concepts, and degree of constructivity.

A fundamental notion, and perhaps the primary motivation for much research in this area, is that of \emph{equality at a type}, where Martin-L\"of's higher identity types play a dominant role. 
A proof of equality is typically a \emph{path} in a space, serving as a witness for an equality judgment. This type of equality is often called \emph{propositional equality} and is denoted $\simeq$ in this paper. This is weaker than strict equality, often called \emph{judgmental equality} and here denoted $=$. 
Paths form spaces themselves, so one may speak of paths between paths, and so on, giving the hierarchy of higher identity types. Many of the structures above were designed to capture the properties of this hierarchy at an abstract level.

In this paper we focus on the model theory of higher identity types in their various incarnations as weak $\omega$-categories, homotopy models, and presheaf models over cubical sets. Our starting point is the observation that standard presentations tend to conflate desirable equational properties that come from strict equality ($=$) and those that come from equivalence as witnessed by a path ($\simeq$), whereas there are subtle but important reasons to treat them separately. 
For example, two $n$-cells of a higher identity type can be composed only when their
boundaries agree on the nose: the composite\footnote{For $\alpha$ and $\beta$ $n$-dimensional cells, $\alpha \crct mk \beta$ means the composition of $\alpha$ and $\beta$ along their $(n-k)$th $(m-1)$-dimensional faces; see section~\ref{sec:homotopycomp} for the definition.} $\alpha\crct n{n-k}\beta$ exists only
if $\src^k\alpha=\tgt^k\beta$, and no path can be asked to witness that equation,
because until it holds there is nothing for a path to connect.  The same is true of the globular laws $\src\tgt\alpha=\src^2\alpha$ and $\tgt\src\alpha=\tgt^2\alpha$.  

Standard presentations, however, do not isolate it and in fact mix laws that hold strictly with
laws that hold only up to equivalence witnessed by a path.  
Strict equality enjoys the advantage that all relevant algebraic operations (e.g., the globular maps $\src$ and $\tgt$, the identity map $\id$, and compositions $\crct nk$) are congruences. However, associativity of composition and the identity laws $\alpha=\alpha\crct n{n-k}\id^k\src^k\alpha$ and $\alpha=\id^k\tgt^k\alpha\crct n{n-k}\alpha$ do not hold strictly, but only up to $\simeq$. On the other hand, although $\simeq$ satisfies associativity and the identity laws, it fails to be a congruence with respect to $\id$. Write $\alpha\simeq\beta$ when some cell $\phi$ one dimension higher has $\alpha$ as source and $\beta$ as target. If $\id\alpha\simeq\id\beta$ via $\phi$, then $\alpha = \src\id\alpha = \src\src\phi = \src\tgt\phi = \src\id\beta = \beta$, using only $\src\id\alpha=\alpha$ and the globular law. So $\simeq$ can be a congruence for $\id$ only if it is equality.

Definitions of weak $\omega$-categories \cite{Batanin98,BaezDolan98,Bergner06,Hermida00,Hermida01,Hermida02,Futia04} and computads \cite{Makkai05} impose both
kinds of law at once.  
The division also  presents itself in the exchange rule that 
is usually presented as a law of $\simeq$, yet under the standard topological
definition of composition it holds strictly.  
Another example is $\id\alpha=\id\alpha\crct n{n-1}\id\alpha$.  
In the strict setting it is an immediate consequence of the identity laws
\cite{BrownHiggins81a,AlAgl02,Futia04}, and its type-theoretic shadow is the
judgmental equality $\refl a\cdot\refl a\equiv\refl a$ of
\cite[\S2.1]{HoTTBook}.  But once the identity laws are
weakened to hold only up to $\simeq$, it seems that it must be taken as a law in its own right.

These observations indicate that there is insight to be gained from a more careful separation of concerns. We propose to consider separately the properties of $=$ and $\simeq$ in models of higher identity types. In this paper, we focus on the former. 
We provide  $\theory$, an equational axiomatization of strict equality ($=$) consisting of eleven equational axioms and explore its consequences. The axiomatization consists of various well-known laws involving the globular maps $\src$ and $\tgt$, the identity map $\id$, and compositions $\crct nm$, including the globular laws and exchange. 
Thus our strict equational theory is quite different from computads \cite{Futia04,Makkai05}, which are assumed to satisfy associativity and identity laws for composition, and from strong $\omega$-categories~\cite{Street87}, in which additionally $\id$ is a congruence. Futia \cite{Futia04} takes a step in this direction and considers the two equational theories separately, but his axiomatization of the strict theory does not mention the equation $\id\alpha = \id\alpha\crct n{n-1}\id\alpha$ and treats the exchange law as a property of $\simeq$. The separation of the two layers also has a semantic counterpart: Swan \cite{Swan26} shows that path types and identity types can be separated in presheaf models of
univalent type theory.

Our main results are that $\theory$ is sound and complete for the strict equational theory of all homotopy models. More strongly, we show that it is sound for a (very wide) class of presheaf models over a certain cubical base category and complete for a (very narrow) class of topological models, namely the finite CW-complexes. We also show that the equational theory admits an exponential-time decision procedure.

That the laws of $\theory$ hold on the nose is a consequence of how composition is
obtained.  Cubical algebra has carried composition as primitive structure since
\cite{BrownHiggins81a,AlAgl02}, where each $\crct nm$ is part of the data.  
The presheaf models of cubical type theory do not. They differ in what the
interval object of the base category carries, a de Morgan algebra
in~\cite{Cohen16} and merely Cartesian structure in~\cite{Angiuli21}, but in
all of them composition of cells is absent from the base category: each type
must instead be equipped with Kan composition structure of its own.
Thus, composition is not primitive and does not respect globularity.  
We here recover it as structure on the base category itself (\S\ref{sec:presheaf}). 
Concretely, we propose a primitive composition mechanism in the form of a certain pushout on the base category of cubes. One can then define compositions natively without recourse to space-filling and in a way that respects the structure of higher identity types.

The construction provides a direct, constructive account of the weak laws. In standard presheaf models of cubical type theory, higher-dimensional witnesses are obtained through Kan filling, whereas we construct them from morphisms of the cubical base category. In particular, witnesses for associativity and identity laws arise from the same pushout structure defining composition, and coherence conditions such as the pentagon identity can be verified equationally. This complements our decidability result for strict equality: equality of terms in our equational theory can be decided algorithmically, while associativity and the identity laws are given explicit witnesses.

\paragraph{Main contributions:}
\begin{itemize}
\item
an equational presentation of the strict equational theory of higher identity types, including a construction of free models and their properties (\S\ref{sec:eqpres});
\item
completeness for homotopy models (\S\ref{sec:completeness});
\item
completeness for finite CW-complexes (\S\ref{sec:CW});
\item
an exponential-time decision procedure for strict equality (\S\ref{sec:decproc});
\item
a formulation of presheaf models with native composition and soundness over these models (\S\ref{sec:presheaf});
\item
a calculus for the morphisms of $\boxdot$, in which witnesses for associativity
and the identity laws are constructed directly (\S\ref{sec:connections}, \S\ref{apx:boxdot1}, and \S\ref{apx:boxdot}).
\item
a homeomorphism between higher identity types as iterated path spaces and as cubical sets (\S\ref{sec:homotopymodels} and \S\ref{apx:ident}).\end{itemize}

\section{Homotopy Models}
\label{sec:homotopymodels}

Our motivating semantic models are the homotopy models of higher identity types~\cite{AwodeyWarren09} as they arise in topological models of HoTT and cubical type theory. There are two equivalent views of these models, one as an iterated hierarchy of spaces of continuous paths generated inductively in a uniform way, the other as a class of homotopies $\alpha:I^n\to X$ satisfying a certain boundary-preservation condition, where $I$ is the closed real unit interval and $X$ is an arbitrary topological space. %
The two views are homeomorphic. Passing from the first to the second is currying: a path between paths is a map $I\to(I\to X)$, which uncurries to a map $I^2\to X$, and iterating gives $I^n\to X$. This is a homeomorphism and not
merely a bijection because $I$ is compact Hausdorff, so that $(-)\times I$ is left adjoint to $(-)^I$ in the compact-open topology. Under the correspondence, the endpoint conditions at each level, $p(0)=a$ and $p(1)=b$, become the
boundary-preservation condition on $I^n\to X$. Further details are in \S\ref{apx:ident}.

In the first view, one has an inductively constructed hierarchy of spaces
\begin{itemize}
\item Level 0: $a : X$
\item Level 1: $p : \Id Xab$, where $a,b : X$
\item Level 2: $\alpha : \Id{\Id Xab}pq$, where $p,q : \Id Xab$
\item Level 3: $\phi : \Id{\Id{\Id Xab}pq}\alpha\beta$, where $\alpha,\beta : \Id{\Id Xab}pq$,
\end{itemize}
etc., where the space of continuous functions $I\to X$ is endowed with the compact-open topology relative to $X$ and $\Id Xab$ denotes the subspace of all paths $p:I\to X$ such that $p(0)=a$ and $p(1)=b$.

In the second view, we have a family of \emph{$n$-cells} for each $n\ge 0$, which are continuous maps $\alpha:I^n\to X$ satisfying a certain \emph{boundary preservation condition}. These will be models of the axioms to be presented later in \S\ref{sec:axioms}.

An \emph{$n$-dimensional homotopy} is a continuous map $\alpha:I^n\to X$. Not all homotopies are $n$-cells, but only those satisfying the boundary-preservation condition. We define $\dim\alpha=n$. The set of $n$-dimensional homotopies is denoted $H_n$.

Given an $n$-dimensional homotopy $\alpha:I^n\to X$ and $1\le k\le n$, let $\alpha^{(k)}$ be $\alpha$ curried on its $k$-th component:
\begin{align*}
& \alpha^{(k)}:I\to I^{n-1}\to X \qquad 
& \alpha\s k{x_k}(\seq x1{k-1},\seq x{k+1}n) = \alpha(\seq x1n).
\end{align*}
We always think of the first component as time and often write $t$ for $x_1$, so
\begin{align*}
& \alpha\s 1t(\seq x2n) = \alpha(t,\seq x2n).
\end{align*}
For $\alpha\in H_n$, $x\in I$, and $1\le k\le n$, $\alpha\s kx$ is called a \emph{section} of $\alpha$. A \emph{face} of $\alpha$ is a section $\alpha\s kb$ with $b\in\{0,1\}$. Two special faces are $\src\alpha=\alpha\s 10$ and $\tgt\alpha=\alpha\s 11$, called the \emph{source} and \emph{target} of $\alpha$, respectively. We thus view $\alpha$ as a continuous deformation of $\src\alpha$ to $\tgt\alpha$, passing through intermediate states $\alpha\s 1t:I^{n-1}\to X$. For $n\ge 2$, the maps $\alpha\s 20$ and $\alpha\s 21$ describe the trajectory of $\src\src\alpha$ and $\tgt\src\alpha$ throughout the deformation.

Lower-dimensional sections and faces can be defined inductively. For $1\le i_1<i_2<\cdots<i_k\le n$ and $\alpha\in H_n$, let $\alpha^{(\seq i1k)}:I^k\to I^{n-k}\to X$ be $\alpha$ curried on components $\seq i1k$:
$\alpha\s{\seq i1k}{\seq x1k} = \alpha\s{i_k}{x_k}\s{\seq i1{k-1}}{\seq x1{k-1}}$. 
The function $\alpha\s{\seq i1k}{\seq x1k}:I^{n-k}\to X$ is an $n-k$-dimensional section and represents $\alpha$ partially evaluated at values $\seq x1k$ for the variables at positions $\seq i1k$, respectively. For example,
$\alpha\s{1,3}{a,b}(x,y,z) = \alpha(a,x,b,y,z)$.
An $\alpha\in H_n$ is completely determined by its sections $\alpha\s{\seq i1k}{\seq x1k}$ for any fixed selection of indices $1\le i_1 < \cdots < i_k\le n$. One can define $\alpha$ by specifying these sections and verifying that the corresponding uncurried version is continuous. An essential tool in this regard is Lemma \ref{lem:compactopen}.

Regarding notation: When using prefix and postfix unary operators together, the postfix takes precedence. For example, $\src\alpha\s1t$ should be read as $\src(\alpha\s1t)$. To obtain $(\src\alpha)\s1t$, parentheses must be used. As before, binary operators have least precedence; thus $\src\alpha\crct nk\tgt\alpha$ should be read as $(\src\alpha)\crct nk(\tgt\alpha)$.

\subsection{Identities and Composition}
\label{sec:homotopycomp}

If $\alpha\in H_n$, define $\id\alpha$ by $(\id\alpha)\s 1t=\alpha$; that is, $(\id\alpha)^{(1)}$ is the constant function $\lam t\alpha$. Then $\id\alpha\in H_{n+1}$, as it is the composition of $\alpha:I^n\to X$ with the continuous projection $\lam{t,\bar x}{\bar x}:I^{n+1}\to I^n$. An \emph{identity} is any homotopy of the form $\id\alpha$; that is, $\beta$ is an identity iff the value of $\beta\s1t$ is independent of $t$. Note that $\beta$ is an identity iff $\beta = \id\src\beta$, and also iff $\beta = \id\tgt\beta$.

If $\alpha,\beta\in H_n$, $1\le k\le n$, and $\alpha\s k1=\beta\s k0$, we can compose them along their $k$-th $(n-1)$-dimensional faces:
{\renewcommand\arraystretch{1}
\begin{align}
(\beta\crct n{n-k}\alpha)\s kx &= \begin{cases}
\alpha\s k{2x}, & \text{if $0\le x\le 1/2$}\\
\beta\s k{2x-1}, & \text{if $1/2\le x\le 1$.}
\end{cases}
\label{eq:kacomp}
\end{align}}
There is no discontinuity at $x=1/2$, because $\alpha\s k1=\beta\s k0$. For $k=1$ this is the familiar concatenation, defined whenever $\tgt\alpha=\src\beta$.

Routine expansion of these definitions shows that $\src$ and $\tgt$ are left inverses of $\id$ (see \eqref{eq:srcofid}), that the source and target of a composite are determined by those of its arguments (see \eqref{eq:srcofcomp1}, \eqref{eq:srcofcompk} and \eqref{eq:srckofcomp}), and that $\id$ distributes over composition (see \eqref{eq:idofcomp}). The calculations are carried out in \S\ref{apx:soundness}.

Viewing $\alpha\in H_n$ for $n\ge 1$ as a deformation from $\src\alpha$ to $\tgt\alpha$, there is an $\rev\alpha\in H_n$ giving a deformation in the reverse direction, namely $(\rev\alpha)\s 1t = \alpha\s 1{1-t}$. This is an involution in the sense that $\rev{(\rev\alpha)} = \alpha$.

\subsection{Cells}
\label{sec:cell}

As mentioned, a \emph{cell} is a homotopy satisfying a certain boundary-preservation condition. All $0$- and $1$-dimensional homotopies are cells by definition. For $n\ge 2$, an \emph{$n$-cell} is an $\alpha\in H_n$ such that (a) for all $t\in I$, $\alpha\s 1t$ is (inductively) an $(n-1)$-cell, and (b) for $b\in\{0,1\}$, $\alpha\s 2b$ is an identity. Condition (b) says that $\alpha\s 2b\s1t$ is independent of $t$, equivalently $\alpha\s 2b = \id\src\alpha\s 2b$; for $b=0$ it is equivalent to $\alpha\s 20 = \id\src^2\alpha$, since $\src\alpha\s 20 = \alpha\s10\s 20 = \alpha\s10\s 10 = \src^2\alpha$.

One thinks of a cell as an open set whose boundary comprises lower-dimensional cells. 
Fig.~\ref{fig:cell-diagram} %in the appendix 
illustrates a cell $\alpha:I^3\to X$ of dimension three. 
%\begin{toappendix}
\newcommand\radius{1.5}
\tdplotsetmaincoords{65}{125}
\begin{figure}
\begin{align*}
&
\begin{tikzpicture}[-, auto, scale = 0.6, inner sep=0pt, minimum size=0pt]
\small
%\draw[help lines] (-4.4,-1.2) grid (8.7,4.9);
\useasboundingbox (-4.4,-1.2) rectangle (8.7,4.9);
\node (000) at (0,0) {};
\node (001) at (0,2) {};
\node (010) at (2,0) {};
\node (011) at (2,2) {};
\node (100) at (2.5,1.5) {};
\node (101) at (2.5,3.3) {};
\node (110) at (4.3,1.5) {};
\node (111) at (4.3,3.3) {};
\draw (000) -- (100);
\draw (001) -- (101);
\draw (010) -- (110);
\draw (011) -- (111);
\draw (000) -- (001);
\draw (000) -- (010);
\draw (011) -- (010);
\draw (011) -- (001);
\draw (100) -- (101);
\draw (100) -- (110);
\draw (111) -- (110);
\draw (111) -- (101);
\fill[red!60, opacity=.4] (2.5,1.5) -- (2.5,3.3) -- (4.3,3.3) -- (4.3,1.5);
\fill[red!60, opacity=.4] (2.5,1.5) -- (2.5,3.3) -- (0,2) -- (0,0);
\fill[red!80, opacity=.45] (0,0) -- (2.5,1.5) -- (4.3,1.5) -- (2,0);
\fill[red!60, opacity=.4] (0,0) -- (0,2) -- (2,2) -- (2,0);
\fill[red!40, opacity=.4] (0,2) -- (2.5,3.3) -- (4.3,3.3) -- (2,2);
\fill[red!60, opacity=.4] (2,0) -- (2,2) -- (4.3,3.3) -- (4.3,1.5);
\path (-1,1.7) edge[->] (-.3,1.7);
\node[anchor=east] at (-1.1,1.7) {$\alpha\s20 = \id\src^2\alpha$};
\path (5.3,1.7) edge[->] (4.6,1.7);
\node[anchor=west] at (5.4,1.7) {$\alpha\s21 = \id\tgt^2\alpha$};
\path (2.5,-.7) edge[->] (2.5,0);
\node[anchor=north] at (2.5,-.8) {$\alpha\s30 = \id^2\src^3\alpha$};
\path (2.5,4.3) edge[->] (2.5,3.6);
\node[anchor=south] at (2.5,4.4) {$\alpha\s31 = \id^2\tgt^3\alpha$};
\path (-1,-.1) edge[->] (-.3,.3);
\node[anchor=east] at (-1.3,-.3) {$\alpha\s10 = \src\alpha$};
\path (5.3,3.4) edge[->] (4.6,3);
\node[anchor=west] at (5.5,3.6) {$\alpha\s11 = \tgt\alpha$};
\path (000) edge[->, >=stealth', auto, inner sep=2pt, thick] node {$x_3$} (001);
\path (000) edge[->, >=stealth', auto, inner sep=2pt, thick] node[swap] {$x_2$} (010);
\path (000) edge[->, >=stealth', auto, inner sep=2pt, thick, pos=.6] node {$x_1$} (100);
\end{tikzpicture}
&\hspace{5mm}&
\begin{tikzpicture}[-, auto, scale = 0.7, inner sep=0pt, minimum size=0pt]
%\draw[help lines] (-1.5,-2.2) grid (1.5,2);
\useasboundingbox (-1.5,-2.2) rectangle (1.5,2);
 \coordinate (O) (0,0);
 \draw[ball color=cyan!50!blue, opacity=0.4, draw opacity=0] (O) circle[radius={\radius}];
 \begin{scope}[tdplot_main_coords]
  \begin{scope}[canvas is xy plane at z={\radius}] 
  \fill (0,0) circle[radius={3pt}];
 \end{scope}
 \begin{scope}[canvas is xy plane at z={-\radius}] 
  \fill (0,0) circle[radius={3pt}];
 \end{scope}
 \begin{scope}[canvas is yz plane at x=0]
  \draw[dashed] (\tdplotmainphi:\radius) 
  arc[start angle={\tdplotmainphi}, 
  end angle={\tdplotmainphi+180}, 
  radius={\radius}];
  \draw[thick] (\tdplotmainphi:\radius) 
  arc[start angle={\tdplotmainphi}, 
  end angle={\tdplotmainphi-180}, 
  radius={\radius}];
 \end{scope}
 \end{scope}
\end{tikzpicture}
\end{align*}
\caption{A 3-cell $\alpha:I^3\to X$ in cubical and globular representations.}
\label{fig:cell-diagram}
\end{figure}
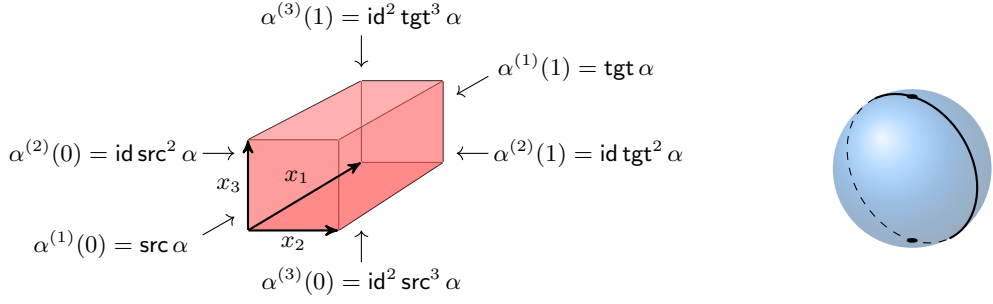
%\end{toappendix}

In the rectilinear diagram, each point of $I^3$ maps to a point of $X$, and the source and target of $\alpha$ are the front and back faces. Along the left face every vertical line has the same image as $\src^2\alpha$, so that
two-dimensional face represents a one-dimensional line in $X$; likewise $\alpha$ is constant on the top face, which therefore represents a single point. The map $\alpha:I^3\to X$ thus factors through the spherical diagram beside it. The
rectilinear representation comes from cubical type theory and it is more convenient for formalizing compositions and is in general more amenable to constructions.
Further properties of these models, including the soundness of the axioms of \S\ref{sec:axioms}, are developed in \S\ref{apx:soundness}.

\begin{toappendix}
\section{Soundness over Homotopy Models}
\label{apx:soundness}

We record here the calculations underlying the equations cited in \S\ref{sec:homotopycomp} and prove soundness of the axioms of \S\ref{sec:eqpres}.
 
Both $\src$ and $\tgt$ are left inverses of $\id$:
\begin{align}
\src\id\alpha &= (\id\alpha)\s10 = \alpha &
\tgt\id\alpha &= (\id\alpha)\s11 = \alpha.\label{eq:srcofid}
\end{align}
For $\beta\crct n{n-1}\alpha$, evaluating \eqref{eq:kacomp} at $k=1$ and $x\in\{0,1\}$ gives
\begin{align}
\src(\beta\crct n{n-1}\alpha) &= (\beta\crct n{n-1}\alpha)\s 10 = \alpha\s 10 = \src\alpha,\nonumber\\
\tgt(\beta\crct n{n-1}\alpha) &= (\beta\crct n{n-1}\alpha)\s 11 = \beta\s 11 = \tgt\beta.\label{eq:srcofcomp1}
\end{align}
For $k>1$,
\begin{align*}
(\src\alpha)\s{k-1}1 &= \alpha\s 10\s{k-1}1 = \alpha\s k1\s 10\\
&= \beta\s k0\s 10 = \beta\s 10\s{k-1}0\\
&= (\src\beta)\s{k-1}0,
\end{align*}
so $\src\beta\crct n{n-k}\src\alpha$ exists, and
\begin{align}
\src(\beta\crct n{n-k}\alpha) &= (\beta\crct n{n-k}\alpha)\s10\\
&= \begin{cases}
\alpha\s k{2x}\s10, & \text{if $0\le x\le 1/2$}\nonumber\\
\beta\s k{2x-1}\s10, & \text{if $1/2\le x\le 1$}
\end{cases}\\
&= \begin{cases}
\alpha\s10\s{k-1}{2x}, & \text{if $0\le x\le 1/2$}\nonumber\\
\beta\s10\s{k-1}{2x-1}, & \text{if $1/2\le x\le 1$}
\end{cases}\\
&= \begin{cases}
(\src\alpha)\s{k-1}{2x}, & \text{if $0\le x\le 1/2$}\nonumber\\
(\src\beta)\s{k-1}{2x-1}, & \text{if $1/2\le x\le 1$}
\end{cases}\\
&= \src\beta\crct n{n-k}\src\alpha.\label{eq:srcofcompk}
\end{align}
Similarly, $\tgt(\beta\crct n{n-k}\alpha) = \tgt\beta\crct n{n-k}\tgt\alpha$.
It follows inductively that
\begin{align}
\src^k(\beta\crct n{n-k}\alpha) &= \src^k\alpha & \tgt^k(\beta\crct n{n-k}\alpha) &= \tgt^k\beta.\label{eq:srckofcomp}
\end{align}

Also, for $n$-cells $\alpha,\beta$ and $1\le k\le n$,
\begin{align*}
& (\id\beta\crct n{n-k}\id\alpha)\s1t\s ks\\
&= (\id\beta\crct n{n-k}\id\alpha)\s{k+1}s\s1t\\
&= \begin{cases}
(\id\alpha)\s{k+1}{2s}\s1t & \text{if $0\le s\le 1/2$}\\
(\id\beta)\s{k+1}{2s-1}\s1t & \text{if $1/2\le s\le 1$}
\end{cases}\\
&= \begin{cases}
(\id\alpha)\s1t\s{k}{2s} & \text{if $0\le s\le 1/2$}\\
(\id\beta)\s1t\s{k}{2s-1} & \text{if $1/2\le s\le 1$}
\end{cases}\\
&= \begin{cases}
\alpha\s{k}{2s} & \text{if $0\le s\le 1/2$}\\
\beta\s{k}{2s-1} & \text{if $1/2\le s\le 1$}
\end{cases}\\
&= (\beta\crct n{n-k}\alpha)\s ks,
\end{align*}
thus $(\id\beta\crct n{n-k}\id\alpha)\s1t = \beta\crct n{n-k}\alpha$. As this is independent of $t$,
\begin{align}
\id(\beta\crct n{n-k}\alpha) = \id\beta\crct n{n-k}\id\alpha.\label{eq:idofcomp}
\end{align}

\begin{lemma}
\label{lem:cellprops}
Let $\alpha\in H_n$. The following are equivalent:
\begin{enumerate}[{\upshape(i)}]
\item $\alpha$ is a cell;
\item for all $2\le k\le n$, for all $\seq x1{k-1}\in I^{k-1}$ and $b\in\{0,1\}$,
\begin{align}
\alpha\s{1,2,\ldots,k}{\seq x1{k-1},b}
= \alpha\s{1,2,\ldots,k}{b,\ldots,b,b};\label{eq:lotsofbs}
\end{align}
\item for all $1\le k\le n$, $\alpha\s k0 = \id^{k-1}\src^k\alpha$ and $\alpha\s k1 = \id^{k-1}\tgt^k\alpha$;
\item for all $1\le m\le n-1$ and $b\in\{0,1\}$, the value of $\alpha\s{m,m+1}{x,b}$ is independent of $x$.
\end{enumerate}
\end{lemma}
\begin{proof}
First we observe that if $\alpha$ is a cell, then so are $\src\alpha$, $\tgt\alpha$, and $\id\alpha$. The first two are by definition of cells. For $\id\alpha$, we have $(\id\alpha)\s1t=\alpha$, so the first condition is satisfied, and $(\id\alpha)\s2b\s1t = (\id\alpha)\s1t\s1b = \alpha\s1b$, which is independent of $t$, so the second condition is satisfied.

(i) $\Imp$ (ii): The argument is by induction on dimension.
\begin{align*}
& \alpha\s{1,\ldots,k}{\seq x1{k-1},b}\\
&= \alpha\s1{x_1}\s{1,\ldots,k-1}{\seq x2{k-1},b}\\
&= \alpha\s1{x_1}\s{1,\ldots,k-1}{b,\ldots,b,b} && \text{induction hypothesis}\\
&= \alpha\s2b\s1{x_1}\s{1,\ldots,k-2}{b,\ldots,b}\\
&= \alpha\s2b\s1{b}\s{1,\ldots,k-2}{b,\ldots,b} && \text{since $\alpha\s2b$ is an identity}\\
&= \alpha\s{1,\ldots,k}{b,\ldots,b}.
\end{align*}

(ii) $\Imp$ (iii): By induction on $k$.

The base case $k=1$ is by definition of $\src$. For $k\ge 2$,
\begin{align*}
\alpha\s k0\s 1t
&= \alpha\s{1,k}{t,0}
= \alpha\s{1,k}{0,0}\\
&= \alpha\s 10\s{k-1}0
= (\src\alpha)\s{k-1}0
\end{align*}
with the second inference from \eqref{eq:lotsofbs}.
By the induction hypothesis, this is
\begin{align*}
\id^{k-2}\src^{k-1}\src\alpha = \id^{k-2}\src^k\alpha. 
\end{align*}
As this is independent of $t$, $\alpha\s k0 = \id^{k-1}\src^k\alpha$. The argument for $\alpha\s k1 = \id^{k-1}\tgt^k\alpha$ is symmetric.

(iii) $\Imp$ (iv): For all $1\le m\le n-1$, for all $\seq x1m,x$,
\begin{align*}
& \alpha\s{m,m+1}{x,0}\s{1,2,\ldots,m-1}{\seq x1{m-1}}\\
&= \alpha\s{m+1}{0}\s{1,2,\ldots,m}{\seq x1{m-1},x}\\
&= (\id^m\src^{m+1}\alpha)\s{1,2,\ldots,m}{\seq x1{m-1},x}
= \src^{m+1}\alpha,
\end{align*}
which is independent of $x$, therefore so is $\alpha\s{m,m+1}{x,0}$. Symmetrically, the same holds for $\alpha\s{m,m+1}{x,1}$.

(iv) $\Imp$ (i): By induction on dimension.
For $\alpha$ satisfying (iv), $1\le m\le n-2$, $t,x\in I$, and $b\in\{0,1\}$,
\begin{align*}
\alpha\s1t\s{m,m+1}{x,b} &= \alpha\s{m+1,m+2}{x,b}\s1t\\
&= \alpha\s{m+1,m+2}{0,b}\s1t = \alpha\s1t\s{m,m+1}{0,b},
\end{align*}
so $\alpha\s1t$ satisfies (iii). By the induction hypothesis, $\alpha\s1t$ is a cell. Moreover, $\alpha\s{1,2}{t,b}$ is independent of $t$ by assumption, thus $\alpha\s2b$ is an identity. Therefore $\alpha$ is a cell.
\exqed
\end{proof}

It follows from the characterization of Lemma \ref{lem:cellprops}(iii) that
\begin{align}
\begin{array}c
\src\tgt\alpha = \alpha\s11\s10 = \alpha\s20\s11 = \tgt\id\src^2\alpha = \src^2\alpha\\[1ex]
\tgt\src\alpha = \alpha\s10\s11 = \alpha\s21\s10 = \src\id\tgt^2\alpha = \tgt^2\alpha.
\end{array}
\label{eq:srcoftgt}
\end{align}

\begin{lemma}
\label{lem:cellsection}
All sections of cells are cells.
\end{lemma}
\begin{proof}
By induction on dimension. Any section of dimension zero or one is a cell by definition. For $n\ge 2$ and $1\le i_1 < \cdots < i_k \le n$,
\begin{align*}
\alpha\s{\seq i1k}{\seq x{i_1}{i_k}}\s 1x
&= \alpha\s{\seq j1{k+1}}{\seq x{j_1}{j_{k+1}}}
\end{align*}
after reindexing, which is a cell by the induction hypothesis, thus $\alpha\s{\seq i1k}{\seq x{i_1}{i_k}}$ satisfies the first condition of cellhood. For the second condition, we must show that $\alpha\s{\seq i1k}{\seq x{i_1}{i_k}}\s 20$ is an identity, provided $n-k\ge 2$ (otherwise, there is nothing to prove). By \eqref{eq:lotsofbs},
\begin{align*}
\alpha\s{\seq i1k}{\seq x{i_1}{i_k}}\s 2b\s 1x
&= \alpha\s{\seq i1k}{\seq x{i_1}{i_k}}\s{1,2}{x,b}\\
&= \alpha\s{\seq i1k}{\seq x{i_1}{i_k}}\s{1,2}{b,b}.
\end{align*}
As this is independent of $x$, $\alpha\s{\seq i1k}{\seq x{i_1}{i_k}}\s 2b$ is an identity.
\exqed
\end{proof}

\subsection{Soundness}

All of the equations of \S\ref{sec:axioms} hold in homotopy models.
Axioms (i) and (ii) were established in \eqref{eq:srcofcompk},
axioms (iii) and (iv) in \eqref{eq:srcofcomp1},
axiom (v) in \eqref{eq:idofcomp},
axioms (vi) and (vii) in \eqref{eq:srcofid}, and
axioms (viii) and (ix) in \eqref{eq:srcoftgt}.
This leaves only axiom (x) and the exchange law axiom (xi), which we argue now.

\begin{lemma}[Axiom (x)]
\label{lem:axiomx}
For all $\alpha$ of dimension $n-1$, $\id\alpha\crct n{n-1}\id\alpha = \id\alpha$.
\end{lemma}
\begin{proof}
The term on the left-hand side is always meaningful, as $\src\id\alpha = \alpha = \tgt\id\alpha$. By definition of $\crct n{n-1}$,
\begin{align*}
(\id\alpha\crct n{n-1}\id\alpha)\s1t
&= \begin{cases}
(\id\alpha)\s1t & \text{if $0\le t\le 1/2$}\\
(\id\alpha)\s1t & \text{if $1/2\le t\le 1$}
\end{cases}\\
&= \begin{cases}
\alpha & \text{if $0\le t\le 1/2$}\\
\alpha & \text{if $1/2\le t\le 1$}
\end{cases}\\
&= \alpha = (\id\alpha)\s1t.
\extagqed
\end{align*}
\end{proof}

\begin{lemma}[Exchange]
\label{lem:exchange}
For all $1\le k < m \le n$ and for all $n$-cells $\alpha,\beta,\gamma,\delta$ satisfying the requisite boundary conditions,
\begin{align*}
(\alpha\crct n{n-k}\beta)\crct n{n-m}(\gamma\crct n{n-k}\delta) = (\alpha\crct n{n-m}\gamma)\crct n{n-k}(\beta\crct n{n-m}\delta).
\end{align*}
\end{lemma}
\begin{proof}
Expanding the definitions, both expressions give
\begin{align*}
& ((\alpha\crct n{n-k}\beta)\crct n{n-m}(\gamma\crct n{n-k}\delta))\s{k,m}{x,y}\\
&= ((\alpha\crct n{n-m}\gamma)\crct n{n-k}(\beta\crct n{n-m}\delta))\s{k,m}{x,y}\\
&= \begin{cases}
\alpha\s{k,m}{2x,2y}, & \text{if $0\le x\le 1/2$ and $0\le y\le 1/2$}\\
\beta\s{k,m}{2x-1,2y}, & \text{if $1/2\le x\le 1$ and $0\le y\le 1/2$}\\
\gamma\s{k,m}{2x,2y-1}, & \text{if $0\le x\le 1/2$ and $1/2\le y\le 1$}\\
\delta\s{k,m}{2x-1,2y-1}, & \text{if $1/2\le x\le 1$ and $1/2\le y\le 1$.}
\end{cases}\\[-5.5mm]
\extagqed
\end{align*}
\end{proof}

For cells $\alpha,\beta$ of dimension $n$, we might say $\alpha$ and $\beta$ are \emph{path equivalent} and write $\alpha\simeq\beta$ if there exists a cell $\phi$ of dimension $n+1$ such that $\alpha=\src\phi$ and $\beta=\tgt\phi$. One can show that $\simeq$ is an equivalence relation and a congruence with respect to the operations $\src$, $\tgt$, and $\crct nk$, but not $\id$. Indeed, if $\id\alpha \simeq \id\beta$ as witnessed by $\phi$, then
\begin{align}
\alpha = \src\id\alpha = \src\src\phi = \src\tgt\phi = \src\id\beta = \beta,\label{eq:idnotcong}
\end{align}
so even if $\alpha\simeq\beta$, it is not the case that $\id\alpha\simeq\id\beta$ unless $\alpha=\beta$.
\end{toappendix}

\section{Equational Presentation}
\label{sec:eqpres}

This section provides an equational presentation that delineates a wide class of algebraic models, among them the homotopy models of \S\ref{sec:homotopymodels}. In \S\ref{sec:completeness} we show that the axioms are complete for those models.

\subsection{Syntax}

Let $(X,\dim)$ be a \emph{stratified set} with $\dim:X\to\naturals$. The value $\dim x$ is the \emph{dimension} of $x$. A term over atomic symbols $X$, unary operators $\src$, $\tgt$, $\id$, and binary operators $\crct nm$ for $0\le m<n$ is \emph{well-formed} (and has a dimension) if these facts can be inferred from the following rules:
\begin{itemize}
\item
every $x\in X$ is well-formed, and its dimension is $\dim x$
\item
if $t$ is well-formed and has dimension $\dim t$, then $\id t$ is well-formed and $\dim\id t=\dim t+1$
\item
if $t$ is well-formed and $\dim t\ge 1$, then $\src t$ and $\tgt t$ are well-formed and $\dim\src t=\dim\tgt t=\dim t-1$
\item
if $u$ and $v$ are well-formed and $\dim u = \dim v = n \ge 1$, and if $0\le m\le n-1$, then $u \crct nm v$ is well-formed and $\dim(u \crct nm v)=n$.
\end{itemize}
Let $T$ denote the set of well-formed terms, let $X_n = \set{x\in X}{\dim x=n}$, and let $T_n$ denote the set of well-formed terms over $\bigcup_{m\le n}X_m$. Not all well-formed terms are meaningful; see \S\ref{sec:semantics} below.

\subsection{Axioms for Equality}
\label{sec:axioms}

In addition to the usual axioms for equality and congruence with respect to the operations $\id$, $\src$, $\tgt$, and $\crct nm$, we have the axioms of Fig.~\ref{fig:axioms}. We write $\theory$ for the resulting equational theory.
\begin{figure}[t]
\begin{enumerate}[{\upshape(i)}]
\item
$\src(u\crct n{n-k} v) = \src u\crct {n-1}{n-k}\src v,\ 2\le k\le n$
\item
$\tgt(u\crct n{n-k} v) = \tgt u\crct {n-1}{n-k}\tgt v,\ 2\le k\le n$
\item
$\src(u\crct n{n-1} v) = \src v$
\item
$\tgt(u\crct n{n-1} v) = \tgt u$
\item
$\id(u\crct n{n-k}v) = \id u\crct {n+1}{n-k}\id v$
\item
$\src\id u = u$
\item
$\tgt\id u = u$
\item
$\src\tgt u = \src\src u$
\item
$\tgt\src u = \tgt\tgt u$
\item
$\id u\crct n{n-1}\id u = \id u$
\item
$(u \crct n{n-k} v) \crct n{n-m} (x \crct n{n-k} y) = (u \crct n{n-m} x) \crct n{n-k} (v \crct n{n-m} y),\ 1\le k<m\le n$
\end{enumerate}
\caption{The axioms of $\theory$}
\longline
\label{fig:axioms}
\end{figure}
Axiom (xi) is called the \emph{exchange rule} and axioms (viii) and (ix) are called the \emph{globular rules}.

\subsection{Algebraic Semantics}
\label{sec:semantics}

Because we have postulated an equational system, there is a ready-made algebraic semantics available from equational logic and universal algebra \cite{Gratzer68} that can be stated without reference to topological models or higher categories.

As mentioned, not all well-formed terms are meaningful in all models. 
Let $\EE$ be a dimension-respecting equational theory containing $\theory$. 
A well-formed term is \emph{meaningful} (under $\EE$) if that fact can be deduced from the following rules:
\begin{itemize}
\item
A well-formed term $u\crct nm v$ is meaningful iff $u$ and $v$ are meaningful and $\EE\vdash\src^{n-m} u=\tgt^{n-m} v$, where $n = \dim u = \dim v = \dim u \crct nm v$. For $m=n-1$, this says $\EE\vdash\src u=\tgt v$.
\item
Well-formed terms $\src u$, $\tgt u$, and $\id u$ are meaningful if $u$ is.
\item
Any $x\in X$ is meaningful.
\end{itemize}

Given $X$ and an equational theory $\EE$ over $X$, there is a quotient algebra $T(X)/\EE$ consisting of the set of meaningful terms over $X$ modulo the equations $\EE$. This is a model of the axioms and of $\EE$, and is the free such model.

If $X\subs X'$ and $\EE\subs\EE'$, then there is a canonical homomorphism $T(X)/\EE\to T(X')/\EE'$ extending the inclusion map $X\hookrightarrow X'$. However, because there may be terms over $X$ that are meaningful under $\EE'$ but not under $\EE$, it does not necessarily have an epi-mono factorization, as the image of $T(X)/\EE$ in $T(X')/\EE'$ may not be closed under the composition operations.

These models are called \emph{globular models} because of the maps $\src$ and $\tgt$ on $\EE$-congruence classes, which go from dimension $n$ to dimension $n-1$ and satisfy the globular rules (viii) and (ix). The underlying globular set is the set of $\EE$-congruence classes along with these maps. 

\subsection{Free Globular Models}
\label{sec:globdef}

The \emph{free globular models} are formed from a sequence of equational theories $\EE_n$ defined with respect to a given pair of \emph{globular maps} $s,t:X\to T$ satisfying certain properties listed below. In addition to the axioms listed in Fig.~\ref{fig:axioms} of \S\ref{sec:axioms}, we take the additional equations
\begin{enumerate}[(i)]
\setcounter{enumi}{11}
\item
$\src x = s(x),\ \dim x\ge 1$
\item
$\tgt x = t(x),\ \dim x\ge 1$.
\end{enumerate}
These are not really additional axioms, but just a way to link $\src$ and $\tgt$ to the globular maps $s,t$ in the model specification when interpreting over a particular model.
The theory $\EEst$ is the equational theory on $T$ generated by (i)-(xiii), and $\EEst_n$ is the equational theory on $T_n$ generated by instances of (i)-(xiii) restricted to $T_n$. We require that the globular functions $s,t$ satisfy the following properties: For $x\in X_n$,
\begin{itemize}
\item
$s(x),t(x)\in T_{n-1}$
\item
$\dim s(x)=\dim t(x)=n-1$
\item
$\EEst_{n}\vdash\src t(x)= \src s(x)$ and $\EEst_{n}\vdash\tgt s(x) = \tgt t(x)$.
\end{itemize}
We are interested in the sequence of models $T_n/\EEst_n$ and their colimit $T(X)/\EEst$.

It will follow from Corollary \ref{cor:nonewequations} below that the substructure of $T(X)/\EEst$ generated by $T_n$ is already determined at stage $n$; that is, the canonical maps $T_n/\EEst_n\to T(X)/\EEst$ and $T_n/\EEst_n\to T_{n+1}/\EEst_{n+1}$ are injective. This says that $T_{n+1}/\EEst_{n+1}$ is a conservative extension of $T_n/\EEst_n$. This result is analogous to \cite[\S4.(2), p.~35]{Makkai05} for computads, but we emphasize that these are not computads, nor are they weak or strong $\omega$-categories.

\subsection{Properties of Free Globular Models}
\label{sec:properties}

The following lemma says that the models $T_n/\EEst_n$ are \emph{$n$-skeletal} \cite{Futia04}.
\begin{lemma}
\label{lem:skeletal}
Every term in $T_n$ of dimension $k\ge n$ is $\EEst_n$-equivalent to a term of the form $\id^{k-n} a$, where $a\in T_n$ and $\dim a=n$.
\end{lemma}
\begin{proof}
By induction on term structure. Let $t\in T_n$, $\dim t=k\ge n$. If $t\in X_n$ or $k=n$, or if $k=n+1$ and $t=\id a$, there is nothing to prove. For $t$ of the form $\id b$ and $k\ge n+2$, by the induction hypothesis, $\EEst_n\vdash b=\id^{k-1-n} a$ for some $a\in T_n$ and $\dim a=n$, so $\EEst_n\vdash t=\id^{k-n} a$.

For $t$ of the form $\src b$, by the induction hypothesis, $\EEst_n\vdash b=\id^{k-n+1} a$ for some $a\in T_n$ and $\dim a=n$, so $\EEst_n\vdash t=\src\id^{k-n+1} a=\id^{k-n} a$. The case of $\tgt$ is similar.

Finally, for $t$ of the form $u \crc m v$, by the induction hypothesis, $\EEst_n\vdash u=\id^{k-n} a$ and $\EEst_n\vdash v=\id^{k-n} b$ for some $a,b\in T_n$ and $\dim a=\dim b=n$. If $m < n$, then $\EEst_n\vdash u \crc m v=\id^{k-n} a\crc m id^{k-n} b=\id^{k-n}(a\crc m b)$ by (v), and $\dim(a\crc m b)=n$. Otherwise $n\le m\le k-1$. Since $\id^{k-n} a\crc m id^{k-n} b$ is meaningful, $\EEst_n\vdash\src^{k-m}\id^{k-n} a=\src^{k-m}\id^{k-n} b$, therefore $\EEst_n\vdash\id^{m-n} a=\id^{m-n} b$, and applying $\src^{m-n}$ gives $\EEst_n\vdash a=b$. It follows that
\begin{align*}
\EEst_n\vdash\id^{m-n+1} a\crc m \id^{m-n+1} b = \id^{m-n+1} a\crc m \id^{m-n+1} a = \id^{m-n+1} a
\end{align*}
by (x), therefore it is a consequence of $\EEst_n$ that
\begin{align*}
\id^{k-n} a\crc m \id^{k-n} b &= \id^{k-m-1}(\id^{m-n+1} a\crc m \id^{m-n+1} b)\\
&= \id^{k-m-1}\id^{m-n+1} a = \id^{k-n} a.
\extagqed
\end{align*}
\end{proof}

Read left to right, axioms (i)--(ix) of Fig.~\ref{fig:axioms} together with
(xii) and (xiii) of \S\ref{sec:globdef} form a rewrite system. The rules push
$\src$, $\tgt$ and $\id$ down, collapse $\src\id$ and $\tgt\id$, change
$\src\tgt$ to $\src\src$ and $\tgt\src$ to $\tgt\tgt$, and replace $\src x$ and
$\tgt x$ with $s(x)$ and $t(x)$. Axioms (x) and (xi) are left out since they are not
orientable in a way that preserves confluence. However, they play a special role that will emerge from the following development.

\begin{lemma}
\label{lem:strongnorm}
The rewrite system above is strongly normalizing.
\end{lemma}
\begin{proof}
The system without (xii) or (xiii) is clearly strongly normalizing because $\src$, $\tgt$, and $\id$ move down and the $\crc k$ move up, and only finitely many applications of the globular rules can occur before something else has to happen. Rules (xii) and (xiii) only introduce variables of lower dimension than the variable $x$ that was eliminated, and no other rule introduces new variables. Thus the whole system is strongly normalizing.
\exqed
\end{proof}

\begin{lemma}
\label{lem:nonewelements}
Adding new elements at stages above $n$ does not create any new elements of dimension $n$ or less.
That is, if $t\in T$ and $\dim t\le n$, then there exists $t'\in T_n$ such that $\EEst\vdash t = t'$.
\end{lemma}
\begin{proof}
Let $t'$ be an irreducible form of $t\in T$ obtained by applying the above rewrite rules. There are no occurrences of $\src$ or $\tgt$ in $t'$, as they all disappeared via (vi)-(vii) or (xii)-(xiii), so dimension does not increase along any path from the root to a leaf. Thus if $\dim t = \dim t'\le n$, then all atomic symbols in $t'$ are of dimension at most $n$, so $t'\in T_n$, and $\EEst\vdash t = t'$ since all the rewrite rules are axioms.
\exqed
\end{proof}

\begin{lemma}
\label{lem:confluent}
On terms in $T_{n}$, the rewrite system consisting of rules {\upshape(i)-(ix)} above is confluent. If the rules {\upshape(xii), (xiii)} are included, the system is confluent modulo $\EEst_{n-1}$.
\end{lemma}
\begin{proof}
To show confluence, we must show how to resolve conflicts due to overlapping redexes. The only overlapping redexes are between
\begin{itemize}
\item
rule (v) and rules (vi), (vii),
\item
rule (viii) and rules (ii), (iv), (vii), (ix), and (xiii),
\item
rule (ix) and rules (i), (iii), (vi), (viii), and (xii).
\end{itemize}
Up to symmetry, all but one of these cases have an immediate common reduct.

For the conflict between (v) and (vi),
\begin{align*}
& \src\id(u\crc k v) \gto{(vi)} u\crc k v\\
& \src\id(u\crc k v) \gto{(v)} \src(\id u\crc k\id v) \gto{(i)} \src\id u\crc k\src\id v \gto{(vi)^2} u\crc k v.
\end{align*}

For the conflict between (viii) and (ii) where $k<\dim u-2$,
\begin{align*}
& \src\tgt(u\crc k v)
\gto{(ii)} \src(\tgt u\crc k\tgt v)
\gto{(i)} \src\tgt u\crc k\src\tgt v
\gto{(viii)^2} \src\src u\crc k\src\src v\\
& \src\tgt(u\crc k v)
\gto{(viii)} \src\src(u\crc k v)
\gto{(i)^2} \src\src u\crc k\src\src v.
\end{align*}

For the conflict between (viii) and (ii) where $k=\dim u-2$,
\begin{align*}
& \src\tgt(u\crc k v)
\gto{(ii)} \src(\tgt u\crc k\tgt v)
\gto{(iii)} \src\tgt v
\gto{(viii)} \src\src v\\
& \src\tgt(u\crc k v)
\gto{(viii)} \src\src(u\crc k v)
\gto{(i)} \src(\src u\crc k\src v)
\gto{(iii)} \src\src v.
\end{align*}

For the conflict between (viii) and (iv),
\begin{align*}
& \src\tgt(u\crc k v)
\gto{(iv)} \src\tgt v
\gto{(viii)} \src\src v\\
& \src\tgt(u\crc k v)
\gto{(viii)} \src\src(u\crc k v)
\gto{(iii)} \src\src v.
\end{align*}

For the conflict between (viii) and (vii),
\begin{align*}
& \src\tgt\id u
\gto{(vii)} \src u
&& \src\tgt\id u
\gto{(viii)} \src\src\id u
\gto{(vi)} \src u.
\end{align*}

For the conflict between (viii) and (ix),
\begin{align*}
& \src\tgt\src u
\gto{(ix)} \src\tgt\tgt u
\gto{(viii)^2} \src\src\src u
&& \src\tgt\src u
\gto{(viii)} \src\src\src u.
\end{align*}

All other cases not involving (xii) and (xiii) are symmetric versions of those above. This establishes the first statement of the lemma.

For the second statement, in the case of a conflict between (viii) and (xiii), we do not have an immediate common reduct, but we use the assumption on the globular functions $s$ and $t$.
\begin{align*}
& \src\tgt x
\gto{(ix)} \src t(x)
&& \src\tgt x
\gto{(viii)} \src\src x
\gto{(xii)} \src s(x)
\end{align*}
Since $\dim x\le n$, we have by assumption $\EEst_{n-1}\vdash\src s(x)=\src t(x)$.

The case of a conflict between (ix) and (xii) is symmetric.
\exqed
\end{proof}

Two terms are \emph{picture-equivalent} if they are provably equivalent using axioms (v), (x), and (xi) only.

\begin{theorem}
\label{thm:nonewequations}
Let $u,v\in T_n$.
\begin{enumerate}[{\upshape(a)}]
\item
All irreducible forms of $u$ are picture-equivalent.
\item
If $\EEst_n\vdash u=v$, then for any $u',v'$ irreducible forms of $u,v$ respectively, $u'$ and $v'$ are picture-equivalent.
\end{enumerate}
\end{theorem}
\begin{proof}
We prove (a) and (b) by mutual induction on $n$.

(a)
For $n=0$, when reducing $u$ to an irreducible form, the reduction rules (xii) and (xiii) never apply. By Lemma \ref{lem:confluent}, the rewrite system without these rules is confluent, therefore ireducible forms are unique. Thus (a) holds trivially in this case.

For $n>0$, by Lemma \ref{lem:strongnorm} and K\"onig's lemma, there is a uniform bound on the length of derivations starting from any term (which may depend on $s$ and $t$). This allows proofs by induction on derivation length, which we will use here.

Consider two reduction sequences from $u$ starting with $u\to u_1$ and $u\to u_2$, respectively. If $u\to u_1$ and $u\to u_2$ are instances of any pair of overlapping redexes except (viii),(xiii) or (ix),(xii), or if the redexes do not overlap, then by Lemma \ref{lem:confluent} there is a term $u'$ such that $u_1\to u'$ and $u_2\to u'$. By the induction hypothesis on derivation length, all irreducible forms of $u_1$ are picture-equivalent, all irreducible forms of $u_2$ are picture-equivalent, and there is a nonempty intersection containing an irreducible form of $u'$, thus all irreducible forms of $u$ starting from $u\to u_1$ or $u\to u_2$ are picture-equivalent.

If the first steps of a pair of reductions from $u$ involve (viii) and (xiii), say $u=e(\src\tgt x)$ and $e(\src\tgt x)\to e(\src\src x)$ and $e(\src\tgt x)\to e(\src t(x))$, let $z$ be a fresh variable of dimension $\dim x-2$. Any irreducible form of $e(z)$ is of the form $e'(\prefix z)$, where $\prefix$ is an operator of the form $\id^i\src^j$ or $\id^i\tgt^j$ and $e'(\prefix z)$ contains no occurrences of $\src$ or $\tgt$ outside $\prefix z$. Since $u\in T_n$, $e'(\prefix\src s(x))$ and $e'(\prefix\src t(x))$ are also in $T_n$. Let $u_1$ and $u_2$ be irreducible forms of $\prefix\src s(x)$ and $\prefix\src t(x)$, respectively. Then $e'(u_1)$ and $e'(u_2)$ are irreducible and
\begin{align*}
& e(\src\src x)\to e'(\prefix\src s(x)) \to e'(u_1)
&& e(\src t(x))\to e'(\prefix\src t(x)) \to e'(u_2).
\end{align*}
We have $\EEst_{n-1}\vdash\src s(x)=\src t(x)$ by an assumption on $s$ and $t$, therefore $\EEst_{n-1}\vdash\prefix\src s(x)=\prefix\src t(x)$. By the induction hypothesis on $n$, $u_1$ and $u_2$ are picture-equivalent, therefore so are $e'(u_1)$ and $e'(u_2)$. By the induction hypothesis on derivation length, all irreducible forms of $e(\src\src x)$, including $e'(u_1)$, are picture-equivalent, and all irreducible forms of $e(\src t(x))$, including $e'(u_2)$, are picture-equivalent, therefore all irreducible forms of $u$ starting with (viii) and (xiii) are picture-equivalent.

The argument for conflicting reduction steps (ix) and (xii) is symmetric.

(b)
By (a), it suffices to show that for any $e(u),e(v)\in T_n$ where $u=v$ is an axiom of $\EEst_n$, there exist irreducible forms $u',v'$ of $e(u),e(v)$, respectively, such that $u'$ and $v'$ are picture-equivalent. This is trivial if $u=v$ is one of (i)-(ix) and (xii)-(xiii), as these are reduction rules themselves, so $u$ and $v$ have a common irreducible form and there is nothing to do. Thus we need only consider the picture axioms (x) and (xi).
\begin{enumerate}[{\upshape(i)}]
\setcounter{enumi}{9}
\item
$\id u\crc k\id u = \id u,\ k = \dim u$
\item
$(u \crc n v) \crc m (x \crc n y) = (u \crc m x) \crc n (v \crc m y),\ m\ne n$
\end{enumerate}

For the case of (x), say we have (x)-equivalent terms $e(\id u\crc k\id u)$ and $e(\id u)$. Let $z$ be a fresh variable of dimension $\dim u$. Apply the reduction rules in tandem to $e(\id z\crc k\id z)$ and $e(\id z)$ to obtain irreducible forms. Depending on the number and arrangement of occurrences of $\src$, $\tgt$, and $\id$ at ancestors of the critical subterms, after this step we will have either
\begin{itemize}
\item
$e'(\id^m z)$ and $e'(\id^m z)$ for some $m\ge 0$, if there was a reduction step using (iii) or (iv) applied to the left-hand side, for example (iii) with redex $\src(\id z\crc k\id z)$ on the left and $\src\id z$ on the right,
\item
$e'(v)$ and $e'(v)$ with no occurrence of $z$, if there was a reduction step using (iii) or (iv) that eliminated the critical subterms altogether, for example (iii) with redex $\src((\id z\crc k\id z)\crc k v)$ on the left and $\src(\id z\crc k v)$ on the right, or
\item
$e'(\id^m z\crc k \id^m z)$ and $e'(\id^m z)$ for some $m\ge 1$, if neither of the above events occurred.
\end{itemize}
Note that the operator $\id^m$ that appears in the first and third cases cannot be $\id^i\src^j$ or $\id^i\tgt^j$ with $j\ge 1$ as in case (a), because (iii) or (iv) would have applied.

In the first two cases, we have a common irreducible form, so there is nothing to do. In the last case, let $u'$ be an irreducible form of $u$ and let $u''$ be an irreducible form of $\id^m u'$. Then $u''$ is an irreducible form of $\id^m u$ as well, and $e'(u''\crc k u'')$ and $e'(u'')$ are irreducible forms of the original terms $e(\id u\crc k\id u)$ and $e(\id u)$, respectively. Note that since $u'$ is irreducible, $u''$ is obtained from $\id^m u'$ simply by applying axiom (v) in the left-to-right direction, thus $\id^m u'$ and $u''$ are equivalent modulo (v). Moreover,
\begin{align*}
\id^m u'\crc k\id^m u'
&= \id^{m-1}(\id u'\crc k\id u') && \text{by (v)}\\
&= \id^{m-1}\id u' && \text{by (x)}\\
&= \id^m u',
\end{align*}
so $u''\crc k u''$ and $u''$ are picture-equivalent, therefore so are $e'(u''\crc k u'')$ and $e'(u'')$.

For the case of (xi), suppose we have (xi)-equivalent terms
\begin{align}
& e((u \crc k v) \crc m (x \crc k y)) && e((u \crc m x) \crc k (v \crc m y))\label{eq:exchorig}
\end{align}
We will show that they have (xi)-equivalent irreducible forms. Let $z_u,z_v,z_x,z_y$ be fresh variables of the same dimension as $u,v,x,y$. Apply the reduction rules in tandem to
\begin{align*}
& e((z_u \crc k z_v) \crc m (z_x \crc k z_y)) && e((z_u \crc m z_x) \crc k (z_v \crc m z_y))
\end{align*}
to obtain irreducible forms. By an analysis similar to the previous case, depending on the number and arrangement of occurrences of $\src$, $\tgt$, and $\id$ at ancestors of the critical subterms, after this step we will have either
\begin{itemize}
\item
$e'(v)$ and $e'(v)$,
\item
$e'(\prefix z_u \crc k \prefix z_v)$ and $e'(\prefix z_u \crc k \prefix z_v)$,
\item
$e'(\prefix z_x \crc k \prefix z_y)$ and $e'(\prefix z_x \crc k \prefix z_y)$,
\item
$e'(\prefix z_u \crc m \prefix z_x)$ and $e'(\prefix z_u \crc m \prefix z_x)$,
\item
$e'(\prefix z_v \crc m \prefix z_y)$ and $e'(\prefix z_v \crc m \prefix z_y)$, or
\item
$e'((\prefix z_u \crc k \prefix z_v) \crc m (\prefix z_x \crc k \prefix z_y))$ and $e'((\prefix z_u \crc m \prefix z_x) \crc k (\prefix z_v \crc m \prefix z_y))$,
\end{itemize}
where $\prefix$ is an operator of the form $\id^i\src^j$ or $\id^i\tgt^j$ for some $i,j\ge 0$. In all but the last case, we have common irreducible forms, so there is nothing to do. In the last case, let $u',v',x',y'$ be irreducible forms of $\prefix u,\prefix v,\prefix x,\prefix y$, respectively. Then the original terms \eqref{eq:exchorig} have irreducible forms
\begin{align*}
& e'((u' \crc k v') \crc m (x' \crc k y')) && e'((u' \crc m x') \crc k (v' \crc m y')), 
\end{align*}
respectively, which are equivalent via (xi).
\exqed
\end{proof}

\begin{corollary}
\label{cor:nonewequations}
Adding new elements at higher stages does not create any new equations among elements of dimension $n$ or less. That is, if $u,v\in T$, $\dim u=\dim v\le n$, and $\EEst\vdash u = v$, then $\EEst_n \vdash u' = v'$, where $u'$ and $v'$ are irreducible forms of $u$ and $v$, respectively.
\end{corollary}
\begin{proof}
Suppose $u,v\in T$, $\dim u=\dim v\le n$, and $\EEst\vdash u = v$. Let $u'$ and $v'$ be irreducible forms of $u$ and $v$, respectively. Since $\dim u'=\dim u=\dim v=\dim v'\le n$, and since $u'$ and $v'$ contain no occurrences of $\src$ or $\tgt$, dimension cannot increase on any path from a root to a leaf in $u'$ or $v'$, thus all subterms of $u'$ and $v'$ are of dimension at most $n$ and $u',v'\in T_n$. By Theorem \ref{thm:nonewequations}, $u'$ and $v'$ are picture equivalent, therefore provably equivalent in $\EEst_n$.
\exqed
\end{proof}

\section{Completeness and Decidability}
\label{sec:completeness}

In this section we prove that $\theory$ is complete for the strict
equational theory of the homotopy models of \S\ref{sec:homotopymodels}, and that
the theory is decidable. Both rest on a single construction: an interpretation
$\Pic$ that reads a term as a cell over an ambient space with no topological
structure at all. We show that $\Pic$ is a complete invariant, that it is
determined by finite data, and that the completeness is not an
artifact of the trivial topology, since the same separation is realized in a
finite CW-complex, as shown in \S\ref{sec:CW}.

\subsection{Pictures}
\label{sec:pictures}

Fix a free globular model $T/\EEst$ relative to a given pair of globular maps $s,t$. A \emph{picture} is one of a certain class of $n$-cells over the ambient space $X$ consisting of the elements of the stratified set $(X,\dim)$ with the trivial indiscrete topology whose only open sets are the empty set and $X$. This is a little bizarre, as the topology has no separation whatsoever. However, the definition of homotopy models does not place any restrictions at all on the topology of the ambient space. For this $X$, all functions $I^n\to X$ are continuous, thus the only interesting structure we care about is boundary preservation.

We define an interpretation $\Pic$ that interprets each term $t$ of dimension $n$ as a picture $\Pic(t):I^n\to X$. 
A \emph{picture} is then a cell of the form $\Pic(t)$ for some $t$.

For $a\in X_n$, define
\begin{equation*}
\Pic(a)(\seq x1n) = \begin{cases}
a, & \text{\small if $0<x_i<1$ for all $1\le i\le n$}\\
\Pic(s(a))(\seq x2n), & \parbox[t]{4.5cm}{\small if $x_j=0$ for some $1\le j\le n$\\and $0<x_i<1$ for $j+1\le i\le n$}\\
\Pic(t(a))(\seq x2n), & \parbox[t]{4.5cm}{\small if $x_j=1$ for some $1\le j\le n$\\and $0<x_i<1$ for $j+1\le i\le n$.}
\end{cases}
\end{equation*}
The three conditions are mutually exclusive and exhaustive, so $\Pic(a)$ is uniquely defined on all elements of $I^n$.
Once $\Pic$ is defined on elements of $X$, we can extend homomorphically to all terms, as $\theory$ is sound in homotopy models.

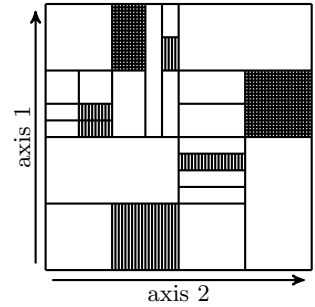
\begin{wrapfigure}{r}{0.44\textwidth}
\vspace{-\intextsep}
\begin{center}
\begin{tikzpicture}[-, >=stealth', auto, thick, scale=0.44]
\small
\draw (0,0) -- (8,0);
\draw (0,0) -- (0,8);
\draw (8,8) -- (8,0);
\draw (8,8) -- (0,8);
\draw[->] (-.3,.2) -- (-.3,7.8);
\node[rotate=90] at (-.7,4) {axis $1$};
\draw[->] (.2,-.3) -- (7.8,-.3);
\node at (4,-.7) {axis $2$};
\draw (4,0) -- (4,8);
\draw (0,4) -- (8,4);
\draw (2,4) -- (2,8);
\draw (3,4) -- (3,8);
\draw (3.5,4) -- (3.5,8);
\draw (2,6) -- (3,6);
\draw (0,6) -- (2,6);
\draw (0,5) -- (2,5);
\draw (0,4.5) -- (2,4.5);
\draw (3.5,6) -- (8,6);
\draw (6,0) -- (6,6);
\draw (4,5) -- (6,5);
\draw (0,2) -- (6,2);
\draw (3.5,7) -- (4,7);
\draw (4,2.5) -- (6,2.5);
\draw (4,3) -- (6,3);
\draw (4,3.5) -- (6,3.5);
\draw (2,0) -- (2,2);
\draw (1,4) -- (1,6);
\foreach \x in {1,...,9}
  \draw (1 + \x*.1,4) -- (1 + \x*.1,5);
\foreach \x in {1,...,19}
  \draw (4 + \x*.1,3) -- (4 + \x*.1,3.5);
\foreach \x in {1,...,19}
  \draw (2 + \x*.1,0) -- (2 + \x*.1,2);
\foreach \x in {1,...,4}
  \draw (3.5 + \x*.1,6) -- (3.5 + \x*.1,7);
\foreach \x in {1,...,9}
  \draw (2 + \x*.1,6) -- (2 + \x*.1,8);
\foreach \y in {1,...,19}
  \draw (2,6 + \y*.1) -- (3,6 + \y*.1);
\foreach \x in {1,...,19}
  \draw (6 + \x*.1,4) -- (6 + \x*.1,6);
\foreach \y in {1,...,19}
  \draw (6,4 + \y*.1) -- (8,4 + \y*.1);
\end{tikzpicture}
\end{center}
\caption{A picture of dimension 2}
\label{fig:picture}
\vspace{-\intextsep}
\end{wrapfigure}
A 2-dimensional picture $I^2\to X$ is illustrated in Fig.~\ref{fig:picture}. In that picture, every point of the interior of each unshaded small rectangle is labeled with the same letter $a\in X_2$. A small rectangle hatched with vertical lines is an identity that duplicates its source and target, which are its northern and southern borders, respectively. Every doubly hatched rectangle is a double identity $\id^2 c$, where $c\in X_0$. Horizontal lines are labeled with elements of $X_1\cup X_0$ and are one-dimensional pictures, defined inductively in the same fashion. Vertical lines are labeled with element of $X_0$ and are identities. Note that if we assume the sides of the outer rectangle are of unit length, then all interior lines occur at fractional powers of 2, aka dyadic rationals; thus the picture is uniquely determined by its points at regularly spaced dyadic rationals of sufficient density. For this reason, a picture conveys only a finite amount of information and can be represented by a finite array.

\begin{lemma}
\label{lem:piccell}
For all terms $t$, $\Pic(t)$ is a cell.
\end{lemma}
\begin{proof}
All operations map cells to cells, so it suffices to show that $\Pic(a)$ is a cell for $a\in X_n$. By Lemma \ref{lem:cellprops}, a homotopy $\alpha:I^n\to X$ is a cell iff whenever $1\le j\le n$ and $x_j\in\{0,1\}$, $\alpha(\seq x1n)$ is independent of $\seq x1{j-1}$. This is true of $\Pic(a)$ for $a\in X_n$: if $j$ is maximum such that $x_j\in\{0,1\}$, then
\begin{align*}
& \Pic(a)(\seq x1n)
= \begin{cases}
\Pic(s(a))(\seq x2n), & \text{if $x_j=0$}\\
\Pic(t(a))(\seq x2n), & \text{if $x_j=1$,}
\end{cases}
\end{align*}
If $j=1$, there is nothing to prove. If $j\ge 2$, then by the induction hypothesis on dimension, $\Pic(a)(\seq x1n)$ is independent of $\seq x2{j-1}$, and since $x_1$ does not appear on the right-hand side in this case, it is independent of $x_1$ also.
\exqed
\end{proof}

It follows from this construction that two provably equivalent terms have the same picture. In the next section we show that a term can be uniquely recovered up to $\EEst_n$-equivalence from its picture.

\subsection{Finite Representation}
\label{sec:discrete}

Although pictures are maps $I^n\to X$, they only contain a finite amount of information. We would like to have a finite representation for algorithmic purposes. In this section we develop such a representation.

For terms $t$ with $\dim t=n$ and $0\le m\le n-1$, let $\hgt mt$ be a function giving the maximum depth of nesting of compositions $\crc m$ in $t$. We define $\hgt mt$ formally by induction on dimension and term structure for terms that are irreducible with respect to the reduction rules of \S\ref{sec:properties}.

For $a\in X_0$ and $k\ge 0$, let $\hgt m{\id^k a} = 1$.
For $a\in X_n$ with $n\ge 1$, let $\hgt m{\id^k a} = \max \hgt m{\src a},\hgt m{\tgt a}$.
For compositions $u\crc k v$, let
\begin{align*}
\hgt m{u \crc k v} = \begin{cases}
\max \hgt mu,\hgt mv + 1 & \text{if $k=m$,}\\
\max \hgt mu,\hgt mv & \text{if $k\ne m$.}
\end{cases}
\end{align*}

\begin{lemma}
\label{lem:height}
For all $t$ with $\dim t\ge 1$, $\hgt mt \ge \max\hgt m{\src t},\hgt m{\tgt t}$.
\end{lemma}
\begin{proof}
For $a\in X_n$ and $k\ge 1$,
\begin{align*}
\hgt m{\id^k a} &= \hgt m{\id^{k-1} a} = \max\hgt m{\src\id^k a},\hgt m{\tgt\id^k a},
\end{align*}
and for $a\in X_n$ with $n\ge 1$, $\hgt ma = \max\hgt m{\src a},\hgt m{\tgt a}$ by definition.

For compositions $u\crc k v$ where $k\ne m$ and $k<\dim u-1=\dim v-1$,
\begin{align*}
\hgt m{u\crc k v} &= \max \hgt mu,\hgt mv\\
&\ge \max \hgt m{\src u},\hgt m{\tgt u},\hgt m{\src v},\hgt m{\tgt v}\\
&\ge \max(\max\hgt m{\src u},\hgt m{\src v}),(\max\hgt m{\tgt u},\hgt m{\tgt v})\\
&= \max \hgt m{\src u\crc k\src v},\hgt m{\tgt u\crc k\tgt v})\\
&= \max \hgt m{\src(u\crc k v)},\hgt m{\tgt(u\crc k v)}.
\end{align*}
For compositions $u\crc k v$ where $k\ne m$ and $k=\dim u-1=\dim v-1$,
\begin{align*}
\hgt m{u\crc k v} &= \max \hgt mu,\hgt mv\\
&\ge \max \hgt m{\src v},\hgt m{\tgt u}\\
&= \max \hgt m{\src(u\crc k v)},\hgt m{\tgt(u\crc k v)}.
\end{align*}
For compositions $u\crc m v$ where $m<\dim u-1=\dim v-1$,
\begin{align*}
\hgt m{u\crc m v} &= 1 + \max \hgt mu,\hgt mv\\
&\ge 1 + \max \hgt m{\src u},\hgt m{\tgt u},\hgt m{\src v},\hgt m{\tgt v}\\
&= \max(1 + \max \hgt m{\src u},\hgt m{\src v}),(1 + \hgt m{\tgt u},\hgt m{\tgt v})\\
&= \max \hgt m{\src u\crc m\src v},\hgt m{\tgt u\crc m\tgt v})\\
&= \max \hgt m{\src(u\crc m v)},\hgt m{\tgt(u\crc m v)}.
\end{align*}
Finally, for compositions $u\crc m v$ where $m=\dim u-1=\dim v-1$,
\begin{align*}
\hgt m{u\crc m v} &= 1 + \max \hgt mu,\hgt mv\\
&> \max \hgt m{\src v},\hgt m{\tgt u}\\
&= \max \hgt m{\src(u\crc m v)},\hgt m{\tgt(u\crc m v)}.
\extagqed
\end{align*}
\end{proof}

For $k\in\naturals\setminus\{0\}$ and $r\in\reals$, define
$I_k = \set{i/2^k}{i\in\{0,1,\ldots,2^k\}} \subs I$ and
$g_k(r) = (\floor{2^{k-1}r} + \ceil{2^{k-1}r})/2^k$. Note that if $r\in I$,
then $g_k(r)\in I_k$.
\begin{lemma}
\label{lem:closestodd}
If $r=i/2^k$ with $i\in\naturals$, then $g_k(r)=r$, otherwise $g_k(r)$ is the closest number to $r$ of the form $m/2^k$ for $m$ odd.
\end{lemma}
\begin{proof}
If $r=i/2^k$, then
\begin{align*}
g_k(r) &= \frac{\floor{2^{k-1}r}+\ceil{2^{k-1}r}}{2^k}
= \frac{\floor{2^{k-1}i/2^k}+\ceil{2^{k-1}i/2^k}}{2^k}\\
&= \frac{\floor{i/2}+\ceil{i/2}}{2^k}
= \frac{i}{2^k} = r.
\end{align*}
If $r = (i+\eps)/2^k$ with $0 < \eps < 1$, if $i$ is even, say $i=2j$,
\begin{align*}
g_k(r) &= \frac{\floor{2^{k-1}r}+\ceil{2^{k-1}r}}{2^k}
= \frac{\floor{2^{k-1}(2j+\eps)/2^k}+\ceil{2^{k-1}(2j+\eps)/2^k}}{2^k}\\
&= \frac{\floor{j+\eps/2} + \ceil{j+\eps/2}}{2^k}
= \frac{2j}{2^k} + \frac{\floor{\eps/2} + \ceil{\eps/2}}{2^k}
= \frac{2j+1}{2^k},
\end{align*}
and if $i$ is odd, say $i=2j+1$,
\begin{align*}
g_k(r) &= \frac{\floor{2^{k-1}r}+\ceil{2^{k-1}r}}{2^k}
= \frac{\floor{2^{k-1}(2j+1+\eps)/2^k}+\ceil{2^{k-1}(2j+1+\eps)/2^k}}{2^k}\\
&= \frac{\floor{j+(1+\eps)/2} + \ceil{j+(1+\eps)/2}}{2^k}
= \frac{2j}{2^k} + \frac{\floor{(1+\eps)/2} + \ceil{(1+\eps)/2}}{2^k}\\
&= \frac{2j+1}{2^k}.
\extagqed
\end{align*}
\end{proof}

\begin{lemma}
\label{lem:gk}
For  $r\in\reals$, $k\ge 1$: $g_{k-1}(2r) = 2g_k(r)$ and $g_{k-1}(2r-1) = 2g_k(r)-1$.
\end{lemma}
\begin{proof}
\begin{align*}
g_{k-1}(2r)
&= \frac{\floor{2^{k-2}2r} + \ceil{2^{k-2}2r}}{2^{k-1}}
= 2\frac{\floor{2^{k-1}r} + \ceil{2^{k-1}r}}{2^k}
= 2g_k(r),\\
g_{k-1}(2r-1)
&= \frac{\floor{2^{k-2}(2r-1)} + \ceil{2^{k-2}(2r-1)}}{2^{k-1}}
= \frac{\floor{2^{k-1}r-2^{k-2}} + \ceil{2^{k-1}r-2^{k-2}}}{2^{k-1}}\\
&= \frac{\floor{2^{k-1}r} + \ceil{2^{k-1}r}-2^{k-2}-2^{k-2}}{2^{k-1}}
= 2\frac{\floor{2^{k-1}r} + \ceil{2^{k-2}r}}{2^k} - 1\\
&= 2g_k(r)-1.
\extagqed
\end{align*}
\end{proof}

\begin{theorem}
\label{thm:rep}
Let $\dim t=n$ and let $k_m\ge\hgt{n-m}t$, $1\le m\le n$. Then
\begin{align*}
\Pic(t)(\seq r1n) = \Pic(t)(g_{k_1}(r_1),\ldots,g_{k_n}(r_n)). 
\end{align*}
Thus $\Pic(t):I^n\to X$ is uniquely determined by its values on the finite set $I_{k_1}\times\cdots\times I_{k_n}$.
\end{theorem}
\begin{proof}
By definition, for $a\in X_n$ with $n\ge 1$,
\begin{align*}
\Pic(a)(\seq r1n) &= \begin{cases}
a, & \text{if $0 < r_i < 1$ for all $i$,}\\
\Pic(s(a))(\seq r2n), & \text{if $r_j=0$ and $0 < r_i < 1$ for all $i>j$,}\\
\Pic(t(a))(\seq r2n), & \text{if $r_j=1$ and $0 < r_i < 1$ for all $i>j$.}
\end{cases}
\end{align*}
and
\begin{align*}
& \Pic(a)(g_{k_1}(r_1),\ldots,g_{k_n}(r_n))\\
&= \begin{cases}
a, & \text{if $0 < g_{k_i}(r_i) < 1$ for all $i$,}\\
\Pic(s(a))(g_{k_2}(r_2),\ldots,g_{k_n}(r_n)), & \text{if $g_{k_j}(r_j)=0$ and $0 < g_{k_i}(r_i) < 1$ for all $i>j$,}\\
\Pic(t(a))(g_{k_2}(r_2),\ldots,g_{k_n}(r_n)), & \text{if $g_{k_j}(r_j)=1$ and $0 < g_{k_i}(r_i) < 1$ for all $i>j$.}
\end{cases}
\end{align*}
The conditions defining the three cases are equivalent for $r_i$ and $g_{k_i}(r_i)$; that is, $0<r_i<1$ iff $0<g_{k_i}(r_i)<1$, $r_j=0$ iff $g_{k_j}(r_j)=0$, and $r_j=1$ iff $g_{k_j}(r_j)=1$. Thus the corresponding cases in the two definitions are either both true or both false. For the first case, the two definitions are the same. For the second and third cases, we have $k_i\ge\hgt{n-i}a\ge\max\hgt{n-i}{\src a},\hgt{n-i}{\tgt a}$ for all $i$ by Lemma \ref{lem:height}, therefore
\begin{align*}
\Pic(s(a))(\seq r2n) &= \Pic(s(a))(g_{k_2}(r_2),\ldots,g_{k_n}(r_n))\\
\Pic(t(a))(\seq r2n) &= \Pic(t(a))(g_{k_2}(r_2),\ldots,g_{k_n}(r_n))
\end{align*}
by the induction hypothesis. As all three cases coincide, $\Pic(a)(\seq r1n)=\Pic(a)(g_{k_1}(r_1),\ldots,g_{k_n}(r_n))$.

For terms of the form $\id t$, if $k_m\ge\hgt{n-m}{\id t}$ for $1\le m\le n$, then
\begin{align*}
k_m \ge \hgt{n-m}{\id t} \ge \max\hgt{n-m}{\src\id t},\hgt{n-m}{\tgt\id t}=\hgt{n-m}t
\end{align*}
by Lemma \ref{lem:height}, therefore by the induction hypothesis,
\begin{align*}
\Pic(\id t)(\seq r1n) &= \id\Pic(t)(\seq r1n)
= \Pic(t)(\seq r2n)\\
&= \Pic(t)(g_{k_2}(r_2),\ldots,g_{k_n}(r_n))\\
&= \id\Pic(t)(g_{k_1}(r_1),g_{k_2}(r_2),\ldots,g_{k_n}(r_n))\\
&= \Pic(\id t)(g_{k_1}(r_1),g_{k_2}(r_2),\ldots,g_{k_n}(r_n)).
\end{align*}
Finally, for compositions $u\crc{n-m}v$, where $n=\dim u=\dim v$ and $1\le m\le n$, we have 
$k_i\ge\hgt{n-i}{u\crc{n-m}v}=\max\hgt{n-i}u,\hgt{n-i}v$ for $i\ne m$ and  
$k_m\ge\hgt{n-m}{u\crc{n-m}v}=1+\max\hgt{n-m}u,\hgt{n-m}v$,
thus $k_m-1\ge\max\hgt{n-m}u,\hgt{n-m}v$. Using the induction hypothesis and Lemma \ref{lem:gk} at index $m$,
\begin{align*}
& \Pic(u\crc{n-m}v)(\seq r1n)\\
&= (\Pic(u)\crc{n-m}\Pic(v))(\seq r1n)\\ 
&= \begin{cases}
\Pic(v)(\seq r1{m-1},2r_m,\seq r{m+1}n), & \text{if $0\le r_m\le 1/2$}\\ 
\Pic(u)(\seq r1{m-1},2r_m-1,\seq r{m+1}n), & \text{if $1/2\le r_m\le 1$} 
\end{cases}\\
&= \begin{cases}
\Pic(v)(g_{k_1}(r_1),\ldots,g_{k_{m-1}}(r_{m-1}),g_{k_m-1}(2r_m),g_{k_{m+1}}(r_{m+1}),\ldots,g_{k_n}(r_n)),\\
\text{\hspace{8cm}if $0\le g_{k_m-1}(r_m)\le 1/2$}\\ 
\Pic(u)(g_{k_1}(r_1),\ldots,g_{k_{m-1}}(r_{m-1}),g_{k_m-1}(2r_m-1),g_{k_{m+1}}(r_{m+1}),\ldots,g_{k_n}(r_n)),\\
\text{\hspace{8cm}if $1/2\le g_{k_m-1}(r_m)\le 1$}\\ 
\end{cases}\\
&= \begin{cases}
\Pic(v)(g_{k_1}(r_1),\ldots,g_{k_{m-1}}(r_{m-1}),2g_{k_m}(r_m),g_{k_{m+1}}(r_{m+1}),\ldots,g_{k_n}(r_n)),\\
\text{\hspace{8cm}if $0\le g_{k_m}(r_m)\le 1/2$}\\ 
\Pic(u)(g_{k_1}(r_1),\ldots,g_{k_{m-1}}(r_{m-1}),2g_{k_m}(r_m)-1,g_{k_{m+1}}(r_{m+1}),\ldots,g_{k_n}(r_n)),\\
\text{\hspace{8cm}if $1/2\le g_{k_m}(r_m)\le 1$}\\ 
\end{cases}\\
&= (\Pic(u)\crc{n-m}\Pic(v))(g_{k_1}(r_1),\ldots,g_{k_n}(r_n))\\
&= \Pic(u\crc{n-m}v)(g_{k_1}(r_1),\ldots,g_{k_n}(r_n)).
\extagqed
\end{align*}
\end{proof}

\subsection{Completeness over Homotopy Models}
\label{sec:completenesshomotopy}

We have shown that every pair of terms that are $\EEst$-equivalent have the same picture, which is a particular homotopy model with trivial topological structure. Now we show the converse. This says that the $\Pic$ interpretation is isomorphic to the free model with given globular functions $s$ and $t$. It also establishes that $\theory$ (augmented by (xii) and (xiii) for each particular model) is complete for the equational theory of homotopy models, although one might wish to show completeness over a more conventional class of topological models such as compact Hausdorff spaces or CW-complexes~\cite{Hatcher02}. We address the latter question in \S\ref{sec:CW}.

\begin{lemma}
\label{lem:split}
Let $\Pic(t):I^n\to X$ be a picture and let $1\le m\le n$. If $\Pic(t)=\alpha\crc{n-m}\beta$ and the interface $\Pic(t)\s m{1/2} = \id^{m-1}\src^m\alpha=\id^{m-1}\tgt^m\beta$ is a picture, then $\alpha$ and $\beta$ are pictures.
\end{lemma}
\begin{proof}
Assume that $t$ is irreducible. The proof is by induction on term structure and on dimension.

The term $t$ cannot be a symbol $a\in X_n$, because $\Pic(a)(1/2,\ldots,1/2) = a$, whereas the hyperplane $\Pic(t)\s m{1/2}$, which contains that point, is assumed to be a picture of dimension $n-1$, thus cannot contain any element of $X_n$.

If $t=\id u$ and $m=1$, then $\Pic(t)=\id\Pic(u)=\id\Pic(u)\crc{n-1}\id\Pic(u)=\Pic(t)\crc{n-1}\Pic(t)$, so $\alpha=\beta=\Pic(t)$.

If $t=\id u$ and $m>1$, then $\Pic(t) = \id\Pic(u)$, where $u=\src t=\tgt t$.
The picture $\Pic(u)=\src\Pic(t)$ is represented by the bottom horizontal line in \eqref{eq:pic1}.
\begin{align}
\begin{array}c
\begin{tikzpicture}[-, >=stealth', auto]
\small
%\draw[help lines] (0,0) grid (2,2);
\draw (0,0) -- (2,0);
\draw (0,0) -- (0,2);
\draw (2,2) -- (2,0);
\draw (2,2) -- (0,2);
\draw (1,0) -- (1,2);
\draw[->] (-.3,.2) -- (-.3,1.8);
\node[rotate=90] at (-.5,1) {axis $1$};
\draw[->] (.2,-.3) -- (1.8,-.3);
\node at (1,-.5) {axis $m$};
\node at (.5,1) {$\beta$};
\node at (1.5,1) {$\alpha$};
\node[anchor=east] at (-1,1) {$\Pic(t) =$};
\end{tikzpicture}
\end{array}
\label{eq:pic1}
\end{align}
As $\Pic(t)=\id\src\Pic(t)$, this is also true of any fiber over $\Pic(u)$, including $\alpha$ and $\beta$ and their interface $\Pic(t)\s m{1/2}$; thus $\alpha=\id\src\alpha$ and $\beta=\id\src\beta$. But since the interface $\Pic(t)\s m{1/2}$ is a picture, its source $\src(\Pic(t)\s m{1/2})=\Pic(t)\s{1,m}{0,1/2}$ is also a picture. This is represented by the midpoint of the bottom horizontal line in \eqref{eq:pic1}. It is the interface of $\src\beta$ and $\src\alpha$, which compose under $\crc{n-m}$ to give the picture $\src\Pic(t)$. By the induction hypothesis on dimension, $\src\alpha$ and $\src\beta$ are pictures. Since $\alpha=\id\src\alpha$ and $\beta=\id\src\beta$, $\alpha$ and $\beta$ are pictures.

If $t=u\crc{n-k}v$ and $k=m$, then $\Pic(t)=\Pic(u)\crc{n-m}\Pic(v)$, so $\alpha=\Pic(u)$ and $\beta=\Pic(v)$.

Finally, if $t=u\crc{n-k}v$ where $k\ne m$, decompose $\Pic(t)$ into four quadrants $\gamma_{ij}:I^n\to X$ for $i,j\in\{0,1\}$ as follows:
\begin{align*}
\gamma_{01}\s{m,k}{x_m,x_k} &= \Pic(t)\s{m,k}{x_m/2,(x_k+1)/2}\\
\gamma_{11}\s{m,k}{x_m,x_k} &= \Pic(t)\s{m,k}{(x_m+1)/2,(x_k+1)/2}\\
\gamma_{00}\s{m,k}{x_m,x_k} &= \Pic(t)\s{m,k}{x_m/2,x_k/2}\\
\gamma_{10}\s{m,k}{x_m,x_k} &= \Pic(t)\s{m,k}{(x_m+1)/2,x_k/2}
\end{align*}
This decomposition is shown in \eqref{eq:pic2}.
\begin{align}
\begin{array}c
\begin{tikzpicture}[-, >=stealth', auto]
\small
\draw[help lines] (0,0) grid (2,2);
\draw (0,0) -- (2,0);
\draw (0,0) -- (0,2);
\draw (2,2) -- (2,0);
\draw (2,2) -- (0,2);
\draw (0,1) -- (2,1);
\draw (1,0) -- (1,2);
\node at (.5,.5) {$\gamma_{00}$};
\node at (.5,1.5) {$\gamma_{01}$};
\node at (1.5,.5) {$\gamma_{10}$};
\node at (1.5,1.5) {$\gamma_{11}$};
\draw[->] (.2,-.3) -- (1.8,-.3);
\node at (1,-.5) {axis $m$};
\draw[->] (-.3,.2) -- (-.3,1.8);
\node[rotate=90] at (-.5,1) {axis $k$};
\node at (.5,2.5) {$\beta$};
\node[rotate=90] at (.52,2.23) {$=$};
\node at (1.5,2.5) {$\alpha$};
\node[rotate=90] at (1.52,2.23) {$=$};
\node[anchor=west] at (2,.5) {$= \Pic(v)$};
\node[anchor=west] at (2,1.5) {$= \Pic(u)$};
\node[anchor=east] at (-1,1) {$\Pic(t) =$};
\end{tikzpicture}
\end{array}
\label{eq:pic2}
\end{align}
Then
\begin{align*}
\alpha &= \gamma_{11} \crc{n-k} \gamma_{10} &
\Pic(u) &= \gamma_{11} \crc{n-m} \gamma_{01}\\
\beta &= \gamma_{01} \crc{n-k} \gamma_{00} &
\Pic(v) &= \gamma_{10} \crc{n-m} \gamma_{00}
\end{align*}
and
\begin{align}
\Pic(t) &= \alpha \crc{n-m} \beta = (\gamma_{11} \crc{n-k} \gamma_{10}) \crc{n-m} (\gamma_{01} \crc{n-k} \gamma_{00})\label{eq:picexch1}\\ 
\Pic(t) &= \Pic(u)\crc{n-k}\Pic(v) = (\gamma_{11} \crc{n-m} \gamma_{01}) \crc{n-k} (\gamma_{10} \crc{n-m} \gamma_{00})\label{eq:picexch2}
\end{align}
The term $w=\id^{k-1}\src^k u$ must be $\EEst$-equivalent to $\id^{k-1}\tgt^k v$ for $u\crc{n-k}v$ to be meaningful. Then the interface between $\Pic(u)$ and $\Pic(v)$ is $\Pic(t)\s k{1/2} = \id^{k-1}\src^k\Pic(u)=\id^{k-1}\tgt^k\Pic(v)$, and this is a picture, namely $\Pic(w)$. The interface between $\alpha$ and $\beta$, namely $\Pic(t)\s m{1/2} = \id^{m-1}\src^m\alpha=\id^{m-1}\tgt^m\beta$, is also a picture $\Pic(z)$ by assumption.

It is not hard to see that for $1\le i\le j$, $(\id^j\alpha)\s ix = \id^{j-1}\alpha$. Briefly,
\begin{align*}
& (\id^j\alpha)\s i(x_i)(\seq x1{i-1},\seq x{i+1}n)\\
&= (\id^j\alpha)(\seq x1{i-1},x_i,\seq x{i+1}n)\\
&= (\id^{j-i}\alpha)(\seq x{i+1}n)\\
&= (\id^{j-1}\alpha)(\seq x1{i-1},\seq x{i+1}n).
\end{align*}
Using this fact, we can show that $\Pic(t)\s{m,k}{1/2,1/2}$ is a picture. This is represented by the point at the intersection of the central vertical and horizontal lines in \eqref{eq:pic2}. If $1\le m < k$, then
\begin{align*}
\Pic(t)\s{m,k}{1/2,1/2}
&= \Pic(t)\s k{1/2}\s m{1/2}
= \Pic(w)\s m{1/2}\\
&= (\id^{k-1}\Pic(\src^{k-1}w))\s m{1/2}\\
&= \id^{k-2}\Pic(\src^{k-1}w)
= \Pic(\id^{k-2}\src^{k-1}w).
\end{align*}
Similarly, if $1\le k < m$,
\begin{align*}
\Pic(t)\s{m,k}{1/2,1/2}
&= \Pic(t)\s m{1/2}\s k{1/2}
= \Pic(z)\s m{1/2}\\
&= (\id^{m-1}\Pic(\src^{m-1}z))\s m{1/2}\\
&= \id^{m-2}\Pic(\src^{m-1}z)
= \Pic(\id^{m-2}\src^{m-1}z).
\end{align*}
In either case, $\Pic(t)\s{m,k}{1/2,1/2}$ is a picture of dimension $n-2$. This is the interface between the $(n-1)$-cells forming the interfaces between adjacent $\gamma_{ij}$. By the induction hypothesis on dimension, the interface between $\gamma_{00}$ and $\gamma_{01}$ and the interface between $\gamma_{10}$ and $\gamma_{11}$ are pictures. Likewise, the interface between $\gamma_{00}$ and $\gamma_{10}$ and the interface between $\gamma_{01}$ and $\gamma_{11}$ are pictures. By the induction hypothesis on term structure, $\gamma_{00},\gamma_{01}$, $\gamma_{10}$, and $\gamma_{11}$ are pictures. Thus $\alpha=\gamma_{11}\crc{n-k}\gamma_{10}$ and $\beta=\gamma_{01}\crc{n-k}\gamma_{00}$ are pictures. 

Note that in this case, $t$ can be represented as a composition in two ways \eqref{eq:picexch1} and \eqref{eq:picexch2}, and they are provably equivalent by exchange.
\exqed
\end{proof}

\begin{theorem}[Completeness over homotopy models]
\label{thm:completeness}
$\EEst\vdash t_1=t_2$ if and only if $t_1$ and $t_2$ represent the same cell in all homotopy models with globular maps $s,t$.
\end{theorem}
\begin{proof}
We have shown that $\theory$ is sound over homotopy models, so $\EEst\vdash t_1=t_2$ implies $t_1$ and $t_2$ represent the same cell in all interpretations in homotopy models. For the converse, we show that $\Pic(t_1)=\Pic(t_2)$ implies $\EEst\vdash t_1=t_2$. Suppose $\Pic(t_1)=\Pic(t_2)$ for $t_1$ and $t_2$ in irreducible form. Let $\dim t_1=\dim t_2=n$ and pick the parameters $\seq k1n$ of the representation of Theorem \ref{thm:rep} sufficiently large to accommodate both terms.

The proof proceeds by induction on term structure and on dimension. Assume that the theorem holds for dimensions less than $n$ and subterms of $t_1$ and $t_2$.

If $t=a$ for $a\in X_n$, then $\Pic(t_1)(1/2,\ldots,1/2)=a$. In that case $t_2=a$ as well, because $\Pic(u\crc mv)$ or $\Pic(\id u)$ for any composition or identity would have an element of $X_m$ for some $m<n$ at that position.

If $t_1=\id u_1$ and $t_2=\id u_2$, then $\Pic(t_1)=\id\Pic(u_1)$ and $\Pic(t_2)=\id\Pic(u_2)$, thus $\Pic(u_1)=\Pic(u_2)$. By the induction hypothesis on dimension, $\EEst\vdash u_1=u_2$, therefore $\EEst\vdash \id u_1=\id u_2$.  

If $t_1=\id w$ and $t_2=u \crc{n-m} v$, then $\EEst\vdash\id w\crc{n-1}\id w$, so this case reverts to the next case $t_1=u_1\crc{n-k}v_1$ and $t_2=u_2\crc{n-m}v_2$.

If $t_1=u_1\crc{n-k}v_1$ and $t_2=u_2\crc{n-m}v_2$ and $k=m$, then $\Pic(u_1)=\Pic(u_2)$ and $\Pic(v_1)=\Pic(v_2)$. By the induction hypothesis on term structure, $\EEst\vdash u_1=u_2$ and $\EEst\vdash v_1=v_2$, therefore $\EEst\vdash u_1\crc{n-m}v_1 = u_2\crc{n-m}v_2$.

If $t_1=u_1\crc{n-m}v_1$ and $t_2=u_2\crc{n-k}v_2$ and $k\ne m$, then as in \eqref{eq:pic2}, $\Pic(t_1)=\Pic(t_2)$ can be decomposed into quadrants $\gamma_{ij}$, $i,j\in\{0,1\}$, with
\begin{align*}
\gamma_{01}\s{m,k}{x_m,x_k} &= \Pic(t)\s{m,k}{x_m/2,(x_k+1)/2}\\
\gamma_{11}\s{m,k}{x_m,x_k} &= \Pic(t)\s{m,k}{(x_m+1)/2,(x_k+1)/2}\\
\gamma_{00}\s{m,k}{x_m,x_k} &= \Pic(t)\s{m,k}{x_m/2,x_k/2}\\
\gamma_{10}\s{m,k}{x_m,x_k} &= \Pic(t)\s{m,k}{(x_m+1)/2,x_k/2},
\end{align*}
and by the same argument as in Lemma \ref{lem:split}, these are pictures. By \eqref{eq:picexch1} and \eqref{eq:picexch2}, $t_1$ and $t_2$ are provably equivalent by exchange.
\exqed
\end{proof}

\subsection{Completeness over  Finite CW-Complexes}
\label{sec:CW}

As mentioned, the reader might be somewhat skeptical of the significance of the completeness result of \S\ref{sec:completenesshomotopy}, as the homotopy model constructed in \S\ref{sec:pictures} and used in that proof has a trivial topological structure. One might wish for a stronger result showing completeness over a more conventional class of topological models such as compact Hausdorff spaces or CW-complexes. We show that both are true. This is mainly due to the fact that the picture model is a retract of a CW-complex. 
The complex is obtained by gluing in one closed cube per generator, dimension by
dimension, with the attaching maps prescribed by sources and targets. Collapsing
each cube back to its generator sends the interpretation of a term in the complex
to its picture, so two terms with different pictures are already distinguished
there. A finite subcomplex suffices, and finite CW-complexes are
compact Hausdorff, so the claim holds for that class too.  
The construction is given in full in the proof of Theorem \ref{thm:wellbehaved}.

\begin{theorem}[Completeness over finite CW-complexes]
\label{thm:wellbehaved}
$\EEst\vdash t_1=t_2$ if $t_1=t_2$ in all finite CW-complexes with globular maps $s,t$.
\end{theorem}
\begin{proof}
Let $\FF$ be the free model generated by $(X,\dim,s,t)$.
We will construct a CW-complex $H$ that will serve as the ambient space and simultaneously a map $C$ defined on $\FF$-meaningful terms such that if $\dim t=n$, then $C(t):I^n\to H$ is a cell. The construction is by induction on dimension.
The $0$-skeleton $H_0$ of $H$ is just the discrete space $X_0$, and $C(a)=a$ for $a\in H_0$.

Assume that the $(n-1)$-skeleton $H_{n-1}$ of $H$ and $C(t):I^m\to H_{n-1}$ have been constructed for all $t\in T_{n-1}$ of dimension $m$. For each $a\in X_n$, let $U_a$ be a fresh copy of the closed unit hypercube $I^n$. Define $C(a):I^n\to H_{n-1} \uplus \biguplus_{a\in X_n} U_a$ as follows.
\begin{align*}
& C(a)(\seq x1n)\\
&= \begin{cases}
(\seq x1n)\in U_a & \text{if $0 < x_i < 1$ for all $1\le i\le n$}\\
C(s(a))(\seq x2n) & \text{if $x_j=0$ and $0 < x_i < 1$ for $j+1\le i\le n$}\\
C(t(a))(\seq x2n) & \text{if $x_j=1$ and $0 < x_i < 1$ for $j+1\le i\le n$.}
\end{cases}
\end{align*}
The latter two cases are defined on the boundary of $I^n$ and take values in $H_{n-1}$, thus when applied to the boundary of $U_a$ specify attaching maps.

Now let $H_n = (H_{n-1}\uplus\biguplus_{a\in X_n} U_a)/{\sim}$, where $\sim$ is the equivalence relation determined by the attaching maps. Modulo $\sim$, $C(a)$ has type $I^n\to H_n$. We claim that with this type, $C(a)$ is a cell. To show this, we use Lemma \ref{lem:cellprops}. We must show that if some $x_j\in\{0,1\}$, then $C(a)(\seq x1n)$ is independent of $\seq x1{j-1}$. Let $j$ be the largest index such that $x_j\in\{0,1\}$. If $j=1$, there is nothing to prove. Otherwise, if $x_j=0$, by the second clause in the definition, $C(a)(\seq x1n) = C(s(a))(\seq x2n)$. By the induction hypothesis on dimension, $C(s(a))(\seq x2n)$ is independent of $\seq x2{j-1}$, and the definition does not involve $x_1$, so $C(a)(\seq x1n)$ is independent of $x_1$ as well. A similar argument holds if $x_j=1$ using the third clause in the definition.

Having constructed the $n$-skeleton $H_n$ and the $n$-cells $C(a)$ for $a\in X_n$, we can now extend $C$ to give a cell $I^m\to H_n$ to each term $t\in T_n$ of dimension $m$ by applying the cell operations inductively; that is, 
\begin{align*}
& C(\src t) = \src C(t) = C(t)\s10\\
& C(\id t) = \id C(t) = \lam{(\seq x0n)}{C(t)(\seq x1n)}\\
& C(\tgt t) = \tgt C(t) = C(t)\s11\\
& C(t_1 \crc m t_2) = C(t_1)\crc m C(t_2),
\end{align*}
where in the last, $0\le m<n$, $\src^{n-m} C(t_1)=\tgt^{n-m} C(t_2)$, and
\begin{align*}
(\alpha\crc m\beta)\s{n-m}x &= \begin{cases}
\alpha\s{n-m}{2x} & \text{if $0\le x\le 1/2$}\\
\beta\s{n-m}{2x-1} & \text{if $1/2\le x\le 1$.}
\end{cases}
\end{align*}
These operations map cells to cells, so each $C(t):I^m\to H_n$ is a cell, where $t\in T_n$ with $\dim t=m$.

Now we construct another map $h:H\to X$ as follows: for each $a\in H_0=X_0$, let $h(a)=a$. For $a\in X_n$ with $n\ge 1$, let $h(x)=a$ for all $x$ in the \emph{interior} of $U_a$, that is, for $x=(\seq x1n)$ where $0<x_i<1$ for all $1\le i\le n$. The value of $h(x)$ for $x$ on the boundary of $a$ will be some $b\in X$ of lower dimension as determined by the attaching maps.

We claim that $\Pic(t) = h\circ C(t)$ for all $t\in T$. Thus for $t$ of dimension $n$,
\begin{align*}
\begin{tikzpicture}[->, >=stealth', auto]
\small
\node (O) {};
\node (T) at (90:1cm) {$I^n$};
\node (W) at (210:1cm) {$H$};
\node (E) at (330:1cm) {$X$};
\path (T) edge node[swap] {$C(t)$} (W);
\path (T) edge node {$\Pic(t)$} (E);
\path (W) edge node[swap] {$h$} (E);
\end{tikzpicture}
\end{align*}
We recall the definition of $\Pic(t)$:
\begin{align*}
& \Pic(a)(\seq x1n)\\
&= \begin{cases}
a, & \text{if $0<x_i<1$ for all $1\le i\le n$}\\
\Pic(s(a))(\seq x2n), & \text{if $x_j=0$ and $0<x_i<1$ for $j+1\le i\le n$}\\
\Pic(t(a))(\seq x2n), & \text{if $x_j=1$ and $0<x_i<1$ for $j+1\le i\le n$.}
\end{cases}
\end{align*}
For $a\in H_0=X_0$, we have $\Pic(a) = a = h(a) = h(C(a))$. For $a\in X_n$ with $n\ge 1$, if if $0<x_i<1$ for all $1\le i\le n$, we have $\Pic(a)(\seq x1n) = a = h(C(a)(\seq x1n))$, as $C(a)(\seq x1n)$ is in the interior of $U_a$. If $x_j=0$ and $0<x_i<1$ for $j+1\le i\le n$, then by the induction hypothesis on dimension,
\begin{align*}
\Pic(a)(\seq x1n) &= \Pic(s(a))(\seq x2n)\\
&= h(C(s(a))(\seq x2n))
= h(C(a))(\seq x1n).
\end{align*}
A similar argument holds if $x_j=1$.

Since $\theory$ is sound over homotopy models, $\EEst\vdash t_1=t_2$ implies $C(t_1)=C(t_2)$. Conversely, if $C(t_1)=C(t_2)$, then $\Pic(t_1)=h(C(t_1))=h(C(t_2))=\Pic(t_2)$, and by Theorem \ref{thm:completeness}, $\EEst\vdash t_1=t_2$. Thus $\theory$ is sound and complete over CW-complexes. 
Although $H$ itself may be infinite, each interpretation lands in a finite subcomplex. Let $X'$ be the set of generators occurring in $t_1$ or $t_2$, closed under $s$ and $t$; this set is finite, since $s(a)$ and $t(a)$ involve only generators of dimension strictly less than $\dim a$. Let $K\subs H$ be the union of the closed cells $\overline{U_a}$ for $a\in X'$. Because $X'$ is closed under $s$ and $t$, the attaching map of each such cube lands in cells of $K$, so $K$ is a finite subcomplex of $H$. The image of $C(t_i)$ lies in $K$, and $K$ carries the subspace topology, so corestriction gives a cell $C(t_i):I^{n}\to K$. Hence $t_1$ and $t_2$ are already distinguished in the finite CW-complex $K$.
\exqed
\end{proof}

\subsection{Decision Procedure}
\label{sec:decproc}

The decision procedure falls out of the construction in \S\ref{sec:completeness}. By Theorem
\ref{thm:completeness}, two terms are provably equal exactly when they have the
same picture, so deciding the theory reduces to constructing two pictures and
comparing them, and by \S\ref{sec:discrete} a picture is determined by its values
on a regular dyadic grid, so each is a finite array. The grid resolution needed
is governed by how deeply compositions nest, which gives the bound below. 
Two consequences are worth mentioning. 
First, the procedure decides more than
equality: because a composite $u\crct nm v$ is meaningful exactly when
$\EEst\vdash\src^{n-m}u=\tgt^{n-m}v$ (\S\ref{sec:semantics}), the same algorithm
decides whether a term denotes a cell at all. 
Second, we have no matching lower
bound, and the array representation is almost certainly wasteful, so the gap
between this upper bound and the true complexity of the word problem is open.
 
\begin{theorem}
\label{thm:decproc}
Equality of terms $t_1$ and $t_2$ is decidable in time $O(2^{\max h(t_1),h(t_2)})$, where $h(t)=\sum_{m=0}^{n-1}\hgt mt$ and $\hgt mt$ is the $\crc m$-height of $t$ in irreducible form as defined in \S\ref{sec:discrete}.
\end{theorem}
\begin{proof}
Given two terms of dimension $n$ in irreducible form, by Theorem \ref{thm:completeness} it suffices to construct their pictures and check whether they are equal. A finite representation of each picture as described in Theorem \ref{thm:rep} can be constructed inductively and requires a multidimensional array of size
\begin{align*}
2^{k_1}\times\cdots\times 2^{k_n} = 2^{k_1+\cdots+k_n}, 
\end{align*}
where $k_m=\hgt{n-m}t$ is the $\crc m$-height of $t$ as defined in \S\ref{sec:discrete}.
\exqed
\end{proof}

\section{Soundness over Presheaf Models}
\label{sec:presheaf}

In presheaf models of cubical type theory \cite{Bezem14,Cohen16,Awodey16}, a \emph{cubical set} is a functor $\square^{\mathrm{op}}\to\Set$, where $\square$ is a category capturing in an abstract way the properties of finite-dimensional hypercubes. There is some variation in the encodings found in the literature, but one typically sees constructions for face maps, degeneracies, and connections, the last being used in conjunction with Kan space-filling operations to synthesize paths to witness equalities.

The prototypical example of such a presheaf model, and the primary motivating example, are the homotopy models. These functors map the object $\square_n$ to the set of continuous maps $\alpha:I^n\to X$, where $I$ is the unit real interval and $X$ is an arbitrary topological space.

In the traditional approach, compositions are normally defined using de Morgan connections in conjunction with Kan space-filling operations. However, the construction of compositions using de Morgan connections violates globularity. 
One would like a more symmetric construction that respects globularity. This is possible, as we show in this section. We introduce a variant in which compositions are internalized in the base category.

Like $\square$, our base category $\boxdot$ is a strict monoidal category with objects $[n]$ for $n\in\naturals$, strict tensor product $[m]\otimes[n]=[m+n]$, monoidal unit $[0]$ which is a terminal object, primitive face maps $\delta_0,\delta_1:[0]\to[1]$, $\delta_0\ne\delta_1$, and degeneracy map $\sigma:[1]\to[0]$.

Unlike $\square$, $\boxdot$ will always have at least countably many other morphisms $\delta_t:[0]\to[1]$ besides $\delta_0$ and $\delta_1$. One can view $\Hom([0],[1])$ as a prototypical interval, linearly ordered with least element $\delta_0$ and greatest element $\delta_1$. As $[0]$ is terminal, $\sigma\circ\delta_t = 1_{[0]}$ for all $\delta_t:[0]\to[1]$.

A morphism of the form $\delta_t\circ\sigma:[1]\to[1]$ is called a \emph{reflexor}. These morphisms correspond to the reflexivity paths $\refl a$ of identity types in Martin-L\"of type theory. There will be other morphisms $[1]\to[1]$ as well.

Let $\boxdot_1$ be the induced subcategory of $\boxdot$ on objects $[0]$ and $[1]$. We assume that all morphisms of $\boxdot$ are generated by morphisms of $\boxdot_1$ under composition and tensor product.

Higher face and degeneracy maps are defined via the monoidal structure. Define the endofunctor 
\begin{align}
& H^n_m:\boxdot\to\boxdot & H^n_m[d] &= [n+d] & H^n_mf &= \ident{m}\otimes f\otimes\ident{n-m}:[n+d]\to[n+e]\label{eq:Hdef}
\end{align}
for $0\le m\le n$ and $f:[d]\to[e]$. Note that $H^0_0$ is just the identity functor. The higher face and degeneracy maps are
\begin{align*}
& H^{n-1}_m\delta_t:[n-1]\to[n] && H^{n-1}_m\sigma:[n]\to[n-1].
\end{align*}

These maps commute modulo appropriate reindexing:
\begin{lemma}
\label{lem:commute}
If $f:[d]\to[e]$, $g:[i]\to[j]$ and $0\le k\le m\le n$, then
\begin{align*}
H^{n+j}_k f\circ H^{n+d}_{m+d} g = H^{n+e}_{m+e} g\circ H^{n+i}_k f:[n+d+i]\to[n+e+j].
\end{align*}
\end{lemma}
\begin{proof}
This is a consequence of the monoidal structure.
\begin{align*}
& H^{n+j}_k f\circ H^{n+d}_{m+d} g\\
&= (1_{[k]}\otimes f\otimes 1_{[n+j-k]})\circ(1_{[m+d]}\otimes g\otimes 1_{[n-m]})\\
&= (1_{[k]}\otimes f\otimes 1_{[m-k]}\otimes 1_{[j]}\otimes 1_{[n-m]})\circ(1_{[k]}\otimes 1_{[d]}\otimes 1_{[m-k]}\otimes g\otimes 1_{[n-m]})\\
&= 1_{[k]}\otimes f\otimes 1_{[m-k]}\otimes g\otimes 1_{[n-m]}\\
&= (1_{[k]}\otimes 1_{[e]}\otimes 1_{[m-k]}\otimes g\otimes 1_{[n-m]})\circ(1_{[k]}\otimes f\otimes 1_{[m-k]}\otimes 1_{[i]}\otimes 1_{[n-m]})\\
&= (1_{[m+e]}\otimes g\otimes 1_{[n-m]})\circ(1_{[k]}\otimes f\otimes 1_{[n+i-k]})\\
&= H^{n+e}_{m+e} g\circ H^{n+i}_k f.
\extagqed
\end{align*}
\end{proof} 
Thus for $0\le k\le m\le n$,
%\begin{align*}
$H^{n+1}_k \delta_s\circ H^{n}_{m} \delta_t = H^{n+1}_{m+1} \delta_t\circ H^{n}_k \delta_s:[n]\to[n+2]
$. %\end{align*}
Also, every $H^n_m\sigma$ is a split epi and every $H^n_m\delta_t$ a split mono:
\begin{align*}
H^n_m\sigma\circ H^n_m\delta_t &= H^n_m(\sigma\circ\delta_t) = H^n_m\ident 0 = 1_{[n]}.
\end{align*}

\subsection{Boundary preservation}
\label{sec:bdrypres}

We wish to limit models to those satisfying the globular axioms, as appropriate for higher identity types. A presheaf $F$ on $\boxdot$ is said to \emph{preserve boundaries} if it satisfies the following
\paragraph{Boundary preservation condition:}
For all $m\ge 0$, all $b\in\{0,1\}$, and all $s,t$,
\begin{align}
F(\delta_s\otimes\delta_b\otimes\ident m) = F(\delta_t\otimes\delta_b\otimes\ident m).\label{eq:bdrypres}
\end{align}

\begin{theorem}
\label{thm:bdrypres}
If $F$ preserves boundaries, then for all $h:[k]\to[m]$, $\seq s1n$, $\seq t1n$ and $b\in\{0,1\}$,
\begin{align*}
F(\delta_{s_1}\otimes\cdots\otimes\delta_{s_n}\otimes\delta_b\otimes h) = F(\delta_{t_1}\otimes\cdots\otimes\delta_{t_n}\otimes\delta_b\otimes h).
\end{align*}
\end{theorem}
\begin{proof}
For any $h:[k]\to[m]$, $\seq s1n$, and $b\in\{0,1\}$,
\begin{align*}
& \delta_{s_1}\otimes\cdots\otimes\delta_{s_n}\otimes\delta_b\otimes h\\
&= ((\delta_{s_1}\otimes\cdots\otimes\delta_{s_{n-1}})\circ\ident 0\circ\ident 0)\otimes(\ident{2}\circ(\delta_{s_n}\otimes\delta_b)\circ\ident 0)\otimes(\ident m\circ \ident m\circ h)\\
&= (\delta_{s_1}\otimes\cdots\otimes\delta_{s_{n-1}}\otimes\ident{2}\otimes\ident{m})\circ(\ident 0\otimes\delta_{s_n}\otimes\delta_b\otimes\ident m)\circ(\ident 0 \otimes \ident 0 \otimes h)\\
&= (\delta_{s_1}\otimes\cdots\otimes\delta_{s_{n-1}}\otimes\ident{m+2})\circ(\delta_{s_n}\otimes\delta_b\otimes\ident m)\circ h,
\end{align*}
therefore by \eqref{eq:bdrypres},
\begin{align*}
& F(\delta_{s_1}\otimes\cdots\otimes\delta_{s_n}\otimes\delta_b\otimes h)\\
&= F((\delta_{s_1}\otimes\cdots\otimes\delta_{s_{n-1}}\otimes\ident{m+2})\circ(\delta_{s_n}\otimes\delta_b\otimes\ident m)\circ h)\\
&= F(h) \circ F(\delta_{s_n}\otimes\delta_b\otimes\ident m) \circ F(\delta_{s_1}\otimes\cdots\otimes\delta_{s_{n-1}}\otimes\ident{m+2})\\
&= F(h) \circ F(\delta_{b}\otimes\delta_b\otimes\ident m) \circ F(\delta_{s_1}\otimes\cdots\otimes\delta_{s_{n-1}}\otimes\ident{m+2})\\
&= F(\delta_{s_1}\otimes\cdots\otimes\delta_{s_{n-1}}\otimes\delta_b\otimes\delta_b\otimes h).
\end{align*}
It follows by induction that
\begin{align*}
F(\delta_{s_1}\otimes\cdots\otimes\delta_{s_n}\otimes\delta_b\otimes h) = 
F(\underbrace{\delta_b\otimes\cdots\otimes\delta_b}_{n+1}\otimes h),
\end{align*}
and the result follows from transitivity.
\exqed
\end{proof}

We will show that the globular axioms are a consequence.

\subsection{Composition}
\label{sec:composition}

Traditionally, composites are produced by Kan composition structure on each
type, with the filling operation derived from it using a \emph{connection}
\cite{Cohen16,OrtonPitts18}. Connections are extra degeneracies in the sense
of Brown and Higgins \cite{BrownHiggins81b}, and by themselves they do not
produce composites.

Here we take an alternative approach, which internalizes composition in the structure of our base category. We postulate
\paragraph{Composition axioms:}
\begin{itemize}
\item
The object $[1]$ is a pushout of $\delta_1:[0]\to[1]$ and $\delta_0:[0]\to[1]$ completed by morphisms $L:[1]\to[1]$ and $R:[1]\to[1]$ with $L\circ\delta_1=R\circ\delta_0$.
\item
$L\circ\delta_0=\delta_0$ and $R\circ\delta_1=\delta_1$.
\end{itemize}
\begin{equation}
\begin{tikzpicture}[->, >=stealth', node distance=12mm, auto, inner sep=1.5pt]
\small
  \node (N) {$[0]$};
  \node (W) [below left of=N] {$[1]$};
  \node (E) [below right of=N] {$[1]$};
  \node (S) [below right of=W] {$[1]$};
  \path (N) edge node[swap] {$\delta_1$} (W);
  \path (N) edge node {$\delta_0$} (E);
  \path (W) edge node[swap] {$L$} (S);
  \path (E) edge node {$R$} (S);
\end{tikzpicture}
\label{eq:pushout}
\end{equation}

We assume that the image of \eqref{eq:pushout} under $H^n_m$ is a pushout as well. This structure will enable us to define composition natively in presheaf models without connections or Kan filling. 
 
Intuitively, we can view \eqref{eq:pushout} as gluing two intervals end-to-end.
Let us write $\delta_{1/2}:[0]\to[1]$ for the morphism $L\circ\delta_1=R\circ\delta_0$. Since $1_{[1]}\circ\delta_1=\delta_1=\delta_1\circ\sigma\circ\delta_0$, by the universality property of pushouts, there is a unique map $L':[1]\to[1]$ such that $L'\circ L=1_{[1]}$ and $L'\circ R=\delta_1\circ\sigma$. Similarly, since $1_{[1]}\circ\delta_0=\delta_0=\delta_0\circ\sigma\circ\delta_1$, there is a unique map $R':[1]\to[1]$ such that $R'\circ R=1_{[1]}$ and $R'\circ L=\delta_0\circ\sigma$. 
\begin{equation}
\begin{tikzpicture}[->, >=stealth', node distance=11mm, inner sep=1.5pt,
                    bend angle=20, auto, baseline=(current bounding box.center)]
\small
  \node (N) {$[0]$};
  \node (W) [below left of=N] {$[1]$};
  \node (E) [below right of=N] {$[1]$};
  \node (S) [below right of=W] {$[1]$};
  \path (N) edge node[swap] {$\delta_1$} (W);
  \path (N) edge node {$\delta_0$} (E);
  \path (W) edge node[swap] {$L$} (S);
  \path (E) edge node {$R$} (S);
  \node (SS) [below of=S] {$[1]$};
  \path (S) edge node[pos=.4] {$L'$} (SS);
  \path (W) edge[bend right] node[swap] {$1_{[1]}$} (SS);
  \path (E) edge[bend left] node {$\delta_1\circ\sigma$} (SS);
\end{tikzpicture}
\hspace{1cm}
\begin{tikzpicture}[->, >=stealth', node distance=11mm, inner sep=1.5pt,
                    bend angle=20, auto, baseline=(current bounding box.center)]
\small
  \node (N) {$[0]$};
  \node (W) [below left of=N] {$[1]$};
  \node (E) [below right of=N] {$[1]$};
  \node (S) [below right of=W] {$[1]$};
  \path (N) edge node[swap] {$\delta_1$} (W);
  \path (N) edge node {$\delta_0$} (E);
  \path (W) edge node[swap] {$L$} (S);
  \path (E) edge node {$R$} (S);
  \node (SS) [below of=S] {$[1]$};
  \path (S) edge node[pos=.4] {$R'$} (SS);
  \path (W) edge[bend right] node[swap] {$\delta_0\circ\sigma$} (SS);
  \path (E) edge[bend left] node {$1_{[1]}$} (SS);
\end{tikzpicture}
\label{eq:pushout2}
\end{equation}
It follows that $L$ and $R$ are split monos and $L'$ and $R'$ are split epis, and that $L$ and $R$ are jointly epic. This implies that $\delta_{1/2}$ cannot be $\delta_0$ or $\delta_1$, otherwise we would have $\delta_0=\delta_1$.
In fact, one can show that $\Hom([0],[1])$ must contain countably many distinct morphisms $\delta_t$, one for each dyadic rational $t\in \set{a2^{-n}}{n\ge 0,\ 0\le a\le 2^n}$, obtained by inductively defining $\delta_{t/2}=L\circ\delta_t$ and $\delta_{(t+1)/2}=R\circ\delta_t$. This set comprises $\Hom([0],[1])$ in the minimal base category.

Concretely, it is helpful to think of the object $[1]$ as the unit real interval $I$, the unit rational interval $I\cap\rationals$, or the set $\set{a2^{-n}}{n\ge 0,\ 0\le a\le 2^n}$ itself, and $L(t) = t/2$ and $R(t)=(t+1)/2$, in which case
\begin{align*}
L'(t) &= \begin{cases}
2t, & \text{if $0\le t\le 1/2$},\\
1, & \text{if $1/2\le t\le 1$}
\end{cases}
&
R'(t) &= \begin{cases}
0, & \text{if $0\le t\le 1/2$},\\
2t-1, & \text{if $1/2\le t\le 1$}.
\end{cases}
\end{align*}

The pushout structure entails the existence of more elements of $\Hom([1],[1])$ than just $L, R, L', R'$. For each pair of maps $U,V:[1]\to[1]$ such that $U\circ\delta_1=V\circ\delta_0$, there exists a unique mediating morphism $U\med V:[1]\to[1]$ such that $U=(U\med V)\circ L$ and $V=(U\med V)\circ R$. Thus $L' = \ident 1\med(\delta_1\circ\sigma)$ and $R' = (\delta_0\circ\sigma)\med\ident 1$. A  characterization of $\boxdot_1$ is given in Appendix \ref{apx:boxdot1}.

\subsection{Soundness of the axioms}

In this section we show that $\theory$ is sound for presheaf models over our base category $\boxdot$ satisfying the boundary preservation condition \eqref{eq:bdrypres}. In particular, these structures are globular, therefore model higher identity types.

Let $F:\boxdot\op\to\Set$ be a pullback-preserving presheaf satisfying the boundary preservation condition \eqref{eq:bdrypres}. The canonical example is
\begin{align*}
F[n] &= \{\text{globular continuous maps $\alpha:I^n\to X$}\}
\end{align*}
as described in \S\ref{sec:homotopymodels}, where $I=[0,1]$ is the real unit interval and $X$ is an arbitrary topological space. The assumption of preservation of pullbacks is needed for the composition structure of \S\ref{sec:composition} to carry over to \Set.

Let $H^n_m$ be the functors defined in \eqref{eq:Hdef}. Define polymorphic source and target maps for $n\ge 1$:
\begin{align*}
\src &= FH^{n-1}_0\delta_0:F[n]\to F[n-1] & \tgt &= FH^{n-1}_0\delta_1:F[n]\to F[n-1].
\end{align*}
Define a polymorphic identity map for $n\ge 0$:
\begin{align*}
\id = FH^{n}_0\sigma:F[n]\to F[n+1].
\end{align*}

There are also partial composition operators $\crct n{n-m}:F[n]\times F[n]\to F[n]$ for $1\le m\le n$ modeling the composition of $n$-cells along their $m$-th $(n-1)$-dimensional faces. We claim that if $1\le m\le n$, $\alpha,\beta\in F[n]$, and $\src^m\beta=\tgt^m\alpha$, then $\beta\crct n{n-m}\alpha$ exists. Since the image of \eqref{eq:pushout} under $H^{n-1}_{m-1}$ is a pushout in $\boxdot$, and since $F$ is assumed to preserve pullbacks, the image of that diagram under $F$ is a pullback in \Set, thus in one-to-one correspondence with
\begin{align*}
\set{(\alpha,\beta)}{FH^{n-1}_{m-1}\delta_1(\alpha) = FH^{n-1}_{m-1}\delta_0(\beta)} = \set{(\alpha,\beta)}{\src^m\beta=\tgt^m\alpha}.
\end{align*}
If $\gamma$ is the unique element of $F[n]$ corresponding to the pair $(\alpha,\beta)$ under this correspondence, then
\begin{align*}
\alpha &= FH^{n-1}_{m-1} L(\gamma) & \beta &= FH^{n-1}_{m-1} R(\gamma)
\end{align*}
and we take $\beta\crct n{n-m}\alpha=\gamma$.

\begin{theorem}
\label{thm:soundness}
The theory $\theory$ is sound for presheaf models over $\boxdot$.
\end{theorem}
\begin{proof}
Axioms (i) and (ii) are
\begin{align*}
\src(\beta\crct n{n-k}\alpha) &= \src \beta\crct{n-1}{n-k}\src\alpha
& \tgt(\beta\crct n{n-k}\alpha) &= \tgt \beta\crct{n-1}{n-k}\tgt\alpha
\end{align*}
for $\alpha,\beta\in F[n]$ and $k > 1$. For (i), let us first check that $\beta\crct n{n-k}\alpha$ is defined iff $\src \beta\crct{n-1}{n-k}\src\alpha$ is defined. The conditions are $\src^k\beta=\tgt^k\alpha$ and $\src^{k-1}(\src\beta)=\tgt^{k-1}(\tgt\alpha)$, respectively, which are clearly equivalent. Now $\beta\crct n{n-k}\alpha$ is the unique $\gamma\in F[n]$ such that
\begin{align*}
\alpha &= FH^{n-1}_{k-1} L(\gamma) & \beta &= FH^{n-1}_{k-1} R(\gamma), 
\end{align*}
and $\src\beta\crct{n-1}{n-k}\src\alpha$ is the unique $\gamma'\in F[n-1]$ such that
\begin{align*}
\src\alpha &= FH^{n-2}_{k-2} L(\gamma') & \src\beta &= FH^{n-2}_{k-2} R(\gamma'), 
\end{align*}
and we wish to show that $\src\gamma=\gamma'$. Since $\gamma'$ is unique, it suffices to show
\begin{align*}
\src\alpha &= FH^{n-2}_{k-2} L(\src\gamma) & \src\beta &= FH^{n-2}_{k-2} R(\src\gamma) 
\end{align*}
or in other words
\begin{align*}
FH^{n-1}_0\delta_0(FH^{n-1}_{k-1} L(\gamma)) &= FH^{n-2}_{k-2} L(FH^{n-1}_0\delta_0(\gamma))\\
FH^{n-1}_0\delta_0(FH^{n-1}_{k-1} R(\gamma)) &= FH^{n-2}_{k-2} R(FH^{n-1}_0\delta_0(\gamma)). 
\end{align*}
These equations follow from Lemma \ref{lem:commute}. For the left-hand equation,
\begin{align*}
FH^{n-1}_0\delta_0 \circ FH^{n-1}_{k-1} L 
&= F(H^{n-1}_{k-1} L \circ H^{n-1}_0\delta_0)\\
&= F(H^{n-1}_0\delta_0 \circ H^{n-2}_{k-2} L)
= FH^{n-2}_{k-2} L \circ FH^{n-1}_0\delta_0.
\end{align*}
and similarly for the right-hand equation. The argument for (ii) is symmetric.

Axioms (iii) and (iv) are
\begin{align*}
\src(\beta\crct n{n-1}\alpha) &= \src\alpha
& \tgt(\beta\crct n{n-1}\alpha) = \tgt\beta
\end{align*}
for $\alpha,\beta\in[n]$. For (iii), $\beta\crct n{n-1}\alpha$ is the unique $\gamma$ such that
\begin{align*}
\alpha &= FH^{n-1}_0 L(\gamma) & \beta &= FH^{n-1}_0 R(\gamma).
\end{align*}
We wish to show $\src\gamma = \src\alpha$ and $\tgt\gamma = \tgt\beta$, or in other words,
\begin{align*}
& FH^{n-1}_0\delta_0(\gamma) = FH^{n-1}_0\delta_0(FH^{n-1}_0 L(\gamma))\\
& FH^{n-1}_0\delta_1(\gamma) = FH^{n-1}_0\delta_1(FH^{n-1}_0 R(\gamma))
\end{align*}
But
\begin{align*}
FH^{n-1}_0\delta_0 \circ FH^{n-1}_0 L
&= F(H^{n-1}_0 L \circ H^{n-1}_0\delta_0)
= FH^{n-1}_0(L\circ\delta_0)
= FH^{n-1}_0\delta_0\\
FH^{n-1}_0\delta_1 \circ FH^{n-1}_0 R
&= F(H^{n-1}_0 R \circ H^{n-1}_0\delta_1)
= FH^{n-1}_0(R\circ\delta_1)
= FH^{n-1}_0\delta_1.
\end{align*}
The argument for (iv) is symmetric.

Axiom (v) is $\id(\beta\crct n{n-m}\alpha) = \id\beta\crct{n+1}{n-m}\id\alpha$ for $1 \le m \le n$ and $\alpha,\beta\in[n]$. Now $\beta\crct n{n-m}\alpha$ is the unique $\gamma\in F[n]$ such that
\begin{align*}
\alpha &= FH^{n-1}_{m-1} L(\gamma) & \beta &= FH^{n-1}_{m-1} R(\gamma), 
\end{align*}
and $\id\beta\crct{n+1}{n-m}\id\alpha$ is the unique $\gamma'\in F[n+1]$ such that
\begin{align*}
\id\alpha &= FH^{n}_{m} L(\gamma') & \id\beta &= FH^{n}_{m} R(\gamma').
\end{align*}
To show that $\id\gamma=\gamma'$, it suffices to show that
\begin{align*}
\id\alpha &= FH^{n}_{m} L(\id\gamma) & \id\beta &= FH^{n}_{m} R(\id\gamma).
\end{align*}
or in other words
\begin{align*}
FH^{n}_0\sigma(FH^{n-1}_{m-1} L(\gamma)) &= FH^{n}_{m} L(FH^{n}_0\sigma(\gamma))\\
FH^{n}_0\sigma(FH^{n-1}_{m-1} R(\gamma)) &= FH^{n}_{m} R(FH^{n}_0\sigma(\gamma)).
\end{align*}
Again, these follow from Lemma \ref{lem:commute}. For the left-hand equation, since $H^{n}_0 \sigma\circ H^{n}_{m} L = H^{n-1}_{m-1} L\circ H^{n}_0\sigma:[n+1]\to[n]$,
\begin{align*}
FH^{n}_0\sigma \circ FH^{n-1}_{m-1} L
&= F(H^{n-1}_{m-1} L \circ H^{n}_0\sigma)
= F(H^{n}_0\sigma \circ H^{n}_{m} L)
= FH^{n}_{m} L \circ FH^{n}_0\sigma.
\end{align*}
and similarly for the right-hand equation.

Axioms (vi) and (vii) are $\src\id\alpha = \alpha$ and $\tgt\id\alpha = \alpha$, respectively. For (vi),
\begin{align*}
\src\id\alpha = FH^n_0\delta_0(FH^{n}_0\sigma(\alpha))
= FH^{n}_0(\sigma\circ\delta_0)(\alpha)
= FH^{n}_0\ident 0(\alpha)
= 1_{F[n]}(\alpha)
= \alpha.
\end{align*}
The argument for (vii) is symmetric.

Axioms (viii) and (ix) are the globularity axioms $\src\tgt\alpha = \src\src\alpha$ and $\tgt\src \alpha = \tgt\tgt \alpha$, respectively. For (viii), by \eqref{eq:bdrypres}, $F(\delta_t\otimes\delta_b\otimes\ident m) = F(\delta_s\otimes\delta_b\otimes\ident m)$ for $b\in\{0,1\}$. Then
\begin{align*}
\src\tgt\alpha
&= FH^{n-2}_0\delta_0(FH^{n-1}_0\delta_1(\alpha))
= F(H^{n-1}_0\delta_1 \circ H^{n-2}_0\delta_0)(\alpha)\\
&= F(\delta_1\otimes\delta_0\otimes\ident{n-2})(\alpha)\\
&= F(\delta_0\otimes\delta_0\otimes\ident{n-2})(\alpha)\\
&= F(H^{n-1}_0\delta_0 \circ H^{n-2}_0\delta_0)(\alpha)
= \src\src \alpha.
\end{align*}
The argument for (ix) is similar.

Axiom (x) is $\id\alpha\crct n{n-1}\id\alpha = \id\alpha$ for $\id\alpha\in[n]$. The composition $\id\alpha\crct n{n-1}\id\alpha$ is defined, since $\src\id\alpha = \tgt\id\alpha = \alpha$, and the composition is the unique $\gamma$ such that
\begin{align*}
\id\alpha &= FH^{n-1}_0 L(\gamma) = FH^{n-1}_0 R(\gamma).
\end{align*}
We must show that $\gamma=\id\alpha$. Since $\gamma$ is unique, it suffices to show
\begin{align*}
\id\alpha &= FH^{n-1}_0 L(\id\alpha) = FH^{n-1}_0 R(\id\alpha),
\end{align*}
or in other words,
\begin{align*}
FH^{n-1}_0\sigma(\alpha) &= FH^{n-1}_0 L(FH^{n-1}_0\sigma(\alpha)) = FH^{n-1}_0 R(FH^{n-1}_0\sigma(\alpha)),
\end{align*}
This is true because
\begin{align*}
FH^{n-1}_0 L \circ FH^{n-1}_0\sigma
&= F(H^{n-1}_0\sigma \circ H^{n-1}_0 L)
= FH^{n-1}_0(\sigma \circ L)
= FH^{n-1}_0(\sigma)
\end{align*}
and similarly $FH^{n-1}_0 R \circ FH^{n-1}_0\sigma = FH^{n-1}_0(\sigma)$.

Finally, axiom (xi) is the exchange law
\begin{align*}
(\beta \crct n{n-k} \alpha) \crct n{n-m} (\phi \crct n{n-k} \psi) = (\beta \crct n{n-m} \phi) \crct n{n-k} (\alpha \crct n{n-m} \psi),
\end{align*}
where $\alpha,\beta,\phi,\psi\in[n]$ and $1\le k < m\le n$, provided both sides exist. We have
\begin{align*}
& \alpha = FH^{n-1}_{k-1}L(\beta \crct n{n-k} \alpha) = FH^{n-1}_{m-1}R(\alpha \crct n{n-m} \psi)\\
& \beta = FH^{n-1}_{k-1}R(\beta \crct n{n-k} \alpha) = FH^{n-1}_{m-1}R(\beta \crct n{n-m} \phi)\\
& \phi = FH^{n-1}_{k-1}R(\phi \crct n{n-k} \psi) = FH^{n-1}_{m-1}L(\beta \crct n{n-m} \phi)\\
& \psi = FH^{n-1}_{k-1}L(\phi \crct n{n-k} \psi) = FH^{n-1}_{m-1}L(\alpha \crct n{n-m} \psi).
\end{align*}
Moreover, the outermost compositions are the unique $\gamma,\gamma'$ such that
\begin{align*}
\beta \crct n{n-k} \alpha &= FH^{n-1}_{m-1}R(\gamma)
& \beta \crct n{n-m} \phi &= FH^{n-1}_{k-1}R(\gamma')\\
\phi \crct n{n-k} \psi &= FH^{n-1}_{m-1}L(\gamma)
& \alpha \crct n{n-m} \psi &= FH^{n-1}_{k-1}L(\gamma').
\end{align*}
We need to show that $\gamma=\gamma'$. Since $\gamma'$ is unique, we need only show
\begin{align*}
\beta \crct n{n-m} \phi &= FH^{n-1}_{k-1}R(\gamma)
& \alpha \crct n{n-m} \psi &= FH^{n-1}_{k-1}L(\gamma).
\end{align*}
Since $\beta \crct n{n-m} \phi$ and $\alpha \crct n{n-m} \psi$ are the unique elements such that
\begin{align*}
& \beta = FH^{n-1}_{m-1}R(\beta \crct n{n-m} \phi)
&& \alpha = FH^{n-1}_{m-1}R(\alpha \crct n{n-m} \psi)\\
& \phi = FH^{n-1}_{m-1}L(\beta \crct n{n-m} \phi)
&& \psi = FH^{n-1}_{m-1}L(\alpha \crct n{n-m} \psi),
\end{align*}
it suffices to show that
\begin{align*}
& \beta = FH^{n-1}_{m-1}R(FH^{n-1}_{k-1}R(\gamma)) = F(H^{n-1}_{k-1}R \circ H^{n-1}_{m-1}R)(\gamma)\\
& \phi = FH^{n-1}_{m-1}L(FH^{n-1}_{k-1}R(\gamma)) = F(H^{n-1}_{k-1}R \circ H^{n-1}_{m-1}L)(\gamma)\\
& \alpha = FH^{n-1}_{m-1}R(FH^{n-1}_{k-1}L(\gamma)) = F(H^{n-1}_{k-1}L \circ H^{n-1}_{m-1}R)(\gamma)\\
& \psi = FH^{n-1}_{m-1}L(FH^{n-1}_{k-1}L(\gamma)) = F(H^{n-1}_{k-1}L \circ H^{n-1}_{m-1}L)(\gamma).
\end{align*}
But we have
\begin{align*}
\beta \crct n{n-k} \alpha &= FH^{n-1}_{m-1}R(\gamma)
& \phi \crct n{n-k} \psi &= FH^{n-1}_{m-1}L(\gamma).
\end{align*}
and using Lemma \ref{lem:commute},
\begin{align*}
\beta &= FH^{n-1}_{k-1}R(\beta \crct n{n-k} \alpha) = FH^{n-1}_{k-1}R(FH^{n-1}_{m-1}R(\gamma))\\
&= F(H^{n-1}_{m-1}R \circ H^{n-1}_{k-1}R)(\gamma) = F(H^{n-1}_{k-1}R \circ H^{n-1}_{m-1}R)(\gamma)\\
\phi &= FH^{n-1}_{k-1}R(\phi \crct n{n-k} \psi) = FH^{n-1}_{k-1}R(FH^{n-1}_{m-1}L(\gamma))\\
&= F(H^{n-1}_{m-1}L \circ H^{n-1}_{k-1}R)(\gamma) = F(H^{n-1}_{k-1}R \circ H^{n-1}_{m-1}L)(\gamma)\\
\alpha &= FH^{n-1}_{k-1}L(\beta \crct n{n-k} \alpha) = FH^{n-1}_{k-1}L(FH^{n-1}_{m-1}R(\gamma))\\
&= F(H^{n-1}_{m-1}R \circ H^{n-1}_{k-1}L)(\gamma) = F(H^{n-1}_{k-1}L \circ H^{n-1}_{m-1}R)(\gamma)\\
\psi &= FH^{n-1}_{k-1}L(\phi \crct n{n-k} \psi) = FH^{n-1}_{k-1}L(FH^{n-1}_{m-1}L(\gamma))\\
&= F(H^{n-1}_{m-1}L \circ H^{n-1}_{k-1}L)(\gamma) = F(H^{n-1}_{k-1}L \circ H^{n-1}_{m-1}L)(\gamma).
\extagqed
\end{align*}
\end{proof}

\section{Witnesses for the Weak Laws}
\label{sec:connections}

The pushout structure on $\boxdot$ serves two purposes. In \S\ref{sec:presheaf} it made $\theory$ hold on the nose, with no space-filling operation anywhere in the argument. Here we show that the same structure generates the
laws that do \emph{not} hold on the nose. Associativity and the identity laws hold only up to $\simeq$, so each requires a cell of the next dimension to witness it, and in the de Morgan approach such cells are produced by Kan
filling. 
In our setting they are produced instead by the mediating morphisms of the pushout, and the point of this section is that reasoning about them is \emph{equational}: five equations on the morphisms of $\boxdot$ suffice to build
the witnesses and to verify the coherence conditions they must satisfy, including the pentagon. 
Thus, laws that are equations get axioms and a decision procedure, while laws that are not equations get witnesses, built from the same categorical structure.

\subsection{A Calculus for the Morphisms of $\boxdot$}
\label{sec:medcalculus}

Let $I$ be the real unit interval and let $J$ be any subset of $I$ satisfying the following closure properties:
\begin{itemize}
\item
$0,1/2,1\in J$
\item
if $t\in J$ then $2t\bmod 1\in J$
\item
if $t,u,v\in J$ then $tu + (1-t)v\in J$.
\end{itemize}
Equivalently,
\begin{itemize}
\item
$0,1/2,1\in J$
\item
if $s,t\in J$, then $st\in J$
\item
if $t\in J$ then $1-t\in J$
\item
if $s,t\in J$ then $(s+t)\bmod 1\in J$.
\end{itemize}
The smallest set satisfying these properties is $A = \set{a2^{-n}}{n\ge 0,\ 0\le a\le 2^n}$, the set of dyadic rationals in $I$. 
Let us assume that the index set $J$ of $\Hom([0],[1])$ satisfies these properties. This is the set of indices $t$ that we have been using to write $\delta_t$. This gives us a topology to work with, as the inclusion map $J\to I$ induces a metric $d$ on $J$ with $d(s,t)=\len{s-t}$.

Now, the covariant hom-functor $\Hom([0],-)$ composed with the bijection
$\Hom([0],[n])\cong J^n$, which identifies
$\delta_{t_1}\otimes\cdots\otimes\delta_{t_n}$ with $(\seq t1n)$, gives a
faithful functor $K:\boxdot\to\Set$ with 
\begin{align*}
& K[n] = J^n && f:[n]\to[m] \mapsto Kf:J^n\to J^m 
\end{align*}
with $Kf(\seq t1n) = (\seq s1m),\ \text{where } f\circ(\delta_{t_1}\otimes\cdots\otimes\delta_{t_n}) = \delta_{s_1}\otimes\cdots\otimes\delta_{s_m}$.
Consequently an identity between morphisms of $\boxdot$ can be verified pointwise on $J^n$, which is how every computation below is carried out. Further details are given in~\S\ref{apx:boxdot1}.

\subsection{Connections and Associativity}

A \emph{connection} is built from a morphism $c:[1]\to[1]$ with
$c\circ\delta_0=\delta_0$ and $c\circ\delta_1=\delta_1$, equivalently
$Kc(0)=0$ and $Kc(1)=1$; these conditions are what give boundary preservation.
The connection itself is the morphism $f:[2]\to[1]$ with
$Kf(t,s) = t\kern.5pt Kc(s) + (1-t)s$. The base category should be augmented with
these morphisms to effect the connection.

Associativity uses $c = (L\med RL)\med RR$, whose inverse is
$LL\med(LR\med R)$. Unwinding the definition,
\begin{align*}
K((L\med RL)\med RR)(s)
= \begin{cases}2s & \text{if $s\le 1/4$}\\
s+1/4 & \text{if $1/4\le s\le 1/2$}\\
(s+1)/2 & \text{if $s\ge 1/2$,}\end{cases}
\end{align*}
so this morphism maps $0\mapsto 0$, $1/4\mapsto 1/2$, $1/2\mapsto 3/4$ and
$1\mapsto 1$. The interpolation is shown below: from top to bottom, the left
segment is stretched, the right segment compressed, and the middle band simply shifted.

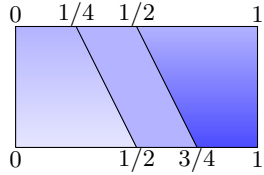
\begin{center}
\vspace{-1ex}
 \begin{tikzpicture}[-, >=stealth', auto, scale=0.8]
%\draw[help lines] (0,0) grid (4,2);
\small
\fill[top color=blue!30, bottom color=blue!10] (0,0) -- (2,0) -- (1,2) -- (0,2) -- (0,0);
\fill[top color=blue!30, bottom color=blue!30] (2,0) -- (3,0) -- (2,2) -- (1,2) -- (2,0);
\fill[top color=blue!30, bottom color=blue!70] (3,0) -- (4,0) -- (4,2) -- (2,2) -- (3,0);
\draw (0,0) -- (4,0) -- (4,2) -- (0,2) -- (0,0);
\draw (1,2) -- (2,0);
\draw (2,2) -- (3,0);
\node at (0,2.2) {$0$};
\node at (1,2.2) {$1/4$};
\node at (2,2.2) {$1/2$};
\node at (4,2.2) {$1$};
\node at (0,-.2) {$0$};
\node at (2,-.2) {$1/2$};
\node at (3,-.2) {$3/4$};
\node at (4,-.2) {$1$};
\end{tikzpicture}
\vspace{-1ex}
\end{center}
That $(L\med RL)\med RR$ and $LL\med(LR\med R)$ are mutually inverse is a
calculation in Lemma \ref{lem:medeq}. So is the observation that the same term
associated the other way collapses:
$L\med (RL\med RR) = L\med (R\circ(L\med R)) = L\med R = \ident 1$.

Now let $F$ be a presheaf as in \S\ref{sec:presheaf} and let $\alpha,\beta,\gamma\in F[n]$ be such that $\alpha\crct n{n-1}\beta$ and $\beta\crct n{n-1}\gamma$ exist. Unwinding the definitions of $\src$ and $\tgt$
shows that $(\alpha\crct n{n-1}\beta)\crct n{n-1}\gamma$ and $\alpha\crct n{n-1}(\beta\crct n{n-1}\gamma)$ exist as well. Associativity up to $\simeq$ is witnessed by applying $FH^{n-1}_0f$ to $(\alpha\crct n{n-1}\beta)\crct n{n-1}\gamma$, where
$f$ is the connection above.
\begin{theorem}
\label{thm:assoc}
With $f$ and $\alpha,\beta,\gamma$ as above,
\begin{align}
\src(FH^{n-1}_0f((\alpha\crct n{n-1}\beta)\crct n{n-1}\gamma)) &= (\alpha\crct n{n-1}\beta)\crct n{n-1}\gamma\label{eq:arnot1}\\
\tgt(FH^{n-1}_0f((\alpha\crct n{n-1}\beta)\crct n{n-1}\gamma)) &= \alpha\crct n{n-1}(\beta\crct n{n-1}\gamma),\label{eq:arnot2}
\end{align}
so $FH^{n-1}_0f((\alpha\crct n{n-1}\beta)\crct n{n-1}\gamma)$ is an
$(n+1)$-cell witnessing the associativity.
\end{theorem}
\begin{proof}
First we observe that $K(f \circ (\delta_t\otimes\ident 1)) = \lam s{tKc(s) + (1-t)s}$, which implies that
$K(f \circ (\delta_0\otimes\ident 1)) = \lam ss = K\ident 1$
and
$K(f \circ (\delta_1\otimes\ident 1)) = \lam s{Kc(s)} = Kc = K((L\med RL)\med RR)$.
Since $K$ is faithful,
$f \circ (\delta_0\otimes\ident 1) = \ident 1$
and
$f \circ (\delta_1\otimes\ident 1) = (L\med RL)\med RR$.
Then for \eqref{eq:arnot1},
\begin{align*}
& \src(FH^{n-1}_0f((\alpha\crct n{n-1}\beta)\crct n{n-1}\gamma)\\
&= FH^n_0\delta_0(FH^{n-1}_0f((\alpha\crct n{n-1}\beta)\crct n{n-1}\gamma)
= F(H^{n-1}_0f \circ H^n_0\delta_0)((\alpha\crct n{n-1}\beta)\crct n{n-1}\gamma)\\
&= FH^{n-1}_0(f \circ (\delta_0\otimes\ident 1))((\alpha\crct n{n-1}\beta)\crct n{n-1}\gamma)
= FH^{n-1}_0\ident 1((\alpha\crct n{n-1}\beta)\crct n{n-1}\gamma)\\
&= F\ident n((\alpha\crct n{n-1}\beta)\crct n{n-1}\gamma)
= (\alpha\crct n{n-1}\beta)\crct n{n-1}\gamma.
\end{align*}
The proof of \eqref{eq:arnot2} is a little more involved. Starting out as above,
\begin{align}
& \tgt(FH^{n-1}_0f((\alpha\crct n{n-1}\beta)\crct n{n-1}\gamma)\nonumber\\
&= FH^n_0\delta_1(FH^{n-1}_0f((\alpha\crct n{n-1}\beta)\crct n{n-1}\gamma)
= F(H^{n-1}_0f \circ H^n_0\delta_1)((\alpha\crct n{n-1}\beta)\crct n{n-1}\gamma)\nonumber\\
&= FH^{n-1}_0(f \circ (\delta_1\otimes\ident 1))((\alpha\crct n{n-1}\beta)\crct n{n-1}\gamma)\nonumber\\
&= FH^{n-1}_0((L\med RL)\med RR)((\alpha\crct n{n-1}\beta)\crct n{n-1}\gamma).\label{eq:arnot3}
\end{align}
Abbreviating the quantity \eqref{eq:arnot3} by $\theta$,
since $\alpha\crct n{n-1}(\beta\crct n{n-1}\gamma)$ is the unique element of $F[n]$ such that
\begin{align*}
\beta\crct n{n-1}\gamma &= FH^{n-1}_0 L(\alpha\crct n{n-1}(\beta\crct n{n-1}\gamma)) & \alpha &= FH^{n-1}_0 R(\alpha\crct n{n-1}(\beta\crct n{n-1}\gamma))
\end{align*}
and $\beta\crct n{n-1}\gamma$ is the unique element of $F[n]$ such that
\begin{align*}
\gamma &= FH^{n-1}_0 L(\beta\crct n{n-1}\gamma) & \beta &= FH^{n-1}_0 R(\beta\crct n{n-1}\gamma),
\end{align*}
to show \eqref{eq:arnot2} it suffices to show that
\begin{align*}
\alpha &= FH^{n-1}_0 R(\theta)\\
\beta &= FH^{n-1}_0 R( FH^{n-1}_0 L( \theta ) ) = FH^{n-1}_0(LR)(\theta)\\
\gamma &= FH^{n-1}_0 L( FH^{n-1}_0 L( \theta ) ) = FH^{n-1}_0(LL)(\theta).
\end{align*}
Here are the arguments, using the equations of Lemma \ref{lem:medeq}.
\begin{align*}
& FH^{n-1}_0 R(FH^{n-1}_0((L\med RL)\med RR)((\alpha\crct n{n-1}\beta)\crct n{n-1}\gamma))\\
&= FH^{n-1}_0(((L\med RL)\med RR) \circ R)((\alpha\crct n{n-1}\beta)\crct n{n-1}\gamma)\\
&= FH^{n-1}_0(RR)((\alpha\crct n{n-1}\beta)\crct n{n-1}\gamma)
= FH^{n-1}_0R(\alpha\crct n{n-1}\beta)
= \alpha\\[1em]
& FH^{n-1}_0(LR)(FH^{n-1}_0((L\med RL)\med RR)((\alpha\crct n{n-1}\beta)\crct n{n-1}\gamma))\\
&= FH^{n-1}_0(((L\med RL)\med RR) \circ LR)((\alpha\crct n{n-1}\beta)\crct n{n-1}\gamma)\\
&= FH^{n-1}_0((L\med RL) \circ R)((\alpha\crct n{n-1}\beta)\crct n{n-1}\gamma)
= FH^{n-1}_0(RL)((\alpha\crct n{n-1}\beta)\crct n{n-1}\gamma)\\
&= FH^{n-1}_0L(FH^{n-1}_0R((\alpha\crct n{n-1}\beta)\crct n{n-1}\gamma))
= FH^{n-1}_0L(\alpha\crct n{n-1}\beta)
= \beta\\[1em]
& FH^{n-1}_0(LL)(FH^{n-1}_0((L\med RL)\med RR)((\alpha\crct n{n-1}\beta)\crct n{n-1}\gamma))\\
&= FH^{n-1}_0(((L\med RL)\med RR) \circ LL)((\alpha\crct n{n-1}\beta)\crct n{n-1}\gamma)\\
&= FH^{n-1}_0L((\alpha\crct n{n-1}\beta)\crct n{n-1}\gamma)
= \gamma.
\extagqed
\end{align*}
\end{proof}

For the identity laws, we use $L$ and $R$ in one direction and $L' = \ident 1 \med (\delta_1\circ\sigma)$ and $R' = (\delta_0\circ\sigma)\med\ident 1$ in the other. Interpolating between $\ident 1$ and $L$, we have $KL(s) = s/2$, giving $Kf(t,s) = ts/2 + (1-t)s = s - ts/2$.
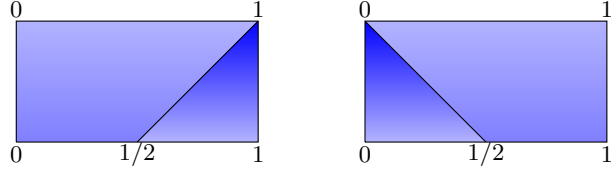
\begin{center}
\vspace{-1ex}
 \begin{tikzpicture}[-, >=stealth', auto, scale=0.8]
%\draw[help lines] (0,0) grid (4,2);
\small
\fill[top color=blue!30, bottom color=blue!50] (0,0) -- (2,0) -- (4,2) -- (0,2) -- (0,0);
\fill[top color=blue!100, bottom color=blue!30] (2,0) -- (4,0) -- (4,2) -- (2,0);
\draw (0,0) -- (4,0) -- (4,2) -- (0,2) -- (0,0);
\draw (4,2) -- (2,0);
\node at (0,2.2) {$0$};
\node at (4,2.2) {$1$};
\node at (0,-.2) {$0$};
\node at (2,-.2) {$1/2$};
\node at (4,-.2) {$1$};
\end{tikzpicture}
\hspace{8mm}
\begin{tikzpicture}[-, >=stealth', auto, scale=0.8]
%\draw[help lines] (0,0) grid (4,2);
\small
\fill[top color=blue!100, bottom color=blue!30] (0,0) -- (2,0) -- (0,2) -- (0,0);
\fill[top color=blue!30, bottom color=blue!50] (2,0) -- (4,0) -- (4,2) -- (0,2) -- (2,0);
\draw (0,0) -- (4,0) -- (4,2) -- (0,2) -- (0,0);
\draw (0,2) -- (2,0);
\node at (0,2.2) {$0$};
\node at (4,2.2) {$1$};
\node at (0,-.2) {$0$};
\node at (2,-.2) {$1/2$};
\node at (4,-.2) {$1$};
\end{tikzpicture}
\vspace{-1ex}
\end{center}
 
\subsection{Higher Coherence}
\label{sec:coherence}
The same calculus can be used to verify coherence conditions. For example, the pentagon identity for associativity is the statement
\begin{align*}
& (((LL\med LR)\med RL) \med RR)\\
&\circ ((LLL\med(LLR\med LR)) \med R)\\
&\circ (LL\med((LRL\med LRR)\med R))\\
&\circ (L\med (RLL\med (RLR\med RR)))\\
&\circ ((L\med RL)\med(RRL\med RRR)) = \ident 1,
\end{align*}
each composition factor representing one of the five reassociations around the pentagon
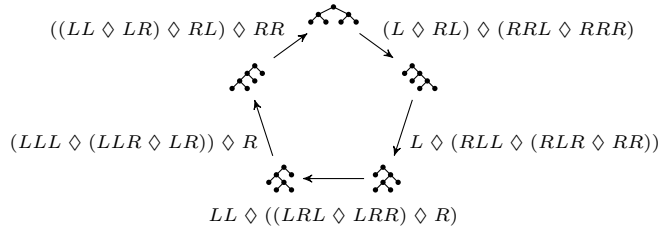
\begin{center}
\begin{tikzpicture}[->, >=stealth', auto]
%\draw[help lines] (-2,-2) grid (2,2);
\small
\node (A) at (90:12mm) {
\begin{tikzpicture}[-, >=stealth', inner sep=0pt, minimum size=0pt, auto]
\fill (0,0) circle (1pt);
\fill (-.2,-.1) circle (1pt);
\fill (.2,-.1) circle (1pt);
\fill (-.3,-.2) circle (1pt);
\fill (-.1,-.2) circle (1pt);
\fill (.1,-.2) circle (1pt);
\fill (.3,-.2) circle (1pt);
\draw (0,0) -- (-.2,-.1);
\draw (0,0) -- (.2,-.1);
\draw (-.2,-.1) -- (-.3,-.2);
\draw (-.2,-.1) -- (-.1,-.2);
\draw (.2,-.1) -- (.1,-.2);
\draw (.2,-.1) -- (.3,-.2);
\end{tikzpicture}
};
\node (B) at (18:12mm) {
\begin{tikzpicture}[-, >=stealth', inner sep=0pt, minimum size=0pt, auto]
\fill (0,0) circle (1pt);
\fill (-.1,-.1) circle (1pt);
\fill (.1,-.1) circle (1pt);
\fill (0,-.2) circle (1pt);
\fill (.2,-.2) circle (1pt);
\fill (.1,-.3) circle (1pt);
\fill (.3,-.3) circle (1pt);
\draw (0,0) -- (-.1,-.1);
\draw (0,0) -- (.1,-.1);
\draw (.1,-.1) -- (0,-.2);
\draw (.1,-.1) -- (.2,-.2);
\draw (.2,-.2) -- (.1,-.3);
\draw (.2,-.2) -- (.3,-.3);
\end{tikzpicture}
};
\node (C) at (-54:12mm) {
\begin{tikzpicture}[-, >=stealth', inner sep=0pt, minimum size=0pt, auto]
\fill (0,0) circle (1pt);
\fill (-.1,-.1) circle (1pt);
\fill (.1,-.1) circle (1pt);
\fill (0,-.2) circle (1pt);
\fill (.2,-.2) circle (1pt);
\fill (-.1,-.3) circle (1pt);
\fill (.1,-.3) circle (1pt);
\draw (0,0) -- (-.1,-.1);
\draw (0,0) -- (.1,-.1);
\draw (.1,-.1) -- (0,-.2);
\draw (.1,-.1) -- (.2,-.2);
\draw (0,-.2) -- (-.1,-.3);
\draw (0,-.2) -- (.1,-.3);
\end{tikzpicture}
};
\node (D) at (-126:12mm) {
\begin{tikzpicture}[-, >=stealth', inner sep=0pt, minimum size=0pt, auto]
\fill (0,0) circle (1pt);
\fill (-.1,-.1) circle (1pt);
\fill (.1,-.1) circle (1pt);
\fill (0,-.2) circle (1pt);
\fill (-.2,-.2) circle (1pt);
\fill (-.1,-.3) circle (1pt);
\fill (.1,-.3) circle (1pt);
\draw (0,0) -- (-.1,-.1);
\draw (0,0) -- (.1,-.1);
\draw (-.1,-.1) -- (0,-.2);
\draw (-.1,-.1) -- (-.2,-.2);
\draw (0,-.2) -- (-.1,-.3);
\draw (0,-.2) -- (.1,-.3);
\end{tikzpicture}
};
\node (E) at (-198:12mm) {
\begin{tikzpicture}[-, >=stealth', inner sep=0pt, minimum size=0pt, auto]
\fill (0,0) circle (1pt);
\fill (-.1,-.1) circle (1pt);
\fill (.1,-.1) circle (1pt);
\fill (0,-.2) circle (1pt);
\fill (-.2,-.2) circle (1pt);
\fill (-.1,-.3) circle (1pt);
\fill (-.3,-.3) circle (1pt);
\draw (0,0) -- (-.1,-.1);
\draw (0,0) -- (.1,-.1);
\draw (-.1,-.1) -- (0,-.2);
\draw (-.1,-.1) -- (-.2,-.2);
\draw (-.2,-.2) -- (-.1,-.3);
\draw (-.2,-.2) -- (-.3,-.3);
\end{tikzpicture}
};
\path (A) edge node[inner sep=2pt] {\scriptsize $(L\med RL)\med(RRL\med RRR)$} (B);
\path (B) edge node[inner sep=2pt] {\scriptsize $L\med (RLL\med (RLR\med RR))$} (C);
\path (C) edge node[yshift=-3mm, inner sep=2pt] {\scriptsize $LL\med((LRL\med LRR)\med R)$} (D);
\path (D) edge node[inner sep=2pt] {\scriptsize $(LLL\med(LLR\med LR)) \med R$} (E);
\path (E) edge node[inner sep=2pt] {\scriptsize $((LL\med LR)\med RL) \med RR$} (A);
\end{tikzpicture}
%\vspace{-1ex}
\end{center}
and this follows by equational reasoning from Lemma \ref{lem:medeq} alone.

These terms are not pulled out of thin air. To write the morphism that reassociates one bracketing of a composite into another, list the leaf paths of the \emph{target} bracketing from left to right as words in $L$ and $R$, then combine that list
with $\med$ according to the shape of the \emph{source} bracketing. For example, to reassociate $(\alpha\crct n{n-1}\beta)\crct n{n-1}\gamma$ to $\alpha\crct n{n-1}(\beta\crct n{n-1}\gamma)$, the target has leaves $L$, $RL$
and $RR$ while the source is bracketed to the left, which gives $(L\med RL)\med RR$. Every factor of the pentagon above is built by the same rule, and so is every higher reassociation.
The construction is therefore quite mechanical, and so is its verification: coherence in this setting is rewriting with five equations rather than diagram chasing and is amenable to mechanization.

\section{Conclusion}
\label{sec:conclusion}

We have given $\theory$, eleven equational axioms  for strict equality between
cells of higher identity types, and shown that they are exactly the equations
that hold in every homotopy model, and indeed already in every finite
CW-complex (\S\ref{sec:completeness}). All but two of them, read left to right,
form a strongly normalizing rewrite system whose normal forms are unique up to
picture equivalence, and adding generators in higher dimensions proves no new
equations below (\S\ref{sec:eqpres}). Thus, the resulting word problem is decidable. 
On the semantic side,  internalizing composition as a pushout on the base category makes all
of $\theory$  hold on the nose in every boundary-preserving, pullback-preserving
presheaf (\S\ref{sec:presheaf}), and the same pushout supplies the witnesses for
the laws that hold only up to $\simeq$ (\S\ref{sec:connections}).

Everything provable in our theory is an equation, and everything else requires a witness. The
calculus of mediating morphisms of \S\ref{sec:connections} produces those
witnesses for associativity and the identity laws, as well as higher coherences. 
A second direction for future work is type theory. The models of \S\ref{sec:presheaf} already
carry composition, boundaries and witnesses, but not the fibration structure that
interpreting dependent types requires. We would like to know which
boundary-preserving presheaves over $\boxdot$ admit such structure, how large
that class is, and whether it is closed under the type constructors. 
Finally, our decision procedure materializes pictures as arrays and gives only a rough exponential upper bound. Comparing succinct representations rather than whole arrays may well do better.

%\newpage
\bibliographystyle{splncs04}
\bibliography{Globular}

\begin{toappendix}
\section{Characterization of $\boxdot_1$}
\label{apx:boxdot1}

\subsection{Set Representation}
\label{apx:setrep}

In this section we develop the properties of the functor $K$ defined in \S\ref{sec:medcalculus}.

The covariant hom-functor $\Hom([0],-):\boxdot\to\Set$ maps each object $[n]$ to the set of morphisms $\Hom([0],[n])$ and each morphism $f:[n]\to[m]$ to the set function $f\circ-:\Hom([0],[n])\to\Hom([0],[m])$ that post-composes with $f$. This functor is faithful by Theorem \ref{thm:decomp}. Specifically, each element of $\Hom([0],[n])$ has a unique representation $\delta_{t_1}\otimes\cdots\otimes\delta_{t_n}$, and morphisms $f:[n]\to[m]$ are uniquely determined by their action on these representations by post-composition. There is a bijection $\Hom([0],[n])\cong J^n$ identifying $\delta_{t_1}\otimes\cdots\otimes\delta_{t_n}$ with the tuple $(\seq t1n)$, and composing this bijection with the hom-functor gives a faithful functor $K:\boxdot\to\Set$ with
\begin{align*}
& K[n] = J^n && f:[n]\to[m] \mapsto Kf:J^n\to J^m
\end{align*}
with
\begin{align*}
Kf(\seq t1n) = (\seq s1m),\ \text{where } f\circ(\delta_{t_1}\otimes\cdots\otimes\delta_{t_n}) = \delta_{s_1}\otimes\cdots\otimes\delta_{s_m}.
\end{align*}
\begin{align}
\begin{array}c
\begin{tikzpicture}[->, >=stealth', auto, node distance=32mm]
\small
\node (NW) {$[n]$};
\node (NE) [right of=NW] {$\Hom([0],[n])$};
\node (SW) [below of=NW, node distance=12mm] {$[m]$};
\node (SE) [below of=NE, node distance=12mm] {$\Hom([0],[m])$};
\path (NW) edge node {$\Hom([0],-)$} (NE);
\path (NW) edge node[swap] {$f$} (SW);
\path (NE) edge node {$f\circ-$} (SE);
\path (SW) edge node[swap] {$\Hom([0],-)$} (SE);
\node (WNE) [right of=NE, node distance=20mm] {$J^n$};
\node (WSE) [right of=SE, node distance=20mm] {$J^m$};
\path (NE) edge node {$\cong$} (WNE);
\path (SE) edge node[swap] {$\cong$} (WSE);
\path (WNE) edge node {$Kf$} (WSE);
\end{tikzpicture}
\end{array}
\label{eq:depcoalg}
\end{align}

In this view, all functions $Kf:J\to J$ representing morphisms $f:[1]\to[1]$ are uniformly continuous in the metric topology on $J$. The maps $KL$ and $KR$ are each contractive with constant of contraction $1/2$: $KL(t) = t/2$ and $KR(t) = (t+1)/2$, so
\begin{align*}
d(KL(s),KL(t)) &= \len{s/2-t/2} = \len{s-t}/2 = d(s,t)/2\\
d(KR(s),KR(t)) &= \len{(s+1)/2-(t+1)/2} = \len{s-t}/2 = d(s,t)/2.
\end{align*}
The function $K(f\med g)$ is Lipschitz with constant $c_{f\med g} = 2\max c_f,c_g$, where $c_f$ and $c_g$ are the Lipschitz constants of $f$ and $g$. To see this, note that
\begin{align*}
K(f\med g)(t) = \begin{cases}
Kf(2t) & \text{if $t\le 1/2$, that is, if $2t\le 1$ and $t=KL(2t)$}\\
Kg(2t-1) & \text{if $t\ge 1/2$, that is, if $2t\ge 1$ and $t=KR(2t-1)$.}
\end{cases}
\end{align*}
Then
\begin{align*}
& d(K(f\med g)(s),K(f\med g)(t))\\
&= \begin{cases}
d(Kf(2s),Kf(2t)) & \text{if $s,t\le 1/2$}\\
d(Kg(2s-1),Kg(2t-1)) & \text{if $s,t\ge 1/2$}\\
d(Kf(2s),Kg(2t-1)) & \text{if $s\le 1/2$ and $t\ge 1/2$.}
\end{cases}
\end{align*}
The first two quantities are bounded by $2c_f d(s,t)$ and $2c_g d(s,t)$, respectively. For the last case, recalling that $Kf(1)=Kg(0)$ and using the triangle inequality,
\begin{align*}
d(Kf(2s),Kg(2t-1)) &\le d(Kf(2s),Kf(1)) + d(Kg(0),Kg(2t-1))\\
&\le c_f(1-2s) + c_g(2t-1)
\le 2(\max c_f,c_g)d(s,t).
\end{align*}

\subsection{Equations for $\boxdot_1$}

\begin{lemma}
\label{lem:medeq}
Let $U,V,W:[1]\to[1]$. The following equations are sound whenever both sides exist. 
\begin{itemize}
\item $(U\med V)\circ L=U$
\item $(U\med V)\circ R=V$
\item $U\circ(V\med W) = (U\circ V)\med(U\circ W)$
\item $(\delta_t\circ\sigma)\med(\delta_t\circ\sigma) = \delta_t\circ\sigma$
\item $L\med R = \ident 1$
\end{itemize}
In all cases, if the left-hand side exists, then so does the right.
\end{lemma}
\begin{proof}
The first two are by definition. For the third, a left distributive law, if $V\circ\delta_1=W\circ\delta_0$, then $U\circ V\circ\delta_1=U\circ W\circ\delta_0$, so if the left-hand side exists, then so does the right. The two sides are equal, because $U\circ(V\med W)\circ L = U\circ V = ((U\circ V)\med(U\circ W))\circ L$ and $U\circ(V\med W)\circ R = U\circ W = ((U\circ V)\med(U\circ W))\circ R$, and the mediating morphism is unique. The last two equations can be proved by the same technique.
\exqed
\end{proof}

\begin{lemma}
\label{lem:boxdot1}
Every element of $\Hom([1],[1])$ generated by $\sigma$, $L$, $R$, and the $\delta_t$ under $\circ$ and $\med$ is represented by a term generated by the grammar
\begin{align*}
S\ &::=\ S \med S \mid X \mid \delta_t\circ\sigma & X\ &::=\ R\circ X \mid L\circ X \mid \ident 1
\end{align*}
with start symbol $S$.
\end{lemma}
\begin{proof}
We reduce any term over the stated primitives to the desired form using the equations of Lemma \ref{lem:medeq}. First apply the left-distributive law $U\circ(V\med W) = (U\circ V)\med(U\circ W)$ in the left-to-right direction as far as possible. After this step, there are no subterms of the form $U\circ(V\med W)$. Any remaining compositions $(U\med V)\circ L$ or $(U\med V)\circ R$ can be simplified to $U$ and $V$, respectively. Any remaining compositions $(U\med V)\circ\delta_t$ are equal to $\delta_s$ for some $s$, so can be replaced by $\delta_s$. There are no compositions of the form $(U\med V)\circ\sigma$ because it is not type-correct.

After this step, there are no compositions with any $U\med V$ as a composition factor. Any remaining compositions with an occurrence of $\sigma$ or $\delta_t$ as a composition factor reduces to some $\delta_t\circ\sigma$, because $\sigma\circ L=\sigma\circ R=\sigma$ and $L\circ\delta_t$ and $R\circ\delta_t$ are equal to some $\delta_s$.

At this point the term is of the desired form as specified by the grammar in the statement of the lemma. However further simplifications are possible using the remaining equations. Any subterm of the form $(\delta_s\circ\sigma)\med(\delta_t\circ\sigma)$ must have $s=t$, as $\delta_s = \delta_s\circ\sigma\circ\delta_1=\delta_t\circ\sigma\circ\delta_0 = \delta_t$, thus the subterm can be reduced to $\delta_t\circ\sigma$. The equation is valid, as $((\delta_t\circ\sigma)\med(\delta_t\circ\sigma))\circ L=\delta_t\circ\sigma=\delta_t\circ\sigma\circ L$ and 
$((\delta_t\circ\sigma)\med(\delta_t\circ\sigma))\circ R=\delta_t\circ\sigma=\delta_t\circ\sigma\circ R$. Also, any subterm of the form $(U\circ L)\med(U\circ R)$ can be reduced to $U$, as $(U\circ L)\med(U\circ R) = U\circ(L\med R) = U\circ\ident 1 = U$.
\footnote{Above we applied left distributive law only in the left-to-right direction. Here the reverse direction (right-to-left) is still valid because $L\circ\delta_1 = R\circ\delta_0 = \delta_{1/2}$.}
\exqed
\end{proof}

Now we show that after these simplifications, the representation is unique. We show that any pair of reduced terms that are semantically equal are also syntactically equal.

\begin{theorem}
\label{thm:boxdot1}
Let $f,g:[1]\to[1]$. Let $A = \set{a2^{-n}}{n\ge 0,\ 0\le a\le 2^n}$. The following are equivalent:
\begin{enumerate}[{\upshape(i)}]
\item $f = g$
\item for all $t$, $f\circ\delta_t=g\circ\delta_t$
\item for all $t\in A$, $f\circ\delta_t=g\circ\delta_t$
\item the reduced forms of $f$ and $g$ as described in Lemma \ref{lem:boxdot1} are identical.
\end{enumerate}
\end{theorem}
\begin{proof}
The implications (i) $\Imp$ (ii), (ii) $\Imp$ (iii), and (iv) $\Imp$ (i) are obvious. For (iii) $\Imp$ (iv), the proof is by induction. Let us write $\equiv$ for syntactic equality of expressions.

Suppose $U_1\med V_1$ and $U_2\med V_2$ are reduced. If for all $t\in A$, $(U_1\med V_1)\circ\delta_t = (U_2\med V_2)\circ\delta_t$, then for all $t\in A$, $U_1\circ\delta_t = (U_1\med V_1)\circ L\circ\delta_t = (U_2\med V_2)\circ L\circ\delta_t = U_2\circ\delta_t$ and $V_1\circ\delta_t = (U_1\med V_1)\circ R\circ\delta_t = (U_2\med V_2)\circ R\circ\delta_t = V_2\circ\delta_t$. By the induction hypothesis, $U_1\equiv U_2$ and $V_1\equiv V_2$, therefore $U_1\med V_1\equiv U_2\med V_2$.

Suppose $U\med V$ is reduced. If for all $t\in A$, $U\med V = \delta_t\circ\sigma$, then for all $t\in A$, $U\circ\delta_t = (U\med V)\circ L\circ\delta_t = \delta_t\circ\sigma\circ L\circ\delta_t = \delta_t\circ\sigma\circ\delta_t$ and $V\circ\delta_t = (U\med V)\circ R\circ\delta_t = \delta_t\circ\sigma\circ R\circ\delta_t = \delta_t\circ\sigma\circ\delta_t$. By the induction hypothesis, $U\equiv\delta_t\circ\sigma$ and $V\equiv\delta_t\circ\sigma$, therefore $U\med V\equiv(\delta_t\circ\sigma)\med(\delta_t\circ\sigma)$. But then $U\med V$ was not fully reduced, as it reduces to $\delta_t\circ\sigma$.

Suppose $U\med V$ is reduced. Let $W$ be a composition of $R$ and $L$. If for all $t\in A$, $(U\med V)\circ\delta_t=W\circ\delta_t$, then for all $t\in A$, $U\circ\delta_t = (U\med V)\circ L\circ\delta_t = W\circ L\circ\delta_t$ and $V\circ\delta_t = (U\med V)\circ R\circ\delta_t = W\circ R\circ\delta_t$. By the induction hypothesis, $U\equiv W\circ L$ and $V\equiv W\circ R$, therefore $U\med V\equiv(W\circ L)\med(W\circ R)$. But then $U\med V$ was not fully reduced, as it reduces to $W$. 

For the remaining cases involving $U$ a composition of $L$ and $R$, let $\#U$ denote the number represented by $U$ in binary, taking $L$ as $0$ and $R$ as $1$, and let $\len U$ be the number of composition factors in $U$. One can show inductively that for any $t\in A$, $U\circ\delta_t = \delta_s$, where $s = (\#U+t)\cdot 2^{-\len U}$. Thus we cannot have $U=\delta_t\circ\sigma$, because $U\circ\delta_0 \ne U\circ\delta_1$, since $\#U\cdot 2^{-\len U} \ne (\#U+1)\cdot 2^{-\len U}$, whereas $\delta_t\circ\sigma\circ\delta_0 = \delta_t\circ\sigma\circ\delta_1 = \delta_t\circ\sigma$.

Moreover, if $U$ and $V$ are compositions of $L$ and $R$ and $U\circ\delta_t=V\circ\delta_t$ for all $t\in A$, then $(\#U+t)\cdot 2^{-\len U} = (\#V+t)\cdot 2^{-\len V}$. In particular, for $t\in\{0,1\}$, we have $\#U\cdot 2^{-\len U} = \#V\cdot 2^{-\len V}$ and $(\#U+1)\cdot 2^{-\len U} = (\#V+1)\cdot 2^{-\len V}$, which implies $\len U=\len V$ and $\#U=\#V$, thus $U\equiv V$.
\exqed
\end{proof}

\section{Properties of $\boxdot$}
\label{apx:boxdot}

\begin{lemma}
\label{lem:ds}
$\delta_t\otimes\sigma = \delta_t\circ\sigma = \sigma\otimes\delta_t$.
\end{lemma}
\begin{proof}
\begin{align*}
\delta_t\otimes\sigma &= (\delta_t\circ\ident 0)\otimes(\ident 0\circ\sigma)
= (\delta_t\otimes\ident 0)\circ(\ident 0\otimes\sigma) = \delta_t\circ\sigma\\
&= (\ident 0\otimes\delta_t)\circ(\sigma\otimes\ident 0)
= (\ident 0\circ\sigma)\otimes(\delta_t\circ\ident 0) = \sigma\otimes\delta_t.
\extagqed
\end{align*}
\end{proof}

\begin{lemma}
\label{lem:fg}
The tensor product in $\boxdot$, as a bifunctor $\otimes:\boxdot\times \boxdot\to \boxdot$, is faithful.
\end{lemma}
\begin{proof}
%\cite{nlab:semicartesian}
We need to show that if $f,f':[n]\to[m]$, $g,g':[k]\to[\ell]$, and $f\otimes g=f'\otimes g'$, then $f=f'$ and $g=g'$. Let $h^{\otimes n} = \underbrace{h\otimes\cdots\otimes h}_{n}$ and $h^{\otimes 0} = \ident 0$. Then $f$ and $g$ can be extracted from $f\otimes g$:
\begin{align*}
& (\ident m\otimes\sigma^{\otimes\ell})\circ(f\otimes g)\circ(\ident n\otimes\delta_0^{\otimes k})\\
&= (\ident m\circ f\circ\ident n)\otimes(\sigma^{\otimes\ell}\circ g\circ\delta_0^{\otimes k})
= f\otimes\ident 0 = f\\
& (\sigma^{\otimes m}\otimes\ident\ell)\circ(f\otimes g)\circ(\delta_0^{\otimes n}\otimes\ident k)\\
&= (\sigma^{\otimes m}\circ f\circ\delta_0^{\otimes n})\otimes(\ident\ell\circ g\circ\ident k)
= \ident 0\otimes g = g,
\end{align*}
hence
\begin{align*}
f &= (\ident m\otimes\sigma^{\otimes\ell})\circ(f\otimes g)\circ(\ident n\otimes\delta_0^{\otimes k})\\
&= (\ident m\otimes\sigma^{\otimes\ell})\circ(f'\otimes g')\circ(\ident n\otimes\delta_0^{\otimes k}) = f'\\
g &= (\sigma^{\otimes m}\otimes\ident\ell)\circ(f\otimes g)\circ(\delta_0^{\otimes n}\otimes\ident k)\\
&= (\sigma^{\otimes m}\otimes\ident\ell)\circ(f'\otimes g')\circ(\delta_0^{\otimes n}\otimes\ident k) = g'.
\extagqed
\end{align*}
\end{proof}

Let $\boxdot_1$ be the induced subcategory of $\boxdot$ on objects $[0]$ and $[1]$. Let $\sum_n \boxdot_1^n$ be the free strict monoidal category on $\boxdot_1$. The objects of $\sum_n \boxdot_1^n$ are sequences $[d_1]\cmp\cdots\cmp[d_n]$, where $d_i\in\{0,1\}$, and the morphisms are sequences $f_1\cmp\cdots\cmp f_n:[d_1]\cmp\cdots\cmp[d_n]\to[e_1]\cmp\cdots\cmp[e_n]$, where $f_i:[d_i]\to[e_i]$. Tensor product for both objects and morphisms is concatenation, denoted by $\cmp$.

The tensor product $\otimes$ in $\boxdot$ gives a monoidal functor $\Otimes:\sum_n \boxdot_1^n\to \boxdot$ in which $\Otimes([d_1]\cmp\cdots\cmp[d_n])=[d_1]\otimes\cdots\otimes[d_n]=[\sum_{i=1}^n d_i]$ and $\Otimes(f_1\cmp\cdots\cmp f_n) = f_1\otimes\cdots\otimes f_n$. It follows from Lemma \ref{lem:fg} that $\Otimes$ is faithful.

A \emph{decomposition} of $f:[n]\to[m]$ is a morphism $\hat f$ of $\sum_n \boxdot_1^n$ such that $\Otimes\hat f= f$.

\begin{lemma}
\label{lem:decomp}
Every morphism of $\boxdot$ has a decomposition $f_1\cmp\cdots\cmp f_n$ such that each $f_i$ is either
\begin{itemize}
\item some $\delta_t:[0]\to[1]$,
\item $\sigma:[1]\to[0]$, or
\item a morphism $[1]\to[1]$ in the reduced form of Lemma \ref{lem:boxdot1}, but not $\delta_t\circ\sigma$,
\end{itemize}
and there is no adjacent pair $f_i\cmp f_{i+1}$ with $f_i=\delta_t$ and $f_{i+1}=\sigma$.
\end{lemma}
\begin{proof}
First we show that every morphism has a decomposition. Recall that the morphisms of $\boxdot$ are generated by morphisms of $\boxdot_1$ under composition and tensor product. The morphisms of $\boxdot_1$ are already decomposed. Preservation of decomposability under $\otimes$ is also obvious.

For composition, suppose $\hat f$ and $\hat g$ are decompositions of $f:[k]\to[\ell]$ and $g:[\ell]\to[m]$, respectively. By adding in extra null identities $\ident 0:[0]\to[0]$ as necessary, we can assume without loss of generality that $\hat f = f_1\cmp\cdots\cmp f_n$ and $\hat g = g_1\cmp\cdots\cmp g_n$, and for all $i$, $f_i:[c_i]\to[d_i]$ and $g_i:[d_i]\to[e_i]$, where $c_i,d_i,e_i\in\{0,1\}$. Then $g_i\circ f_i:[c_i]\to[e_i]$ and $(g_1\circ f_1)\cmp\cdots\cmp(g_n\circ f_n)$ is a decomposition of $g\circ f$:
\begin{align*}
\textstyle\Otimes_i (g_i\circ f_i) = (\Otimes_i g_i)\circ(\Otimes_i f_i) = g\circ f.
\end{align*}

We have shown that every morphism has a decomposition. Now we show how to reduce it to the form described in the statement of the lemma. Null identities $\ident 0:[0]\to[0]$ can be deleted. Any factor $[1]\to[0]$ must be $\sigma$. Any factor $[0]\to[1]$ must be some $\delta_t$. The remaining factors $[1]\to[1]$ can be reduced using Lemma \ref{lem:boxdot1}. Any factor $\delta_t\circ\sigma$ can be replaced with $\sigma\cmp\delta_t:[1]\cmp[0]\to[0]\cmp[1]$, and any adjacent pair $\delta_t\cmp\sigma:[0]\cmp[1]\to[1]\cmp[0]$ can be replaced with $\sigma\cmp\delta_t:[1]\cmp[0]\to[0]\cmp[1]$. These steps are justified by Lemma \ref{lem:ds}.
\exqed
\end{proof}

The following theorem shows that decompositions are unique up to the equations of Lemma \ref{lem:ds} and the addition/deletion of null identities $\ident 0$.
\begin{theorem}
\label{thm:decomp}
Let $f,g:[n]\to[m]$. The following are equivalent:
\begin{enumerate}[\upshape(i)]
\item
$f = g$
\item
for all $h:[0]\to[n]$, $f\circ h=g\circ h$
\item
for all $\seq t1n\in A$, $f\circ(\delta_{t_1}\otimes\cdots\otimes\delta_{t_n})=g\circ(\delta_{t_1}\otimes\cdots\otimes\delta_{t_n})$
\item
the decompositions of $f$ and $g$ in the reduced form described in Lemma \ref{lem:decomp} are identical.
\end{enumerate}
\end{theorem}
\begin{proof}
The implications (i) $\Imp$ (ii), (ii) $\Imp$ (iii), and (iv) $\Imp$ (i) are obvious. For (iii) $\Imp$ (iv), the proof is by induction. Let $\hat f$ and $\hat g$ be decompositions of $f$ and $g$, respectively, of the reduced form described in Lemma \ref{lem:decomp}.

First we show that if $f\circ(\delta_{t_1}\otimes\cdots\otimes\delta_{t_n})=g\circ(\delta_{t_1}\otimes\cdots\otimes\delta_{t_n})$ for all $\seq t1n\in A$, then the leftmost decomposition factors of $\hat f$ and $\hat g$ must be of the same type.

Suppose $\hat f = \delta_t\cmp f'$ and $\hat g=\sigma\cmp g'$. Since $m\ge 1$, there must be a leftmost factor of $\hat g$ with codomain $1$, say $U$, which cannot be any $\delta_s$ because the pattern $\sigma\cmp\delta_s$ is disallowed. Then $U$ must be of type $[1]\to[1]$, but not $\delta_s\circ\sigma$ for any $s$, as that is also disallowed. By Theorem \ref{thm:boxdot1}, there must exist $u,v\in A$ such that $U\circ\delta_u=\delta_v\ne\delta_t$. Post-composing and precomposing $f$ with $\ident 1\otimes\sigma^{\otimes(m-1)}$ and $\delta_u^{\otimes n}$, respectively, yields $\delta_t$. Doing the same with $g$ yields $\delta_v$. As $\delta_t\ne\delta_v$, we have $f\circ\delta_u^{\otimes n}\ne g\circ\delta_u^{\otimes n}$.

Similarly, suppose $\hat f=\delta_t\cmp f'$ and $\hat g = U\cmp g'$, where $U:[1]\to[1]$ is a reduced term of Lemma \ref{lem:boxdot1}, but not $\delta_s\circ\sigma$ for any $s$. As before, there must exist $u,v\in A$ such that $U\circ\delta_u=\delta_v\ne\delta_t$. Post-composing and precomposing $f$ and $g$ with $\ident 1\otimes\sigma^{\otimes(m-1)}$ and $\delta_u^{\otimes n}$ yields $\delta_t\ne\delta_v$, therefore $f\circ\delta_u^{\otimes n}\ne g\circ\delta_u^{\otimes n}$.

Suppose $\hat f = \sigma\cmp f'$ and $\hat g = U\cmp g'$, where $U:[1]\to[1]$ is a reduced term of Lemma \ref{lem:boxdot1}, but not $\delta_s\circ\sigma$ for any $s$. Post-composing and precomposing $f$ and $g$ with $\ident 1\otimes\sigma^{\otimes(m-1)}$ and $\ident 1\otimes\delta_0^{\otimes(n-1)}$ yields some $\sigma\otimes\delta_t$ for some $t$ and $U$. As before, there must exist $u,v\in A$ such that $U\circ\delta_u=\delta_v\ne\delta_t$. Now post-composing and precomposing $f$ and $g$ with $\ident 1\otimes\sigma^{\otimes(m-1)}$ and $\delta_u\otimes\delta_0^{\otimes(n-1)}$ yields $\delta_t\ne\delta_v$, therefore $f\circ(\delta_u\otimes\delta_0^{\otimes(n-1)})\ne g\circ(\delta_u\otimes\delta_0^{\otimes(n-1)})$.

We have shown that if $f\circ(\delta_{t_1}\otimes\cdots\otimes\delta_{t_n})=g\circ(\delta_{t_1}\otimes\cdots\otimes\delta_{t_n})$ for all $\seq t1n\in A$, then the leftmost decomposition factors of $\hat f$ and $\hat g$ must be of the same type. If one of $f,g$ contains only one decomposition factor, then so does the other. In that case the decompositions are identical by Theorem \ref{thm:boxdot1}. Otherwise, if $\hat f = f_1\cmp f'$ and $\hat g = g_1\cmp g'$, then $\Otimes f'$ and $\Otimes g'$ also have the same type. By Lemma \ref{lem:fg}, $f_1=g_1$ and $\Otimes f'=\Otimes g'$, therefore $f'$ and $g'$ also satisfy (iii). By the induction hypothesis, $f'=g'$, therefore $\hat f=\hat g$.
\exqed
\end{proof}

It follows that closure under composition and tensor product do not add any new morphisms to $\boxdot_1$ than already existed from closure of $\delta_t$, $\sigma$, $L$, and $R$ under $\circ$ and $\med$.

\section{Higher Identity Types as Cells and Iterated Path Spaces}
\label{apx:ident}

In traditional HoTT, equations at a type $A$ are witnessed by \emph{paths}, where $A$ is a topological space and a \emph{path} is a continuous functions $I\to A$. The higher identity types over a base type $X$ are generated inductively as spaces of paths:
\begin{itemize}
\item Level 0: points in $X$,
\item Level 1: $\Id Xab = \{\text{paths $p:I\to X$ such that $p(0)=a$ and $p(1)=b$}\}$,
\item Level 2: $\Id{\Id Xab}pq = \{\text{paths $\alpha:I\to\Id Xab$ such that $\alpha(0)=p$ and $\alpha(1)=q$}\}$,
\item Level 3: \parbox[t]{10cm}{$\Id{\Id{\Id Xab}pq}\alpha\beta$\\
$= \{\text{paths $\phi:I\to\Id{\Id Xab}pq$ such that $\phi(0)=\alpha$ and $\phi(1)=\beta$}\}$,}
\end{itemize}
etc. These ``path spaces'' are themselves topological spaces when equipped with the \emph{compact-open topology}, defined as follows. For $X$ and $K$ topological spaces with $K$ compact and Hausdorff ($I$ in our exanple), let $X^K$ denote the set of continuous functions $K\to X$. A subbasic open set of the compact-open topology on $X^K$ is a set of the form $U^C = \set{f\in X^K}{C\subs f^{-1}(U)}$ for $C\subs K$ closed (thus compact) and $U\subs X$ open.

Although the category $\Top$ of topological spaces and continuous functions is famously not cartesian closed, it is cartesian (has finite products), and if $K$ is compact and Hausdorff, then the functor $(-)^K:\Top\to\Top$ is right adjoint to $(-)\times K$. For an object $A$ of $\Top$, the space $A^K$ with the compact-open topology is the internal hom corresponding to the hom-set $K\to A$. The action of $(-)^K$ on morphisms $g:A\to B$ post-composes with $g$; that is, $g^K = g\circ- : A^K\to B^K$. One must verify that $g^K$ is continuous with respect to the compact-open topology on $A^K$ and $B^K$, but this is easy to check: for a subbasic open set $U^C$ of $B^K$ where $C\subs K$ is compact and $U\subs B$ is open,
\begin{align*}
(g^K)^{-1}(U^C)
&= (g\circ-)^{-1}(\set f{C\subs f^{-1}(U)})
= \set h{g\circ h \in \set f{C\subs f^{-1}(U)})}\\
&= \set h{C\subs (g\circ h)^{-1}(U)}
= \set h{C\subs h^{-1}(g^{-1}(U))}
= g^{-1}(U)^C,
\end{align*}
a subbasic open set of $A^K$.

That $(-)\times K$ is left adjoint to $(-)^K$ is a consequence of the following lemma.
\begin{lemma}
\label{lem:compactopen}
Let $A,B$, and $K$ be topological spaces with $K$ compact and Hausdorff.
Let $h:A\times K\to B$ be a function and let $h':A\to K\to B$ be its curried form $h'(a)(x)=h(a,x)$. The following are equivalent:
\begin{enumerate}[{\upshape(i)}]
\item
$h$ is continuous with respect to the product topology on $A\times K$;
\item
for all $a\in A$, $h'(a):K\to B$ is continuous (that is, $h':A\to B^K$) and $h'$ is continuous with respect to the compact-open topology on $B^K$.
\end{enumerate}
\end{lemma}
\begin{proof}
(i) $\Imp$ (ii). Suppose $h:A\times K\to B$ is continuous. We first show that for all $a\in A$, $h'(a):K\to B$ is continuous. For any open set $U\subs B$, by elementary set theory, $h'(a)^{-1}(U) = \set{x\in K}{(a,x)\in h^{-1}(U)}$. This is a section of an open set in a product space, therefore is open.

To show $h':A\to B^K$ is continuous, let $U^C$ be a subbasic open set in $B^K$.
By elementary set theory, $a\in{h'}^{-1}(U^C)$ iff $\{a\}\times C\subs h^{-1}(U)$. Since $h^{-1}(U)$ is open in $A\times K$, for all $x\in C$ there exist basic open sets $U_x\times V_x\subs h^{-1}(U)$ such that $a\in U_x$ and $x\in V_x$. By compactness, there is a finite set $F\subs C$ such that $C\subs\bigcup_{x\in F} V_x$ and $a\in\bigcap_{y\in F} U_x$, thus
\begin{align*}
\{a\}\times C \subs \bigcap_{y\in F} U_y\times\bigcup_{x\in F} V_x = \bigcup_{x\in F}\bigcap_{y\in F} U_y\times V_x \subs h^{-1}(U).
\end{align*}
Again by elementary set theory, $a\in\bigcap_{y\in F} U_y\subs {h'}^{-1}(U^C)$.

(ii) $\Imp$ (i). Suppose $h':A\to B^K$ is continuous. Let $\eps_B:B^K\times K\to B$ be the application function $\eps_B(f,x)=f(x)$. We first show that $\eps_B$ is continuous. For any $U\subs B$ open, $\eps_B^{-1}(U) = \set{(f,x)}{f(x)\in U}$. If $(f,x)\in \eps_B^{-1}(U)$, then $f:K\to B$ is continuous, so $f^{-1}(U)$ is open and $x\in f^{-1}(U)$. Since $K$ is regular, there exists an open set $V$ with closure $C$ such that $x\in V\subs C\subs f^{-1}(U)$. Then $f\in U^C$ since $C\subs f^{-1}(U)$, so $(f,x)\in U^C\times V\subs U^C\times C\subs \eps_B^{-1}(U)$, and $U^C\times V$ is open, therefore $\eps_B$ is continuous.

Now for all $(a,x)\in A\times K$,
\begin{align*}
(\eps_B\circ(h'\times 1_K))(a,x) &= \eps_B(h'(a),x)
= h'(a)(x)
= h(a,x),
\end{align*}
so $h = \eps_B\circ(h'\times 1_K)$. This is a composition of continuous functions, so $h$ is continuous.
\exqed
\end{proof}

The adjunction says that there is a bijection $\curry^K_{AB}$, natural in $A$ and $B$, that takes a continuous map $h:A\times K\to B$ to its curried form $\curry^K_{AB}(h):A\to B^K$. The inverse of $\curry^K_{AB}$ is denoted $\uncurry^K_{AB}$. The unit and counit are $\eta^K_A(a)(x) = (a,x)$ and $\eps^K_B(f,x) = f(x)$, respectively.

But one can say more. In case $A$ is also compact and Hausdorff, say $A=K_1$ and $K=K_2$, then $K_1\times K_2$ is compact and Hausdorff, therefore $K_1\times K_2\to B$ and $K_1\to B^{K_2}$ with the compact-open topology become $B^{K_1\times K_2}$ and $(B^{K_2})^{K_1}$, and we might ask whether the bijection $\curry^{K_2}_{K_1B}$ is a homeomorphism $B^{K_1\times K_2}\to (B^{K_2})^{K_1}$. It turns out that it is:

\begin{lemma}
\label{lem:curry}
Let $K_1,K_2$ be compact and Hausdorff and $B$ an arbitrary space. Then $\curry^{K_2}_{K_1B}:B^{K_1\times K_2}\to (B^{K_2})^{K_1}$ is a homeomorphism.
\end{lemma}
\begin{proof}
We already know that $\curry^{K_2}_{K_1B}$ is bijective on the hom-sets with inverse $\uncurry^{K_2}_{K_1B}$, so we need only show that both maps are continuous with respect to the compact-open topology.
This follows from the observation that
\begin{align*}
\curry^{K_2}_{K_1B} &= \curry^{K_1}_{B^{K_1\times K_2}B^{K_2}}(\curry^{K_2}_{B^{K_1\times K_2}\times K_1,B}(\uncurry^{K_1\times K_2}_{B^{K_1\times K_2}B}(1_{B^{K_1\times K_2}})))\\
&\quad : B^{K_1\times K_2}\to(B^{K_2})^{K_1}\\
\uncurry^{K_2}_{K_1B} &= \curry^{K_1\times K_2}_{(B^{K_2})^{K_1}B}(\uncurry^{K_2}_{(B^{K_2})^{K_1}\times K_1B}(\uncurry^{K_1}_{(B^{K_2})^{K_1}B^{K_2}}(1_{(B^{K_2})^{K_1}})))\\
&\quad : (B^{K_2})^{K_1}\to B^{K_1\times K_2}
\end{align*}
and the fact that $\curry$ and $\uncurry$ map continuous functions to continuous functions.
\exqed
\end{proof}

The adjunction leads to a homeomorphism between the iterated path and cubical views of higher identity types. 
Given any topological space $B$ and $a,b\in B$, we have defined $\Id Bab$ to be the space of continuous paths $p:I\to B$ with $p(0)=a$ and $p(1)=b$ with the compact-open topology. The family of all HoTT path spaces representing higher identity types relative to an ambient space $A$ is defined inductively. The identity types at level $n$ will be denoted $P_n$.
\begin{itemize}
\item
$P_0 = \{A\}$.
\item
If $B\in P_n$ and $\alpha,\beta\in B$, then $\Id B\alpha\beta\in P_{n+1}$.
\end{itemize}
Note that the elements of $P_n$ are pairwise disjoint: For $\Id Bab$ and $\Id{B'}{a'}{b'}$, if $B\ne B'$ then $B\cap B'=\emptyset$ inductively, therefore $\Id Bab\cap\Id{B'}{a'}{b'}=\emptyset$, and if $B=B'$ but $(a,b)\ne(a',b')$, then $\Id Bab\cap\Id B{a'}{b'}=\emptyset$ as well.
 
Each element of $P_n$ is a subspace of a larger space $V_n$, also defined inductively:
\begin{itemize}
\item
$V_0=A$
\item
$V_{n+1} = \{\text{continuous paths $I\to V_n$}\}$ with the compact-open topology.
\end{itemize}
Thus $V_n = \underbrace{I \to I \to \cdots \to I}_n \to A$. A continuous path in the space $V_n$ is a point in the space $V_{n+1}$.

Now we claim that each element of $P_n$ is a subspace of $V_n$. This is shown by induction. The only space in $P_0$ is $A = V_0$, and $A$ is a subspace of itself. An element of $P_{n+1}$ is a space of the form $\Id Bpq$ consisting of certain paths $\alpha:I\to B$, where $B\in P_n$, and these $\alpha$ are continuous in the topology on $B$. By the induction hypothesis, $B$ is a subspace of $V_n$, therefore the open sets of $B$ are exactly the open sets of $V_n$ intersected with $B$. It follows that $\alpha:I\to V_n$ and is also continuous in the topology on $V_n$, thus $\alpha\in V_{n+1}$.

We have shown that $\Id Bpq\subs V_{n+1}$, but to show that it is a subspace, we need to argue that the open sets of $\Id Bpq$ are exactly the open sets of $V_{n+1}$ intersected with $\Id Bpq$. By naive set theory, it suffices to show this for subbasic open sets of $\Id Bpq$ and $V_{n+1}$. A subbasic open set of $\Id Bpq$ is of the form $U^C = \set{f\in\Id Bpq}{C\subs f^{-1}(U)}$, where $U\subs B$ is open and $C\subs I$ is compact. Since $B$ is a subspace of $V_n$, $U = B\cap W$ for some open $W\subs V_n$. Since elements of $\Id Bpq$ only take values in $B$,
\begin{align*}
W^C\cap\Id Bpq &= \set{f\in\Id Bpq}{C\subs f^{-1}(W)}\\
&= \set{f\in\Id Bpq}{C\subs f^{-1}(W\cap B)}\\
&= \set{f\in\Id Bpq}{C\subs f^{-1}(U)} = U^C.
\end{align*}
Thus the subbasic open sets of $\Id Bpq$ are exactly the subbasic open set of $V_{n+1}$ intersected with $\Id Bpq$.

Let $C_n$ be the set of $n$-cells in our sense over the ambient space $A$. The formal definition is
\begin{itemize}
\item
$C_0 = A$
\item
$C_1$ is the set of continuous functions $I\to A$
\item
For $n\ge 2$, $C_n$ consists of all continuous functions $\alpha:I^n \to A$ such that
\begin{itemize}
\item
for all all $t\in I$, $\alpha\s1t\in C_{n-1}$; and
\item
$\alpha\s20 = \id\src^2\alpha$ and $\alpha\s21 = \id\tgt^2\alpha$.
\end{itemize}
\end{itemize}
Let $H_0 = A$ and for $n\ge 1$, let $H_n = A^{I^n}$, the space of all continuous functions $I^n\to A$ with the compact-open topology. Each $C_n$ is a subset of $H_n$, and we topologize it with the relative topology.

\begin{lemma}
For all $n\ge 0$, $H_n$ and $V_n$ are homeomorphic.
\end{lemma}
\begin{proof}
We define homeomorphisms $h_n:H_n\to V_n$ at each level.
\begin{itemize}
\item
$h_0:H_0\to V_0$ is the identity function on $A$.
\item
For $n\ge 1$, $h_n = h_{n-1}^{I}\circ\curry^{I^{n-1}}_{I,A} = \lambda\alpha.h_{n-1}\circ\alpha^{(1)}:H_n \to V_n$.
\end{itemize}
Let us check the inductive case. For $n\ge 1$, any function $\alpha\in H_n$ is of type $I^n\to A$, which is the same as $I\times I^{n-1}\to A$. Currying on the first argument, we get $\curry^{I^{n-1}}_{I,A}(\alpha) = \alpha^{(1)}:I \to I^{n-1}\to A$. By Lemma \ref{lem:curry}, $\curry$ is a homeomorphism
\begin{align*}
\curry^{I^{n-1}}_{I,A}:A^{I^n}\to(A^{I^{n-1}})^{I} = H_n\to H_{n-1}^{I}.
\end{align*}
By the induction hypothesis, $h_{n-1}:H_{n-1}\to V_{n-1}$ is a homeomorphism, therefore so is $h_{n-1}^{I}:H_{n-1}^{I}\to V_{n-1}^{I}$. But $V_{n-1}^{I}$ is just $V_n$. Composing these, we get a homeomorphism
\begin{align*}
h_{n-1}^{I}\circ\curry^{I^{n-1}}_{I,A}:H_n\to V_n.
\extagqed
\end{align*}
\end{proof}

Let us call $\alpha,\beta\in C_n$ \emph{boundary equivalent} and write $\alpha \bdry n \beta$ if $\alpha\s ib = \beta\s ib$ for all $1\le i\le n$ and $b\in\{0,1\}$. Let $[\alpha]$ denote the $\bdry n$-equivalence class of $\alpha$. Let $\Pinh_n$ denote the set of inhabited identity types of $P_n$.

\begin{theorem}
The equivalence classes of $\sim_n$ correspond bijectively to $\Pinh_n$. 
Concretely, there is a bijection $g_n : C_n/{\bdry n} \to \Pinh_n$ such that for every $\alpha\in C_n$, the map $h_n\rest[\alpha] : [\alpha] \;\longrightarrow\; g_n([\alpha])$ is a homeomorphism. In particular, the image of $C_n$ under $h_n$ is $\bigcup P_n$, and $h_n$ maps each boundary equivalence class homeomorphically to its corresponding identity type.
\end{theorem}
\begin{proof}
The maps $g_n$ are defined inductively. For $n=0$, all elements of $A$ are vacuously equivalent under $\bdry 0$, so $A$ is the unique equivalence class of $\bdry 0$, and there is only one identity type $A$ at level 0, thus we define $g_0(A)=A$.

Now given $\alpha\in C_n$, we have by results of our paper that $\alpha\s k0 = \id^{k-1}\src^k\alpha$ and 
$\alpha\s k1 = \id^{k-1}\tgt^k\alpha$. For $k\ge 2$, these are identities. Thus for all $1\le i\le n-1$ and $t\in I$,
\begin{align*}
\alpha\s 1t\s i0
&= \alpha\s{i+1}0\s 1t
= \alpha\s{i+1}0\s 10\\
\alpha\s 1t\s i1
&= \alpha\s{i+1}1\s 1t
= \alpha\s{i+1}1\s 10.
\end{align*}
As these are independent of $t$, we have for all $s,t\in I$, $1\le i\le n-1$, and $b\in\{0,1\}$,
\begin{align*}
\alpha\s 1s\s ib
= \alpha\s 1t\s ib,
\end{align*}
therefore $\alpha\s 1s \bdry{n-1} = \alpha\s 1t$. In particular, for all $t\in I$,
\begin{align}
\alpha\s 1t \bdry{n-1} \src\alpha.
\label{eq:slam}
\end{align}
Moreover, if $\alpha \bdry n\beta$, then
\begin{align*}
\src\alpha &= \alpha\s10 = \beta\s10 = \src\beta &
\tgt\alpha &= \alpha\s11 = \beta\s11 = \tgt\beta.
\end{align*}
We therefore define 
\begin{align*}
g_n([\alpha]) = \Id{g_{n-1}([\src\alpha])}{h_{n-1}(\src\alpha)}{h_{n-1}(\tgt\alpha)},
\end{align*}
and this map is well defined.

Now we show by induction that $h_n$ maps $[\alpha]$ as a subspace of $H_n$ homeomorphically to $g_n([\alpha])$ as a subspace of $V_n$. For $n=0$, for all $a\in A$, $h_0(a) = a \in A = g_0(A)$. For $n\ge 1$, by definition, $h_n(\alpha) = h_{n-1}\circ\alpha^{(1)}$, thus for all $t\in I$,
\begin{align*}
h_n(\alpha)(t)
&= h_{n-1}(\alpha^{(1)}(t))\\
&\in g_{n-1}([\alpha^{(1)}(t)]) && \text{inductive hypothesis}\\
&= g_{n-1}([\src\alpha]) && \text{by \eqref{eq:slam}.}
\end{align*}
Also,
\begin{align*}
h_n(\alpha)(0) &= h_{n-1}(\alpha^{(1)}(0)) = h_{n-1}(\src\alpha) & 
h_n(\alpha)(1) &= h_{n-1}(\alpha^{(1)}(1)) = h_{n-1}(\tgt\alpha).
\end{align*}
Combining these facts, we have
\begin{align}
h_n(\alpha) \in \Id{g_{n-1}([\src\alpha])}{h_{n-1}(\src\alpha)}{h_{n-1}(\tgt\alpha)} = g_n([\alpha]).
\label{eq:slam2}
\end{align}

It remains to show that $g_n:C_n/{\bdry n}\to\Pinh_n$ and each $h_n\rest[\alpha]:[\alpha]\to g_n([\alpha])$ are surjective. We need only show that the inverse image $\alpha = h_n^{-1}(p)$ of any element $p$ of a HoTT identity type $B$ in $\Pinh_n$ is a cell in $C_n$, as then $h_n(\alpha) = p\in B$ and $g_n([\alpha]) = B$ by \eqref{eq:slam2}. Since $B$ and $p$ were arbitrary, $h_n\rest[\alpha]$ and $g_n$ are surjective.

We proceed by induction. The claim is immediate for $n=0$. Now for any inhabited type $\Id Bqr$ at level $n\ge 1$ and any $p\in\Id Bqr$, let $\alpha = h_n^{-1}(p)\in H_n$. To show that $\alpha\in C_n$, by our definition, we need to show two things:
\begin{enumerate}[(i)]
\item for all $t\in I$, $\alpha\s1t\in C_{n-1}$,
\item if $n\ge 2$, then for all $t\in I$ and $b\in\{0,1\}$, $\alpha\s2b^{(1)}$ is an identity, that is, the value of $\alpha\s2b\s1t$ is independent of $t$.
\end{enumerate}

For (i), since $h_n(\alpha) = h_{n-1}\circ\alpha^{(1)}$, we have
\begin{align*}
h_{n-1}(\alpha\s1t) = h_n(\alpha)(t) = p(t)\in B.
\end{align*}
By the inductive hypothesis, $\alpha\s1t\in C_{n-1}$.

For (ii), since $h_{n-1}(\alpha\s1t) = h_{n-2}\circ\alpha\s1t^{(1)}$, we have
\begin{align*}
h_{n-2}(\alpha\s20\s1t) &= h_{n-2}(\alpha\s1t\s10)\\
&= h_{n-1}(\alpha\s1t)(0) = h_n(\alpha)(t)(0) = p(t)(0) = q\\
h_{n-2}(\alpha\s21\s1t) &= h_{n-2}(\alpha\s1t\s11)\\
&= h_{n-1}(\alpha\s1t)(1) = h_n(\alpha)(t)(1) = p(t)(1) = r,
\end{align*}
thus $\alpha\s20\s1t = h_{n-2}^{-1}(q)$ and $\alpha\s21\s1t = h_{n-2}^{-1}(r)$, which are independent of $t$.

Since $h_n$ is injective and the classes $[\alpha]$ partition $C_n$, the images $g_n([\alpha])$ are pairwise distinct. Hence $g_n : C_n/{\bdry n} \to \Pinh_n$ is bijective. Because $h_n$ is a homeomorphism, each restriction $h_n\rest[\alpha] : [\alpha] \to g_n([\alpha])$
is a homeomorphism.
\exqed
\end{proof}
\end{toappendix}

\appendix
\end{document}